\documentclass[aps,10pt,onecolumn,amsmath,amssymb,floatfix,notitlepage,a4paper,prd,nofootinbib,superscriptaddress]{revtex4-2}
\usepackage{amssymb,amsthm,amsfonts,amstext,amsmath}
\usepackage[english]{babel}
\usepackage{graphicx}
\usepackage{float}
\usepackage{xcolor}
\usepackage{array}
\usepackage{tcolorbox}
\usepackage{enumitem}
\usepackage{subcaption}
\usepackage{aliascnt}
\usepackage{hyperref}
\usepackage{cleveref}
\usepackage{orcidlink}

\crefname{appsec}{appendix}{appendices}
\Crefname{appsec}{Appendix}{Appendices}
\crefname{table}{problem}{problems}
\Crefname{table}{Problem}{Problems}

\def\be{\begin{equation}}
\def\ee{\end{equation}}
\def\bea{\begin{eqnarray}}
\def\eea{\end{eqnarray}}

\def\bma{\begin{mathletters}}
\def\ema{\end{mathletters}}
\def\P{{\cal P}}
\def\A{{\cal A}}

\def\q0{\underline{0}}

\def\H{{\cal H}}
\def\T{{\cal T}}
\def\P{{\cal P}}

\def\S{{\cal S}}
\def\C{{\mathbb C}}
\def\id{{\mathbb I}}

\def\EE{\mathbb{E}}
\def\OO{{\cal O}}
\def\W{{\cal W}}
\def\V{\mathbb{V}}

\def\M{{\cal M}}
\def\H{{\cal H}}
\def\AA{I}
\def\U{{\cal U}}
\def\K{{\cal K}}
\def\J{{\cal J}}
\def\B{{\cal B}}
\def\R{\mathbb{R}}
\def\RR{{\cal R}}
\def\Z{\mathbb{Z}}
\def\CC{{\cal C}}

\def\N{\mathbb{N}}
\def\NN{{\cal N}}

\def\tr{\mbox{tr}}
\def\W{{\cal W}}
\def\Q{{\cal Q}}

\def\one{\leavevmode\hbox{\small1\normalsize\kern-.33em1}}

\def\norm#1{\lVert {#1} \rVert}
\def\supp{\operatorname{supp}}
\def\chull{\operatorname{ch}}

\def\slim#1{\underset{#1}{\operatorname{s-lim}}\,}
\def\sstarlim#1{\underset{#1}{\operatorname{s*-lim}}\,}

\def\Al{\hat{\cal A}}

\def\Int#1{\mathrm{int}(#1)} 

\def\EE#1#2{\langle #1, E#2 \rangle} 

\def\bra#1{\langle#1|} \def\ket#1{|#1\rangle}
\def\braket#1#2{\langle#1|#2\rangle}

\def\proj#1{\ket{#1}\!\bra{#1}}

\def\id{{\mathbb I}}

\usepackage{soul}

\newtheorem{theorem}{Theorem}

\newtheorem{defin}[theorem]{Definition}

\newtheorem{remark}[theorem]{Remark}

\newtheorem{lemma}[theorem]{Lemma}

\newtheorem{prop}[theorem]{Proposition}

\newtheorem{corollary}[theorem]{Corollary}

\begin{document}
\title{Causality and realizability of local operations in quantum field theory}

\author{\orcidlink{0000-0003-4827-1806} Jan Mandrysch}
\affiliation{Institute for Quantum Optics and Quantum Information (IQOQI), Austrian Academy of Sciences, Boltzmanngasse 3, A-1090 Vienna, Austria}
\author{\orcidlink{0009-0003-5570-3512} Robin Simmons}
\affiliation{University of Vienna, Faculty of Physics, Vienna Doctoral School in Physics, and Vienna Center for Quantum Science and Technology (VCQ), Boltzmanngasse 5, A-1090 Vienna, Austria}
\affiliation{Institute for Quantum Optics and Quantum Information (IQOQI), Austrian Academy of Sciences, Boltzmanngasse 3, A-1090 Vienna, Austria}
\author{\orcidlink{0000-0003-0717-3927} Miguel Navascu\'es}
\affiliation{Institute for Quantum Optics and Quantum Information (IQOQI), Austrian Academy of Sciences, Boltzmanngasse 3, A-1090 Vienna, Austria}

\begin{abstract}
As noted by Sorkin, regarding quantum instruments whose Kraus operators are localizable within some spacetime region as operations accessible therein leads to superluminal communication. This so-called Sorkin paradox can be resolved by further constraining the set of allowed local operations in quantum field theory (QFT). In this spirit, Fewster and Verch proposed a framework for local QFT operations that generalizes non-relativistic quantum measurement theory and does not lead to Sorkin-like paradoxes. Shortly afterwards, Jubb (and later Oeckl) identified the minimal conditions that QFT instruments must satisfy to be compatible with Einstein's causality. In this work, we study both approaches in the quantum field theory of the free scalar field. First, we prove that a very wide class of causal instruments is FV-realizable: namely, those whose measurement channels are random displacements of the field operators. As we show, this class allows implementing arbitrary instruments in a heralded, probabilistic way, as well as non-demolition measurements deterministically. Second, we construct examples of causal channels that do not admit an approximate FV realization. Some of such causal, not FV-realizable channels violate basic physical principles, so they should not be part of any measurement theory for QFT. Third, we investigate the difficulty of characterizing the set of QFT channels that can be generated through the composition of several FV schemes. In this regard, we find a countable family of simple QFT channels for which no Turing machine can discriminate between channels within or far away from the implementable set.\end{abstract}

\maketitle

\tableofcontents

\section{Introduction}
 Algebraic quantum field theory (AQFT) seeks to provide a general framework to describe quantum fields within a curved spacetime, by positing a correspondence between regions of spacetime and local operator algebras. Traditionally, causality has been enforced by demanding that local algebras corresponding to causally disconnected spacetime regions commute. This axiom, microcausality, has indeed the effect of guaranteeing Einstein's causality in experiments involving two parties conducting operations on two causally disconnected spacetime regions \cite{HK70a,HK70b}. However, as pointed out by Sorkin, microcausality alone is not enough to prevent violations of causality in experiments where operations defined on three or more spacetime regions are at play \cite{Sorkin1993-SORIMO}. To construct his acausal example, which nowadays we refer to as Sorkin's paradox, Sorkin implicitly assumed that any quantum instrument with Kraus operators belonging to a local algebra was \emph{realizable} within its associated spacetime region. As shown in \cite{MuchVerch23}, similar assumptions allow one to produce signalling effects in classical field theories as well. In both scenarios, the take-home message is that the set of realizable local operations must be constrained beyond microcausality to ensure consistency with general relativity.

In this direction, Fewster and Verch proposed in 2020 a covariant measurement model for quantum fields \cite{local_meas_QFT}. In their framework --the FV framework--, a feasible local operation (or quantum instrument) is specified through a local interaction between the target QFT and a probe QFT, followed by a measurement of the latter. Remarkably, communication networks that integrate quantum instruments admitting an FV realization avoid superluminal signalling \cite{impossible_meas}. 

Building on \cite{local_meas_QFT, impossible_meas}, Jubb noticed that demanding instruments to fulfill the so-called \emph{past support non-increasing (PSNI)} property is sufficient to preclude the occurrence of Sorkin-like effects \cite{Jubb_2022}. Furthermore, for scalar quantum fields compliance with the PSNI property is \emph{necessary} to avoid violations of causality. Jubb also noted that unitaries generated by the product of two smeared fields are, in general, acausal; contrarily, displacements and noisy measurements of the smeared quantum field are causal (i.e., they satisfy PSNI). This last point is also corroborated in the more recent paper \cite{Oeckl_2026}, where Oeckl calls the PSNI condition \emph{causal transparency}. 

The PSNI property provides a straightforward criterion to certify the non-causal nature of a given QFT instrument. In contrast, certifying that a given instrument is not FV-realizable is a non-trivial task: it amounts to proving that there do not exist a probe, a target-probe interaction and a probe measurement that generate the desired effective instrument on the target. For this reason, a characterization of the set of instruments available through FV schemes has so far been elusive.

Nonetheless, there has been some progress. In \cite{Fewster2023} it is shown that FV-realizable measurements are tomographically complete, i.e., they suffice to single out the underlying local QFT state. The more recent work of \cite{fewster2025measurementpreparationprotocolsquantum} goes beyond: it shows that, for Klein-Gordon scalar fields, any Positive Operator Valued Measure (POVM) localizable inside a bounded spacetime region admits an FV realization within a slightly larger region \cite{fewster2025measurementpreparationprotocolsquantum}. The interaction proposed in this work swaps the local algebras of target and probe fields in the region where the POVM elements are localized. A subsequent measurement of the probe with said POVM thus induces on the target the desired POVM. It has to be noted, though, that the protocol effectively defines a demolition measurement, since the target's local state prior to the measurement is fully destroyed thereafter. In this regard, a prior paper by two of us had established that, at least, nondemolition measurements of the smeared field are asymptotically FV realizable \cite{field_meas}. 

The above works leave several open problems. One of them is whether FV schemes allow realizing arbitrary POVMs in a non-demolition way. Another question is whether all causal instruments admit an FV realization. Since the FV framework is perceived to be ``more physical'' than just demanding causality, any causal, non-FV instrument would be viewed with suspicion. Finally, we lack a systematic method or algorithm to tell whether a given quantum instrument admits or not an FV realization. Does such a method even exist?

In this work, we solve these three problems. Namely, we prove that any local POVM admits an asymptotic FV realization which is non-demolition, that there exist causal QFT channels that cannot be approximated through an FV scheme and that determining which channels can is an undecidable problem.

The key to all these results is a family of QFT operations that we introduce here: the class of \emph{Weyl instruments}, i.e., instruments whose Kraus operators are elements of the Weyl $\star$-algebra. Each such Weyl instrument is specified through a finite number of parameters; the problem of determining if a given Weyl instrument is causal or admits an FV realization can thus be formulated as a computational problem with a finite-sized input.

With regards to causality, we characterize the set of causal Weyl instruments and show that it is closed under composition within a communication network. This legitimizes the intuition, formulated in \cite{Jubb_2022}, that the outcome of a causal instrument with interaction zone $\OO$ is accessible to any observer acting in the total future\footnote{The total future of $\OO$ is the set of spacetime points which lie in the common causal future over all points in $\OO$.} of $\OO$. We also single out a particular class of causal Weyl instruments: the \emph{diagonal} Weyl instruments, whose measurement channels are convex combinations of field displacements. Remarkably, among the set of Weyl instruments defined over a finite-dimensional, `generic' subspace $\Q$ of smeared fields, all causal Weyl instruments are diagonal. Moreover, for general $\Q$, any Positive Operator Valued Measure (POVM) whose operators belong to the von Neumann algebra generated by  $e^{i\Q}$ is asymptotically implementable through a sequence of diagonal Weyl instruments that act trivially on any smeared field commuting with $\Q$.

We next turn to FV realizability. We show that all diagonal instruments admit an asymptotic FV realization in Klein-Gordon fields, and so any POVM can be asymptotically FV-realized in a non-demolition way. 

By contrast, we also produce concrete examples of causal Weyl channels that cannot be reproduced through asymptotic FV schemes, or compositions thereof. We find such counterexamples by exploiting the connection, first pointed out in \cite{Gisin_2024}, between feasible QFT operations and \emph{entanglement-assisted} bipartite quantum operations. The latter are known to produce so-called \emph{quantum} correlations and were proven to be a strict subset of the set of \emph{non-signalling} bipartite instruments by showing that the non-signalling set allows generating supra-quantum correlations \cite{causal_vs_localizable}. Following this idea, we construct a causal Weyl instrument that, if realized, would allow two separate parties to produce correlations impossible to approximate through bipartite quantum systems. On the contrary, we show that FV-realizable channels can only be used to generate quantum correlations. Thus, our channel cannot be approximated through FV schemes.

We next propose a model to integrate FV schemes within a communication network and consider the problem of deciding whether a Weyl channel is FV-network-implementable, or otherwise far away from being FV-network-implementable. Assuming that the split property holds for all QFTs involved, we show this problem to be uncomputable for a countable class of Weyl channels. We do so by exploiting once more the connection with quantum nonlocality: since the approximate characterization of quantum correlations is undecidable \cite{mipre}, so is the approximate characterization of FV-network-implementable Weyl channels.

Finally, we link the physical realizability of a QFT instrument with an old research program in quantum foundations: the derivation of the set of bipartite quantum correlations from physical principles \cite{Popescu1994-POPQNA,non_trivial}. As we show, the correlations generated in Bell experiments in a scalar QFT where all causal operations were realizable would only be limited by Einstein's causality. Since some such causal correlations violate a number of physical principles
\cite{no_advantage, inf_causality,mac_loc, local_orth}, said QFT would be unphysical. In other words: the set of causal instruments is too large to be the basis for a QFT measurement theory. This leaves the FV framework as the only standing prescription --for the time being-- to define local operations in QFT.

The structure of the paper is as follows: in \Cref{sec:prelim} we fix our notation and collect all ingredients from the literature needed to formulate and prove our results. In \Cref{sec:weylops}, given an abstract $\star$-Weyl algebra, we introduce the classes of Weyl and diagonal Weyl instruments and discuss how these approximate operator-POVMs and -instruments. Our main results are given in \Cref{sec:causality,sec:implementability}, where we assume the setting of a free scalar QFT on a, possibly curved, spacetime. In \Cref{sec:causality}, we characterize \emph{causal} Weyl instruments, prove that they are closed under composition and explain in which sense diagonal Weyl instruments are generic among the set of causal Weyl instruments. In \Cref{sec:implementability}, we prove that all diagonal Weyl instruments are FV-realizable, that there exist causal channels impossible to approximate by FV-realizable ones, and that telling whether one channel can be approximated by composing several FV schemes is an undecidable problem. In \Cref{sec:discussion}, we argue, in the light of the results in \Cref{sec:implementability}, that certain local causal channels must be unphysical in any reasonable QFT measurement theory. Finally, in \Cref{sec:conclusion} we present our conclusions.

\section{Preliminaries}\label{sec:prelim}
In this section, we briefly introduce basic notions of the AQFT framework, quantum measurement theory, causal QFT operations, the FV formalism, and Bell nonlocality. The first three topics are indispensable to understand the paper. The FV formalism is the subject of \Cref{sec:implementability}, while quantum nonlocality is only required to follow \Cref{sec:causal_not_FV,sec:undecidability}.

\subsection{Spacetime and local regions}\label{subsec:spacetime}
Let us first fix some basic notation. Given a topological space and a set $\RR$ of points thereof, we denote by $\Int{\RR}$ and $\overline{\RR}$ its interior and closure, respectively. A regular closed set of points is the closure of an open set.

Let $\M$ be a spacetime manifold\footnote{Technically, this takes $\M$ to be a oriented and time-oriented smooth paracompact nonempty manifold with finitely many connected components. If needed, its metric is denoted by $g_{\mu\nu}$ and assumed to have signature $(+,-,...,-)$.}. Given an arbitrary set ${\cal R}\subset\M$, we denote by $J^{+}({\cal R})$ ($J^{-}({\cal R})$) its causal future (past). Similarly, we define $N^{+}(\RR)$ ($N^{-}(\RR)$) as the set of points $x\in\M$ for which there exists a past-(future-)directed light-like curve from $x$ to some $y\in\RR$, but no time-like curve connecting $x$ and $y$\footnote{In \cite{Min19} $N^\pm(\RR)$ is denoted as $\mathcal{E}^\pm(\RR)$. Note also that in view of \cite[Thm.~2.22]{Min19} any curve as given in the definition must be an achronal light-like geodesic (up to reparametrization).}. We will also make use of $J({\cal R}):=J^+({\cal R})\cup J^-({\cal R})$ and $N({\cal R}):=N^+({\cal R})\cup N^-({\cal R})$. The in-region $\RR^-$, respectively, out-region $\RR^+$ with respect to $\RR$ is the set of points in $\M$ outside the causal future/past of $\RR$, i.e., $\RR^\pm := \M \setminus \overline{J^\mp(\RR)}$. Similarly, the causal complement of $\RR$ is given as $\RR^\perp = \RR^+ \cap \RR^- = \M \setminus \overline{J(\RR)}$. Its future (past) Cauchy development $D^+({\cal R})$ ($D^-({\cal R})$) is the set of points $x\in\M$ such that any past-inextendible (future-inextendible) causal curve starting in $x$ intersects $\RR$. We denote by $D({\cal R})=D^+({\cal R})\cup D^-({\cal R})$ the Cauchy development of ${\cal R}$. The total future (past) of $\RR$, denoted as $\RR^\vee$ ($\RR^\wedge$), is the set of all $x\in\M$ that are in the causal future (past) of every point in $\RR$. We say that two sets $\RR,\RR'\subset\M$ are causally disconnected (denoted as $\RR\perp \RR'$) when $\RR\subset(\RR')^\perp$. They are causally orderable if $J^+(\RR)\cap J^-(\RR')=\emptyset$ or $J^+(\RR')\cap J^-(\RR)=\emptyset$ holds: in the first case, we write $\RR\succ \RR'$; in the second, $\RR \prec \RR'$; and if both hold then $\RR \perp \RR'$ follows. We say that a set ${\cal R}\subset \M$ is causally convex if any causal curve starting in $x\in {\cal R}$ and ending in $y\in {\cal R}$ is fully contained in $\RR$. The causal hull $\chull({\cal R})$ of ${\cal R}$ is the smallest causally convex set containing ${\cal R}$.
If ${\cal R}$ is nonempty, open, causally convex and has finitely many connected components, we call it a \emph{region}. A region is called bounded if its closure is compact.

A Cauchy surface is any surface with the property that any time-like inextendible curve intersects it at exactly one point. In the following, we will assume the given spacetime $\M$ to be globally hyperbolic, i.e., it contains no closed time-like curves and the causal hull of any compact set is compact. Equivalently, the spacetime $\M$ admits a foliation in terms of Cauchy surfaces $(t,\Sigma_t)_{t \in \mathbb{R}}$. Note that our conventions imply that, if $\M$ is globally hyperbolic, then any region ${\cal R} \subset \M$ defines a globally hyperbolic spacetime in its own right. Further, if $\RR \subset \M$ is a region, its causal future/past, causal complement, future/past Cauchy development, the interior of its total future/past define regions as well. More explicitly: $J^\pm(\RR)$ and $J(\RR)$ are causally convex by definition. Then, if $\RR \subset \M$ is open, $J^\pm(\RR)$ and $J(\RR)$ are open (see e.g. \cite[Remark~2.1(a), Lemma~2.3]{BR07}). Next, for any subset $\RR \subset \M$, $\RR^\perp = \M \setminus \overline{J(\RR)}$ is open by definition. $\RR^\perp$ is causally convex, since any causal curve with two end points which intersects $\overline{J(\RR)}$ has also one of its endpoints in $\overline{J(\RR)}$. Thus, $\RR^\perp$ is also a region. Next, it is straightforward to show that $D^\pm(\RR)$ and $D(\RR)$ are regions, if $\RR \subset \M$ is a region (see Proposition \ref{prop:D_plus} in Appendix \ref{app:lemmas_causal}). The total future/past is an (possibly infinite) intersection of causal future/past regions, thus they are also causally convex and their interiors define regions, too.

\subsection{Algebraic QFT}\label{subsec:aqft}
An algebraic QFT (AQFT) on a globally hyperbolic spacetime $\M$ is defined in terms of a $\star$-algebra $\Al$ and a net of $\star$-subalgebras $\Al(\RR)$ labelled by bounded regions $\RR \subset \M$ which all share the same unit element $1$ and satisfy the following standard axioms \cite{AQFT_intr}: for all bounded regions $\RR, {\cal T} \subset \M$,
\begin{enumerate}
\item Isotony: $\Al(\RR)\subset \Al({\cal T})$, if $\RR\subset {\cal T}$.
\item Microcausality: $[\Al(\RR),\Al({\cal T})]=0$, for $\RR \perp {\cal T}$.
\item The timeslice property: $\Al(\RR)=\Al(D(\RR))$.
\end{enumerate}

The algebras $\Al(\RR)$ are referred to as \emph{local algebras} and its Hermitian elements are commonly considered as the observables measurable within the spacetime region $\RR$. For an unbounded region or an arbitrary subset $\OO \subset \M$, it is customary to define
\begin{equation}
    \Al(\OO) = \bigvee_{\text{region }\RR \subset \OO} \Al(\RR)
\end{equation}
which inherits conditions 1.--3. from above.

An AQFT has different classes of inequivalent representations (often called sectors). To fix one of those, we still need one more ingredient: a state. If $\Al$ is a C$^\star$-algebra, given any algebraic state, i.e., a positive linear normalized functional $\omega$ on $\Al$, through the GNS construction we obtain a separable Hilbert space $\H=\H_\omega$, a vector $\Omega=\Omega_\omega \in \H$ and, to every region $\RR \subset \M$, an associated operator algebra $\A(\RR) = \A_\omega(\RR) \subset B(\H)$ with three significant properties:
\begin{enumerate}
 \item $\A(\RR)\Omega \subset \H$ is dense
 \item $\A(\RR)$ is closed under the strong operator topology (SOT); this renders $\A(\RR)$ a von-Neumann algebra and is equivalent to $\A(\RR)=\A(\RR)''$. Denoting the limit with respect to SOT by $\slim{}$, closedness under SOT entails that for any (strongly) convergent sequence or net $(A_k)_k \subset \A(\RR)$,
\begin{equation}
    A:= \slim{k} A_k \in \A(\RR).
\end{equation}
\item If $\omega$ is pure, then $\A(\M) = B(\H)$, and the operator algebras $\{ \A(\RR) \}$ define an AQFT on $B(\H)$. 
\end{enumerate}
If $\Al$ is merely a $\star$-algebra, the associated operator algebras are generally unbounded. However, in favourable cases (see \Cref{subsec:freeqft} below) the setting of an AQFT on $B(\H)$ still arises. 

In \Cref{sec:undecidability} we will demand an AQFT on $B(\H)$ to satisfy two more axioms: the \emph{Reeh-Schlieder property} and \emph{the split property}. The first demands that, if $\RR$ is bounded, $A \in \A(\RR)$ with $A\Omega = 0$ implies $A=0$. The second demands that, for regions $\RR, \S$, with $\overline{\RR}\subset \S$, there exists a type-I factor\footnote{That is, an operator algebra isomorphic to $B(\H)$ with trivial center.} ${\cal C}$ satisfying
\begin{equation}
\A(\RR)\subset {\cal C}\subset \A(\S).
\label{split_prop}
\end{equation}
The split property has an important consequence for experiments involving two or more strictly space-like separated regions. Let $\RR,\T$ be two regions with $\overline{\RR}\perp\overline{\T}$, and let $\S$ be a region such that $\overline{\RR}\subset \S$, $\S\perp \T$. Since ${\cal C}$ is a type-I factor, there exists a unitary $U$ and Hilbert spaces $\H_1$,$\H_2$ such that $U\H=\H_1\otimes \H_2$, with $U{\cal C}U^*=B(\H_1)\otimes \id_2$, $U{\cal C}'U^*=\id_1\otimes B(\H_2)$, where ${\cal C}'$ denotes the commutant of ${\cal C}$. By isotony and Einstein's causality, it follows that
\begin{align}
U\A(\RR) U^*\subset B(\H_1)\otimes \id_2, U\A(\T) U^*\subset \id_1\otimes B(\H_2).
\end{align}
The split property thus implies that the local algebras of strictly space-like separated regions can be regarded to act on independent quantum systems.

Finally, the existence of AQFT states such that the associated operator algebras obey the Reeh-Schlieder and split property is known in arbitrary globally hyperbolic spacetimes (at least for free fields) and it includes the quasifree Hadamard states on ultrastatic spacetimes, e.g., ground and thermal states; see \cite{splitproperty} for further discussion and references. 

\subsection{Free fields}\label{subsec:freeqft}

We next define the quantum field theory of a free scalar field. Let us assume that the classical field satisfies an equation of motion of the form $K\phi(x)=0$, where $K$ is a differential operator which we assume to be Green-hyperbolic \cite{Greenhyperbolic}, e.g., the Klein-Gordon field $K = \square_g +m^2$. The corresponding AQFT is determined by specifying the associated net of local algebras. For this, we consider test functions given as elements of $C_0^\infty(\M)$, the space of compactly supported smooth functions on $\M$, and the retarded (advanced) Green's operator for $K$, denoted by $E^{+}$ ($E^{-}$). Given $f \in C_0^\infty(\M)$, the smooth functions $\phi^\pm=E^\pm f$ satisfy $K\phi^\pm=f$ and $\supp \phi^{\pm} \subset J^{\pm}(\supp f)$. We define $E:=E^--E^+$ and the induced pre-symplectic form on $C_0^\infty(\M)$ by $\EE{f}{g} := \int_\M f Eg$ using integration of test functions across spacetime. There are now two conventional schemes to form the $\star$-algebra describing our quantum fields. The (polynomial) field algebra $\hat{\cal P}$ is the unital $\star$-algebra with generators $\{\Phi(f):f\in C_0^\infty(\M)\}$, where we assume that, for all $f,g \in C_0^\infty(\M)$ and $\alpha \in \mathbb{R}$, the map $\Phi:C_0^\infty(\M)\to\hat{{\cal P}}$ satisfies the following properties:
\begin{align}
&\Phi(\alpha f+g)=\alpha \Phi(f)+\Phi(g),\nonumber\\
&\Phi(f)^*=\Phi(f^\ast),\nonumber\\
&\Phi(Kf)=0,\nonumber\\
&[\Phi(f),\Phi(g)]=i\EE{f}{g}.
\label{eqs_field}
\end{align}

On the other hand, the Weyl field algebra $\hat{\cal W}$ is the unital $^\star$-algebra with unitary generators $\{ e^{i\Phi(f)} : f \in C_0^\infty(\M)\}$ satisfying the Weyl relation
\begin{equation}
e^{i\Phi(f)}e^{i\Phi(g)}=e^{-\frac{i}{2}\EE{f}{g}}e^{i\Phi(f+g)}.
\label{Weyl_rels}
\end{equation}
Our notation suggests $e^{i\Phi(f)}$ to be formed as the imaginary exponential of the smeared field $\Phi(f)$, which strictly speaking only applies when $\Phi(f)$ is given as a self-adjoint Hilbert space operator satisfying eqs. (\ref{eqs_field}) in which case the Weyl relation \eqref{Weyl_rels} can be viewed as a consequence of the Baker-Campbell-Hausdorff identity. 

In both settings, polynomial $\Al= \hat{\cal P}$ or Weyl $\Al = \hat{\cal W}$, for an arbitrary subset $\OO \subset \M$, the algebras $\Al(\OO)$ are defined as $\star$-subalgebras to $\Al$ by restricting to generators labeled by test functions with support contained in $\OO$. With this definition, the axioms given in the preceding section are known to hold (see \cite[Section 4, Appendix C]{local_meas_QFT} for the polynomial case) and we can add a localization property known as \emph{Haag property}: namely, if $A \in \Al$ commutes with all elements of $\Al({\cal T}^\perp)$ for some region ${\cal T} \subset \M$, then $A \in \Al(\RR)$ where $\RR$ is any region\footnote{In non-scalar theories, one might need to restrict this property to apply only to connected $\RR$.} containing $\overline{{\cal T}}$. In some cases, explicitly stated, we will require our operator representations $\pi$ to have the following continuity property: given a sequence (or net) of test functions $(g_\lambda)_\lambda \subset C_0^\infty(\M)$ with $\lim_{\lambda \to \infty} g_\lambda = g$ (in test function topology), it holds that
\begin{equation}\label{eq:weylcontinuity}
    \slim{\lambda \to \infty} \pi(e^{i\Phi(g_\lambda)}) = \pi(e^{i\Phi(g)}).
\end{equation}
This property is known to hold on arbitrary globally hyperbolic spacetimes in the GNS representation of any quasifree state with distributional two-point function \cite[Appendix~E]{Fewster2023}; in particular, for any quasifree Hadamard state.

\subsection{Quantum measurement theory}\label{subsec:measurements}
For $A\in \N$, an $A$-outcome measurement in $B(\H)$ is expressed through a Positive Operator Valued Measure (POVM), i.e., an $A$-tuple of operators $M := (M_a)_{a=1}^A \in B(\H)^A$ such that
\begin{equation}
M_a\geq 0, \sum_{a} M_a=1.
\end{equation}
For a given state $\omega:B(\H)\to\C$, the probability of obtaining the result $a$ when we measure $M$ is given by:
\begin{equation}
P(a|\omega, M)=\omega(M_a).    
\end{equation}

To describe the state of the quantum system after it has been measured, POVMs do not suffice: one needs to use \emph{instruments}. A general quantum instrument is defined by an $A$-tuple of completely positive maps $\Omega:=(\Omega_a)_a$ on $B(\mathcal{H})$, satisfying the normalization condition
\begin{equation}
\sum_a\Omega_a(1)=1.
\end{equation}
Applying instrument $\Omega$ to a state $\omega$ results in the outcome $a$ with probability
\begin{equation}
P(a|\omega,\Omega)=\omega(\Omega_a(1)).
\label{instr2POVM}
\end{equation}
In case this probability is non-zero, the post-measurement state is
\begin{equation}
\omega_a(\bullet)=\frac{1}{P(a|\omega,\Omega)}\omega\circ\Omega_a(\bullet).
\end{equation}
By eq. (\ref{instr2POVM}), each instrument $(\Omega_a)_a$ implements the POVM $(\Omega_a(1))_a$. Note, however, that different instruments can implement the same POVM.

An instrument with a single outcome is called a \emph{channel}. To each instrument one associates the \emph{measurement channel}
\begin{equation}
    \bar{\Omega} := \sum_a \Omega_a,
\end{equation}
which allows describing the non-selective state update: any state $\omega$ acted on by $\Omega$ will evolve to 
\begin{equation}
    \omega'( \bullet) = \omega \circ \bar{\Omega} (\bullet) = \sum_aP(a|\omega,\Omega)\omega_a(\bullet),
\end{equation}
when ignoring the measurement outcomes.

States that can be represented through density matrices, i.e., where $\omega(A) = \operatorname{tr}(\rho A)$ for some $\rho \in B(\H)$ with $\tr \rho = 1$ and all $A \in B(\H)$, are called \emph{normal}. \emph{Normal instruments} are those instruments which preserve the class of normal states, i.e., where
\begin{equation}
    \omega_a := \omega\circ\Omega_a
\end{equation}
defines a normal state for all $a$ if $\omega$ is a normal state. Normal instruments are exactly those for which each of the completely positive maps $\Omega_a$ can be expressed through Kraus operators, i.e., there exists a countable family $(K_a^j)_{a,j} \subset B(\H)$ such that for all $a$, $\Omega_a$ can be given in terms of the strongly convergent sum
\begin{equation}
\Omega_a(\bullet)=\sum_jK_a^j\bullet (K_a^j)^*.
\end{equation}

It is often important that measurements leave some part of the system's state intact; loosely speaking, we call this a non-demolition measurement. More specifically, let $M:=(M_a)_a$ be a POVM and $\Omega$ an instrument implementing $M$, i.e., $M_a = \Omega_a(1)$. Then $\Omega$ is called \emph{non-disturbing} outside of a $\star$-subalgebra $\A \subset B(\H)$ if 
\begin{equation}\label{eq:nondist}
    \left. \overline{\Omega}(\bullet)\right\rvert_{\A'} = \bullet,
\end{equation}
i.e., $\overline{\Omega}(X) = X$ for all $X \in B(\H)$ which commute with $\A$. If $\Omega$ is normal, this implies that the Kraus operators associated with $\Omega$ commute with $\A'$, see Proposition \ref{prop:normallocality}. If, in addition, $\A$ is a von-Neumann algebra, i.e., $\A=\A''$, this implies that the Kraus operators associated with $\Omega$ lie in $\A$. A particular example is the generalized L\"{u}ders instrument defined via $\Omega_a := \sqrt{M_a} \bullet \sqrt{M_a}$, which is non-disturbing outside of the algebra generated by the $M_a$.

\subsection{Causal operations in QFT}

From the axioms of AQFT it appears natural to define local operations as normal instruments with Kraus operators taken from a corresponding local algebra. While microcausality then implies that there will be no superluminal signalling in any bipartite scenario, superluminal signalling scenarios can occur in scenarios with three or more parties; confer \Cref{fig:sorkin} for the canonical example known as Sorkin's paradox. 

\begin{figure}[H]
\centering
\includegraphics[width=.65\textwidth]{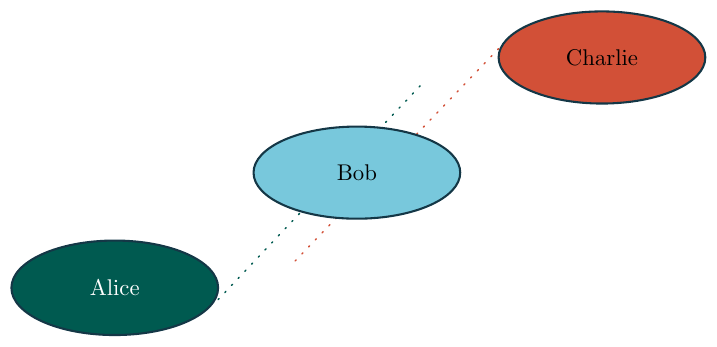}
\caption{Sorkin's paradox. Alice's and Charlie's zones of control are spacelike separated. Nonetheless, superluminal signalling between Alice and Charlie can occur if there are no restrictions on the allowed operations in Bob's region.}
\label{fig:sorkin}
\end{figure}

The usual response to the paradox is that the class of local operations must be restricted by further assumptions to guarantee relativistic causality, i.e., the absence of superluminal signalling scenarios. A criterion for operations to be causal, which we review below, was established in \cite{Jubb_2022}.

Let $\Omega$ be an instrument\footnote{In analogy with the operator setting, we mean here an appropriately normalized collection of completely positive maps acting on the $\star$-algebra $\A$.} associated with an AQFT $\A$ defined over a globally hyperbolic spacetime $\M$. We say that $\Omega$ is \emph{localizable within} the compact $\OO \subset \M$---or \emph{local to} $\OO$ for short---if its measurement channel $\bar{\Omega}$ satisfies
\begin{equation}
\left.\bar{\Omega}(\bullet)\right|_{\A(\OO^\perp)}=\bullet.
\end{equation}
Note that $\Omega$ can be local to various $\OO$; e.g.: if $\Omega$ is local to $\OO_1 \subset D(\OO_2)$, then it is also local to $\OO_2$ since $\OO_2^\perp \subset \OO_1^\perp$.

As we will see below, to discuss whether $\Omega$ is causal or not, it is very relevant where we consider it to be localized. We refer to such a choice $\OO$ as $\Omega$'s \emph{interaction zone}. 
Intuitively, $\OO$ represents the zone of spacetime where one needs to interact with the theory to implement $\Omega$ so that the map can act on elements of the algebra $\A(\OO^+)$, see \Cref{fig:jubb}.

\begin{figure}[H]
\centering
\includegraphics[width=.65\textwidth]{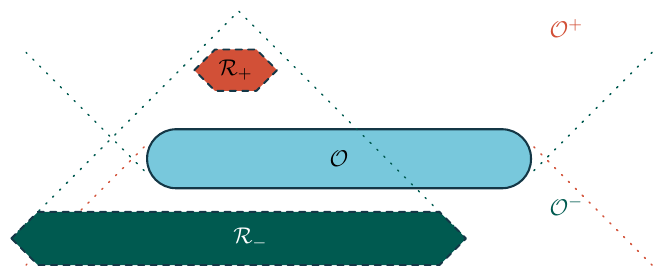}
\caption{Interaction zone $\OO$ and the PSNI property. Region $\OO^+$ ($\OO^-$) is above (below) the red (green) dashed line.}
\label{fig:jubb}
\end{figure}

It is proven in \cite{Jubb_2022} that an instrument local to $\OO$ does not lead to Sorkin-type causal paradoxes if (and even 'iff' for free scalar fields) it satisfies the \emph{past support non-increasing (PSNI) condition with respect to $\OO$}, namely, the condition that, for any pair of regions $\RR_\pm\subset \OO^\pm$ such that $\overline{\RR_{+}}\subset D(\RR_-)$, it holds that
\begin{equation}
\bar{\Omega}(\A(\RR_{+}))\subset \A(\RR_-).
\label{PSNI_def}
\end{equation}
We see that the PSNI condition depends\footnote{In \cite{Jubb_2022}, the dependence on the (arbitrary) choice for $\OO$ is not stressed, but implicit in the definition of PSNI.} on the chosen interaction zone $\OO$: the greater $\OO$ is, the less restrictive the PSNI condition becomes. Note, though, that taking the causal hull of $\OO$ does not alter the condition. 

\subsection{The FV formalism}\label{subsec:fv}

The FV formalism provides a framework to discuss quantum measurement theory for AQFTs, as defined in \Cref{subsec:aqft}. For a pedagogic introduction, we refer to \cite{fewsternotes}. The formalism is based on von-Neumann's concept of measurements, where the allowed measurements and instruments on a target are those induced by coupling to it a probe; measuring and discarding it subsequently. In the FV case, targets and probes are modeled as AQFTs, say $\A$ and $\B$, and it is required that also the coupled theory $\CC$ forms an AQFT. As a standing assumption, it is further required that outside of some compact coupling zone $\K \subset \M$, the coupled theory $\CC$ is isomorphic to the uncoupled theory $\U := \A\otimes \B$. Then, the out- and in-regions $\K^\pm$ both contain a Cauchy surface for $\M$ and it can be shown that there must exist isomorphisms $\tau_\pm: \U(\mathcal K^\pm) \rightarrow \mathcal C(\mathcal K^\pm)$, called response maps. Combined, they provide the map $\tau_-^{-1}\circ \tau_+: \U(\mathcal K^+) \rightarrow \U(\mathcal K^-)$ which extends, based on the timeslice property, to a $\star$-automorphism $\Theta$ on $\U$, called \emph{scattering morphism}, allowing us to evolve the observables of the uncoupled theory with respect to the coupling.

The scattering morphism is known to have several elementary yet important properties \cite[Prop. 3.1]{local_meas_QFT}, which we recall here and use without further reference. Firstly, given an enlarged compact $\hat{\mathcal K}$ containing $\mathcal K$, the scattering morphism $\hat\Theta$ derived using $\hat{\mathcal K}$ agrees with the usual one, $\hat\Theta=\Theta$. Secondly, the scattering morphism is a quantum channel, i.e., $\Theta$ is unital and completely positive\footnote{Note that $\Theta$ is a $\star$-isomorphism, and so $\Theta(\mathbb I)=\mathbb I$, $\Theta(O^*O)=\Theta(O)^*\Theta(O)>0$ for $\mathcal U \ni O\neq 0$. For any $n\in \N$ $M_n(\mathcal U)\cong \mathcal U\otimes M_n(\C)$ is also a $\star$-algebra. It follows that $\Theta_n\cong \Theta\otimes \mathbb I_n$ is a $\star$-isomorphism, and thus positive for all $n$.}. Thirdly, $\Theta$ is local to $\K$: If $\mathcal L\subset \mathcal K^\perp$, then $\Theta$ acts trivially on $\U(\mathcal L)$. Finally, $\Theta$ is PSNI: If $\mathcal {R}\subset \mathcal K^+, \mathcal S\subset \mathcal K^-$ with $D(\mathcal S)\supset \mathcal R$, then $\Theta (\U(\mathcal R)) \subseteq \U(\mathcal S)$. 

In a von Neumann measurement, the evolution is (approximately) $U=e^{-i A\otimes P}$ and a measurement on the probe side of a variable $X$ conjugate to $P$ induces a measurement of $A$ on the target side, as $U^*(\mathbb I\otimes X)U=A\otimes \mathbb I$. Likewise, in the FV framework, a measurement scheme induces an observable on the target $\mathcal A$. Specifically, a measurement scheme for the target observable $A$ is a selection of probe theory $\mathcal B$, probe state $\sigma$, probe observable $B$, and coupling $\Theta$ such that 
\begin{equation}\label{eq:inducedobs}
    A = \eta_\sigma\circ \Theta (\mathbb I\otimes B),
\end{equation}
where $\eta_\sigma:\mathcal A\otimes \mathcal B\rightarrow \mathcal A$ is the conditional expectation that corresponds to `tracing out' the probe state $\sigma$ and extends linearly from $\eta_\sigma(A\otimes B) = \sigma(B) A$. The target's evolution, or nonselective update, then corresponds to the channel 
\begin{equation}
    \overline{\Omega}:A\mapsto \eta_\sigma\circ\Theta(A\otimes \mathbb I).
    \label{induced_chan}
\end{equation}
Choosing a POVM $M$ on the probe, one obtains the selective updates
\begin{equation}\label{eq:inducedinstr}
    \Omega_a:A\mapsto \eta_\sigma\circ\Theta(A\otimes M_a),
\end{equation}
which together form the induced instrument $\Omega:=(\Omega_a)_a$, with measurement channel $\overline{\Omega}$. Specifically, $\Omega$ is the instrument that describes the evolution of $\A$ due to a measurement of the induced POVM $N=(N_a)_{a=1}^A$ on $\A$, where
\begin{equation}\label{eq:inducedpovm}
    N_a=\eta_\sigma\circ \Theta (\mathbb I\otimes M_a).
\end{equation}
That this instrument is local to and PSNI with respect to $\K$ follows immediately from the properties of $\Theta$. 

Deciding whether a channel $\bar{\Omega}$, an instrument $\Omega$ or a POVM $N_a$ can be induced entails finding a probe $\mathcal B$, a probe state $\sigma$, a probe POVM $M$ and a coupling $\Theta$ such that eqs. \eqref{induced_chan}, \eqref{eq:inducedinstr} or \eqref{eq:inducedpovm} respectively hold---or proving that no such $\mathcal B,\sigma,M,\Theta$ exist. If $\Omega_a$ can indeed be induced, we call such an instrument \emph{FV-realizable}. In the case of the free scalar field, we say that an instrument is \emph{asymptotically FV-realizable} in a given representation $\pi$ if there exists a net of FV schemes such that the corresponding net of induced instruments $(\Omega^n)_n$ satisfies
\begin{equation}
\lim_{n\to\infty}\omega(\Omega_a^n(A))=\omega(\Omega_a(A)), \forall a,  
\end{equation}
for any $A\in \pi(\hat{W})$ and any state $\omega$ in the folium of the abstract state generating $\pi$.

The notion of asymptotic FV-realizability also extends to POVMs: a POVM $N$ is asymptotically FV-realizable in some representation $\pi$ if there exists a net of FV schemes such that the associated POVMs $(N^n)_n$ satisfy
\begin{equation}
\lim_{n\to\infty}\omega(N_a^n)=\omega(N_a),
\end{equation}
for any state $\omega$ belonging to the folium of the GNS state of $\pi$.

As of now, it is known that any local POVM is FV-realizable \cite{fewster2025measurementpreparationprotocolsquantum} and that canonical measurements of free smeared fields, e.g. Gaussian measurements, can be realized asymptotically \cite{field_meas}. Sorkin's impossible measurements make it clear that many instruments cannot be asymptotically FV-realizable, as they are acausal, e.g.: ideal measurements \cite{albertini2023idealmeasurementsrealscalar}. As a concrete demonstration, we pick a simple example of an FV realizable instrument that does not even require a probe field: indeed, setting $\B = \C \id$ meets the axioms of AQFT trivially. Consider $\mathcal U$ to be the theory of a free real scalar field $\Phi$ (as constructed in \Cref{subsec:freeqft}) and $\mathcal C_f$ the corresponding one when $\Phi$ is subject to the influence of a classical external source $f \in C_0^\infty(\M)$. The associated scattering map $\Theta$ can be easily found (e.g., by methods discussed in Duetsch \cite{D_tsch_2001}) and we obtain  
\begin{equation}
    \Theta_f:e^{i\Phi(g)}\mapsto e^{i\Phi(f)}e^{i\Phi(g)}e^{-i\Phi(f)} = e^{-i\EE{f}{g}} e^{i\Phi(g)},
\end{equation}
for $\supp g \subset (\supp f)^+$. Thus, smeared field displacements or ``kicks'' are FV realizable. By contrast, adjoint actions with unitaries such as $e^{i\Phi(f)^2}$ are not FV realizable, since they are acausal \cite{Jubb_2022}.

At least three open questions on FV-realizability remain: 1) Which measurements can be realized in a non-demolition way, apart from measurements of the smeared field? 2) Are there causal instruments that cannot be asymptotically FV-realized? 3) If so, which causal instruments admit an asymptotic FV realization?

\subsection{Quantum nonlocality}
\label{sec:quant_nl}
Consider a scenario where two separate parties, whom in the following we shall call Alice and Bob, conduct experiments in two strictly space-like separated regions, see \Cref{fig:Bell}.

\begin{figure}[H]
\centering
\includegraphics[width=0.7\textwidth]{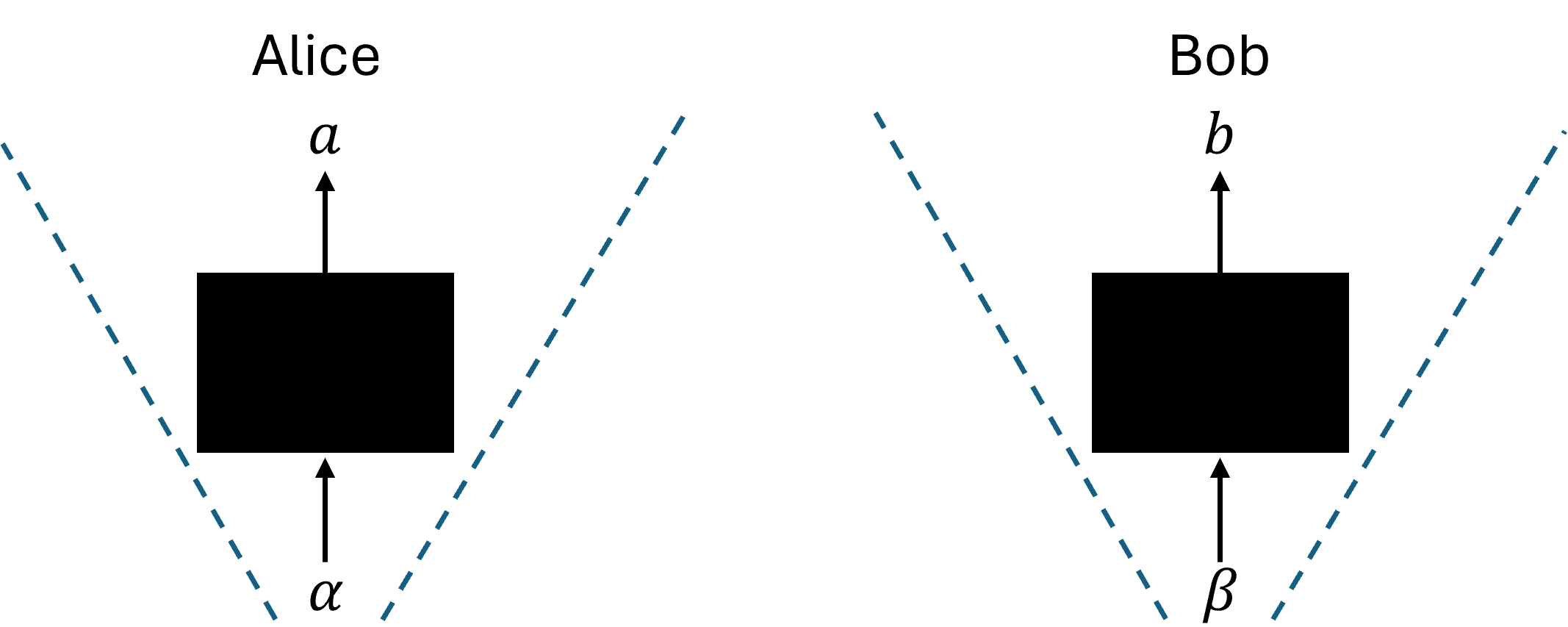}
\caption{Causal diagram of a Bell experiment.}
\label{fig:Bell}
\end{figure}

Denote by $\alpha,\beta \in \{0,...,N-1\}$ Alice's and Bob's respective choices of experiment and by $a,b \in \{0,...,M-1\}$ their respective outcomes. Each pair $(M,N)$ defines a different \emph{Bell scenario}. If Alice and Bob repeat this experiment many times, then they will be able to estimate the vector of probabilities
\begin{equation}
P=(P(a,b|\alpha,\beta):a,b,\alpha,\beta).    
\end{equation}
In quantum nonlocality, $P$ is dubbed a behavior, a box, a distribution or, simply, correlations. It is usually denoted by $P(a,b|\alpha,\beta)$: whether $P(a,b|\alpha,\beta)$ refers to a single probability or a vector thereof is clear from the context. We next introduce three relevant sets of behaviors, denoted by $C_{\mathrm{ns}}$, $C_{\mathrm{qa}}$ and $C_{\mathrm{qc}}$, which will play an important role in \Cref{sec:causal_not_FV,sec:undecidability}.

A behavior $P(a,b|\alpha,\beta)$ is called \emph{non-signalling} if
\begin{align}
&\sum_{a}P(a,b|\alpha,\beta)=P(b|\beta),\mbox{ independent of }\alpha,\nonumber\\
&\sum_{b}P(a,b|\alpha,\beta)=P(a|\alpha),\mbox{ independent of }\beta.
\label{NS_conds}
\end{align}
The first condition implies that Alice cannot signal to Bob through her choice of $\alpha$; the second condition means the same, but reversing the roles. Following standard terminology, we call $C_{\mathrm{ns}}$ the set of non-signalling behaviors \cite{Paulsen_2016}.

A behavior $P(a,b|\alpha,\beta)$ is called \emph{quantum} if it has a quantum representation, i.e., there exist two Hilbert spaces $\H_A,\H_B$, a normalized vector $\ket{\psi}\in \H_A\otimes \H_B$ and POVMs $\{M^A_{a|\alpha}\}_{a,\alpha}\subset B(\H_A), \{M^B_{b|\beta}\}_{b,\beta}\subset B(\H_B)$ such that
\begin{align}
&\sum_aM^A_{a|\alpha}= \sum_bM^B_{b|\beta}=1,\forall \alpha, \beta,\nonumber\\
&P(a,b|\alpha,\beta)=\bra{\psi}M^A_{a|\alpha}\otimes M^B_{b|\beta}\ket{\psi}.
\end{align}
The closure\footnote{Since we are considering finite sets, all norms on the space of behaviors are equivalent. A typical norm is given by the ``box norm" defined below.} of the set $C_{\mathrm{qs}}$ of all quantum behaviors is usually denoted by $C_{\mathrm{qa}}$ \cite{Paulsen_2016}.

The related set of behaviors $C_{\mathrm{qc}}$ is defined by demanding that there exist a joint Hilbert space $\H_{AB}$, a vector $\ket{\psi}\in\H_{AB}$ and POVMs $\{M^A_{a|\alpha}\}_{a,\alpha}, \{M^B_{b|\beta}\}_{b,\beta}\subset B(\H_{AB})$ satisfying:
\begin{align}
&\sum_aM^A_{a|\alpha}=\sum_bM^B_{b|\beta}=1,\forall \alpha, \beta,\nonumber\\
&[M_{a|\alpha},M_{b|\beta}]=0,\forall a,b,\alpha,\beta,\nonumber\\
&P(a,b|\alpha,\beta)=\bra{\psi}M^A_{a|\alpha}M^B_{b|\beta}\ket{\psi}.
\end{align}
$C_{\mathrm{qc}}$ is a relaxation of $C_{\mathrm{qs}}$: indeed, given any quantum representation $(\H_A,\H_B,\ket{\psi},M^A,M^B)$ of $P(a,b|\alpha,\beta)\in C_{qs}$, define $\H_{AB}:=\H_A\otimes\H_B$, $\tilde{M}^A_{a|\alpha}:= M^A_{a|\alpha}\otimes \id_{\H_B}$, and $\tilde{M}^B_{a|\alpha}:= \id_{\H_A}\otimes M^B_{b|\beta}$ so that $(\H_{AB},\ket{\psi},\tilde{M}^A,\tilde{M}^B)$ represents $P(a,b|\alpha,\beta)$ within the ``commuting'' paradigm associated to $C_{\mathrm{qc}}$.

The three sets $C_{\mathrm{qa}}, C_{\mathrm{qc}}, C_{\mathrm{ns}}$ are convex and closed, and the inclusion relations $C_{\mathrm{qa}}\subset C_{\mathrm{qc}}\subset C_{\mathrm{ns}}$ are strict. In the Bell scenario $N=M=2$, an example of a non-signalling behavior that cannot be approximated by elements of $C_{\mathrm{qc}}$ is the famous Popescu-Rohrlich box $P_{PR}$ \cite{khalfin1985quantum, Popescu1994-POPQNA}, with
\begin{equation}
P_{PR}(a,b|\alpha,\beta):=\tfrac{1}{2}\delta_{a\oplus b,\alpha\beta},\forall a,b,\alpha,\beta\in\{0,1\},
\label{PR_box}
\end{equation}
where $a\oplus b$ denotes the sum modulo $2$ of the outcome bits $a,b$. The existence of distributions $P(a,b|\alpha,\beta)\in C_{\mathrm{qc}}$ such that $P(a,b|\alpha,\beta)\not\in C_{\mathrm{qa}}$ is also known: it is a consequence of the recent negative resolution of Connes' embedding conjecture \cite{mipre}.  

For any Bell scenario $(M,N)$, a \emph{non-local game} $G$ is a pair $G=(p, V)$, where $p:\{0,...,N-1\}^2\to \R^+$ is a probability distribution, and $V:\{0,...,M-1\}^2\times \{0,...,N-1\}^2\to\{0,1\}$ is a score function. Given a behavior $P$, its game value is
\begin{equation}
G(P):=\sum_{\alpha,\beta}p(\alpha,\beta)\sum_{a,b}P(a,b|\alpha,\beta)V(a,b,\alpha,\beta).
\end{equation}
Intuitively, the ``game'' consists of a referee sampling the inputs $\alpha,\beta$ from the distribution $p(\alpha,\beta)$ and sending $\alpha$ ($\beta$) to Alice (Bob). Alice (Bob) inputs $\alpha$ ($\beta$) in her (his) box and sends the referee the outcome $a$ ($b$). Alice and Bob win the game if $V(a,b,\alpha,\beta)=1$; otherwise, they lose. The game value $G(P)$ therefore denotes the probability that Alice and Bob win the game through strategy $P$. It is easy to see that, for any two distributions $P,\tilde{P}\in C_{\mathrm{ns}}$,
\begin{equation}
|G(P)-G(\tilde{P})|\leq \|P-\tilde{P}\|,
\end{equation}
where, for any $Q:\{0,...,M-1\}^2\times \{0,...,N-1\}^2\to \R$, we define the ``box norm'':
\begin{equation}
\|Q\|:=\max_{\alpha,\beta}\sum_{a,b}|Q(a,b|\alpha,\beta)|.
\label{def:box_norm}
\end{equation}

For any game $G$ and each of the sets of correlations $\{C_x\}_x$, with $x=qa, qc, ns$, we define the maximum achievable game value as:
\begin{equation}
C_x(G):=\max_{P\in C_x}G(P).
\end{equation}

It is through non-local games that we know that the inclusion relations $C_{\mathrm{qa}}\subset C_{\mathrm{qc}}\subset C_{\mathrm{ns}}$ are strict. Consider, for instance, the Clauser-Horne-Shimony-Holt (CHSH) game $G_{CHSH}$ \cite{chsh}, defined in the $(2,2)$ Bell scenario as $p(\alpha,\beta)=\frac{1}{4}$, $V(a,b,\alpha,\beta)=\delta_{a\oplus b,\alpha\beta}$. It was shown by Tsirelson \cite{Cirelson1980} that
\begin{equation}
C_{\mathrm{qa}}(G_{CHSH})=C_{\mathrm{qc}}(G_{CHSH})=\frac{2+\sqrt{2}}{4}\approx 0.854.
\label{tsirelson_bound}
\end{equation}
However, one can verify that the Popescu-Rorhlich box $P_{PR}$ defined in eq. (\ref{PR_box}) satisfies $G_{CHSH}(P_{PR})=1$.

We finish this introduction by stating a celebrated result in quantum nonlocality, based on what from now on we refer to as \Cref{problem:MIP}. 

\begin{table}[H]
\centering
\begin{tabular}{|ll|}
     \hline
     INPUT: & a Bell scenario $(M,N)$ and a non-local game $G$ defined therein such that\\
     & either $C_{\mathrm{qa}}(G)=1$ or $C_{\mathrm{qa}}(G)<\frac{1}{2}$.\\
     \hline
     OUTPUT: & ``YES'' in the first case and ``NO'' in the second.\\
     \hline
\end{tabular}
\captionsetup{name=Problem}
\caption{Deciding a nonlocal game value.}
\label{problem:MIP}
\end{table}
This problem is undecidable \cite{mipre}. In fact, it is Recursively Enumerable-complete, which means that it is equivalent to the problem of deciding if a Turing machine will halt.

\section{Weyl operations}\label{sec:weylops}

In this section we define and motivate a large class of POVMs and instruments, whose elements, resp. Kraus operators, are elements of the Weyl algebra. The results in this section do not rely on test functions or spacetime. We thus consider an arbitrary fixed symplectic space $(X,\sigma)$ while keeping in mind the application to field theory, where $X = C_0^\infty(\M)/KC_0^\infty(\M)$ and $\sigma=E$ (see \Cref{subsec:freeqft}). The polynomial algebra over $(X,\sigma)$, denoted by $\hat{\P}$, is then given as the unital $\star$-algebra with Hermitian generators $\{ R_x : x \in X \}$ such that $x \mapsto R_x$ is linear and $[R_x,R_y] = i\sigma(x,y)$. By analogy with quantum optics, we call any nonzero $R_x$ a \emph{quadrature} and any pair of quadratures $(R_1,R_2)$ with $[R_1,R_2] = i$ a \emph{mode}. A pair of modes $(R_1,R_2)$, $(S_1,S_2)$ will be called \emph{independent} if $[R_j,S_k] = 0$ for all $j,k=1,2$. Note also that, for any non-zero quadrature, there exists another (non-unique) quadrature, called a \emph{conjugate}, such that together they form a mode. The Weyl algebra over $(X,\sigma)$ can then (with a slight abuse of notation) be defined as the $\star$-algebra with unitary generators $\{ e^{iR} : R \text{ quadrature}\}$ such that
\begin{equation}\label{eq:abstractweylrel}
    e^{iR}e^{iS} = e^{-\frac{1}{2} [R,S]} e^{i(R+S)}
\end{equation}
and will be denoted by $\hat{\W}$. 

In this work, we will be interested in algebras generated by a finite number of quadratures. Consequently, for a \emph{subspace of quadratures} $\Q$, i.e., any $\R$-linear subspace of $\{ R_x : x \in X\}$, we will denote by $\hat{\W}(\Q)$ the $\star$-algebras generated by $\{e^{iR} : R\in \Q\}$. Note that for any element $w \in \hat{\W}$, there exists a finite-dimensional subspace of quadratures $\Q$ such that $w \in \hat{\W}(\Q)$. Given any representation $\pi$ of $\hat{\W}$, we will also consider the von-Neumann algebra $\A_\pi(\Q) := \overline{\pi(\hat{\W}(\Q))}$, with closure being taken in SOT.

In the construction of POVMs with elements of the Weyl algebra, one faces the problem of guaranteeing that a certain element $w$ of $\hat{\W}$ satisfies $\pi(w)\geq 0$. In this regard, a simple \emph{sufficient} condition is that $w$ admits a sum of Hermitian squares (SOS) decomposition, i.e., there exist $n\in \N, \{d_l\}_{l=1}^n\subset \hat{\W}$ such that
\begin{equation}
    w=\sum_{l=1}^n d_ld^*_l.
\label{SOS_decomp_q}
\end{equation}
For any representation $\pi$ of $\hat{\W}$ it then follows that $\pi(w)=\sum_j \pi(d_l)\pi(d_l)^*\geq 0$.

If we find a subspace of quadratures $\Q$ such that \eqref{SOS_decomp_q} holds with $\{ d_l\}_l \subset \hat{\W}(\Q)$, then we say that $w$ is SOS over $\Q$, implying also $w \in \hat{\W}(\Q)$. 

That we are not loosing much by relying on this simpler notion of positivity is the content of the following lemma and remark:

\begin{lemma}\label{lemma:Archimedeanity}
    For any subspace of quadratures $\Q$ and any representation $\pi$ of $\hat{\W}$, let $w \in \hat{\W}(\Q)$ with $w = w^\ast$ and $\pi(w) > 0$. Then, $w$ is SOS over $\Q$.
\end{lemma}

The proof of this lemma is given in \Cref{app:Archimedean}, by combining the simplicity of the Weyl algebra with Schmmüdgen's positivstellensatz for Archimedean quadratic modules \cite{Schmudgen2009}.

\begin{remark}
\label{remark:Archimedeanity}
    Let $\Q$ be a finite subspace of quadratures. Lemma \ref{lemma:Archimedeanity} implies that, for any $w\in \hat{\W}(\Q)$ with $\pi(w)\geq 0$ and any $\epsilon > 0$, it holds that $w+\epsilon 1$ is SOS over $\Q$.
\end{remark}

We next introduce a useful characterization of the SOS property, in terms of positive semidefinite matrices, subject to linear constraints.
\begin{remark}\label{remark:zdecomp}
    $w\in \hat{\W}$ is SOS over $\Q$ iff there exist $Q_1,\ldots,Q_s\in \Q$ and an $s\times s$ positive semidefinite matrix $Z$ such that
    \begin{equation}
    w=\sum_{j,k=1}^sZ(j,k)e^{iQ_j}e^{-iQ_k}.
    \label{SOS_Z}
    \end{equation}
    Indeed, let $w$ admit the SOS decomposition (\ref{SOS_decomp_q}), for some $\{d_l\}_l\subset \hat{\W}(\Q)$. Then we find $s\in\N$ and $Q_1,...,Q_s\in \Q$ such that $d_l=\sum_{j=1}^s c_l^je^{iQ_j}$. Inserting this into \eqref{SOS_decomp_q}, we find \eqref{SOS_Z} with $Z = \sum_l \proj{c_l} \geq 0$ for vectors $\ket{c_l} := (c_l^j)_j$. Similarly, for any $w$ admitting a decomposition (\ref{SOS_Z}), with $Z\geq 0$, we can find vectors $\{\ket{c_l}\}_l$ such that $Z=\sum_l\proj{c_l}$, in which case $w$ admits an SOS decomposition (\ref{SOS_decomp_q}) with $d_l$ given as above.
\end{remark}

\subsection{Weyl POVMs and Weyl instruments}
We are ready to introduce a family of POVMs that will play an important role in this work. Ideally, we would like to define the set of all tuples $(M_a)_a$ of Hermitian elements of $\hat{\W}$ with the property that, for an operator representation $\pi$ of $\hat{\W}$, $(\pi(M_a))_a$ defines a POVM. 

In this respect, by our earlier remarks, we know that (finitely many) elements of the Weyl $\star$-algebra $\hat{\W}$ can be regarded as elements of $\hat{\W}(\Q)$, for some finite-dimensional subspace of quadratures $\Q$. In addition, by Remark \ref{remark:Archimedeanity} we do not lose much if, rather than enforcing the condition $\pi(M_a)\geq 0$, we simply demand that $M_a$ be SOS over $\Q$. Taking Remark \ref{remark:zdecomp} into account, we arrive at the following definition:
\begin{defin}
Let $A \in \N$ and let $\Q$ be a subspace of quadratures. An $A$-outcome \emph{Weyl POVM} over $\Q$ is a tuple $M := (M_a)_{a=1}^A$ of elements of $\hat{\W}(\Q)$ such that, for some $s\in\N$, there exist a tuple of different quadratures $Q := (Q_j)_{j=1}^s \in \Q^s$ and a tuple $Z:=(Z_a)_{a=1}^A$ of positive semidefinite $s\times s$ matrices with
\begin{equation}
M_a=\sum_{j,k=1}^sZ_a(j,k)e^{iQ_j}e^{-iQ_k} 
\label{sos}
\end{equation}
and $\sum_a M_a=1$. The pair $(Q,Z)$ will be called a \emph{matrix representation} of POVM $M$.
\end{defin}
Note that the normalization condition is equivalent to $\bar{Z} := \sum_a Z_a$ satisfying
\begin{equation}\label{znorm}
    \tr (\bar{Z}) = 1, \qquad \sum_{\substack{j,k=1\\j\neq k}}^s \bar{Z}(j,k) e^{iQ_j} e^{-iQ_k} = 0.    
\end{equation}

From all the above, it follows that, for any Weyl POVM $(M_a)_a$, the tuple of operators $(\pi(M_a))_a$ defines a POVM.

One wonders \emph{how much} we might be missing by restricting ourselves to implementing Weyl POVMs. As the next proposition shows, not much.
\begin{prop}
\label{prop:approx_POVM}
Let $(M_a)_{a=1}^A \in B(\H)^A$ be a POVM and $\pi$ an irreducible representation of $\hat{\W}$ on $B(\H)$. Then, there exists a sequence $(M^n)_n$ of Weyl POVMs such that
\begin{equation}
\slim{n\to\infty} \pi(M^n_a)=M_a, \forall a.    
\end{equation}
Further, if $(M_a) \subset \A_\pi(\Q)$ for a subspace of quadratures $\Q$, then $M_a^n$ can be chosen as elements of $\hat{\W}(\Q)$.
\end{prop}
The proof invokes the Weyl positivstellensatz (namely, Lemma \ref{lemma:Archimedeanity}), and the uniform boundedness principle \cite{Rudin1991FunctionalAnalysis}; the reader can find it in \Cref{app:approx}.

Let us now turn to instruments. Our discussion of Weyl POVMs suggests us to define the following class of instruments.
\begin{defin}
Let $A\in \N$ and let $\Q$ be a subspace of quadratures. An $A$-outcome Weyl instrument over $\Q$ is a tuple $\Omega := (\Omega_a)_{a=1}^A$ of maps on $\hat{\W}$ such that, for some $s\in\N$, there exist a tuple of different quadratures $Q:=(Q_j)_{j=1}^s \subset \Q$ and positive semidefinite $s\times s$ matrices $Z:= (Z_a)_{a=1}^A$ with
\begin{equation}
\Omega_a(X)=\sum_{j,k=1}^sZ_a(j,k)e^{iQ_j} Xe^{-iQ_k}, X \in \hat{\W},
\label{SOS_intr}
\end{equation}
and $\bar{\Omega}(1) = \sum_a \Omega_a(1)=1$. Analogously to the case of Weyl POVMs, the pair $(Q,Z)$ will be called a \emph{matrix representation} of instrument $\Omega$. A Weyl \emph{channel} is a Weyl instrument with $A=1$.
\end{defin}
Normalization is again equivalent to the normalization condition for (matrix representations of) Weyl POVMs, i.e., eq. \eqref{znorm}. Thus, any Weyl POVM can be implemented through a Weyl instrument and any Weyl instrument implements a Weyl POVM. Note further that any Weyl instrument over $\Q$ admits a Kraus-type decomposition with Kraus operators given by $K_a^l := \sum_j c_{a,l}^j e^{iQ_j} \in \hat{\W}(\Q)$ based on the decompositions $Z_a = \sum_l \proj{c_{a,l}}$ (see Remark \ref{remark:zdecomp}). Thus, given an operator representation $\pi$ of $\hat{\W}$, each Weyl instrument $\Omega$ determines a normal instrument $\Omega^\pi$ (on $B(\H)$), simply by replacing its Kraus operators $K_a^l$ by $\pi(K_a^l)$. 

How expressive are Weyl instruments, in comparison with normal instruments? The answer is: extremely! In \Cref{app:approx}, the reader can find a proof for the following analog of Proposition \ref{prop:approx_POVM}.
\begin{prop}
\label{prop:approx_instr}
Let $\Omega$ be a normal instrument (on $B(\H)$) and $\pi$ an irreducible representation of $\hat{\W}$ on $B(\H)$. Then, there exists a sequence of Weyl instruments $(\Omega^n)^n$ such that, for any $X\in B(\H)$,
\begin{equation}
\slim{n\to\infty} \Omega_a^{n,\pi}(X))=\Omega_a(X), \forall a.
\end{equation}
Moreover, for any subspace of quadratures $\Q$, if the Kraus operators of $\Omega$ are in $\A_\pi(\Q)$, then one can choose the Weyl instruments $(\Omega^n)^n$ to be defined over $\Q$.
\end{prop}

Therefore, a Weyl instrument is to a normal instrument what a Weyl POVM is to a general POVM. 

\subsection{Diagonal Weyl instruments}
\label{sec:diagonal}

There is an important subclass of Weyl instruments, which we will call diagonal. It is important in two regards:
\begin{enumerate}
    \item Diagonal Weyl instruments are prototypical for instruments with nice causality and realizability properties. As we will show and exploit in \Cref{sec:causality,sec:implementability}, if $\hat{\W}$ is associated with a QFT on spacetime as described in \Cref{subsec:freeqft}, then any diagonal Weyl instrument is causal and, moreover, realizable in Klein-Gordon fields within the FV-framework.

    \item Diagonal Weyl instruments allow inducing other, more complicated quantum operations. As we will show in this section, any POVM whatsoever can be asymptotically implemented through diagonal instruments. Moreover, any Weyl instrument admits a probabilistic implementation through a diagonal instrument, i.e., we show that there is a diagonal instrument which induces the Weyl instrument with certain nonzero probability.
\end{enumerate}
The two claims in Item 2 can be proven without invoking a field theory setting, so in this section we continue working with an abstract $\star$-Weyl algebra $\hat{\W}$. We start by defining diagonal Weyl instruments:
\begin{defin}\label{def:diag}
A Weyl instrument $\Omega$ over $\Q$ with outcomes in $A$ is \emph{diagonal} if its measurement channel $\bar{\Omega} = \sum_{a} \Omega_a$ takes the form
\begin{equation}\label{diagonal_def}
\bar{\Omega} =\sum_{j=1}^sp_je^{iQ_j}\bullet e^{-iQ_j},
\end{equation}
for some weights $\{p_j\}_j\subset \R_+$, $\sum_jp_j=1$ and $Q_1,...,Q_s\in\Q$.
\end{defin}
Given the matrix representation $(Q,Z)$ of a Weyl instrument, diagonality is equivalent to 
\begin{equation}
    \bar{Z}:=\sum_aZ_a=\operatorname{diag}(p_1,\ldots,p_s),    
\end{equation}
for some $p_j > 0$ with $\sum_j p_j =1$. Namely, diagonal instruments are those Weyl instruments $(Z,Q)$ for which the measurement channel matrix $\bar{Z}$ is diagonal. Alternatively, diagonality of a Weyl instrument $\Omega$ can be seen equivalent to
\begin{equation}
\bar{\Omega}(e^{iQ})=K(Q)e^{iQ},
\label{charac_diagonal_alt}
\end{equation}
for all quadratures $Q$ and for some function $K$ mapping quadratures to complex numbers.

It is easy to see that any Weyl instrument $(\Omega_a)_{a=1}^A$ can be probabilistically implemented through a diagonal Weyl instrument $(\tilde{\Omega}_a)_{a=1}^{A+1}$ in the following sense: for some $\mu\in (0,1]$, when the diagonal instrument $\tilde{\Omega}$ is applied, outcome $A+1$ will occur with probability $1-\mu$, independently of the state. Any other outcome $a=1,...,A$ for a state $\omega$ on $\hat{\W}$ will occur with probability $p_a := \mu \omega(\Omega_a(1))$ and will leave a post-selected state
\begin{equation}
\omega_a(\bullet)=\frac{\omega\left(\Omega_a(\bullet)\right)}{\omega\left(\Omega_a(1)\right)},
\end{equation}
namely, as if we had updated $\omega$ with respect to the original instrument $\Omega$. Conditioned on the outcome of $\tilde{\Omega}$ not being $A+1$, instrument $\tilde{\Omega}$ will therefore implement instrument $\Omega$. In quantum information science, one would say that $\tilde{\Omega}$ provides a \emph{heralded implementation} of $\Omega$ with success probability $\mu$.

Let us see why such a diagonal instrument $\tilde{\Omega}$ must always exist. Let $(Q,Z)$ be a matrix representation of $\Omega$, with $(Z_a)_a\subset \C^{s\times s}$. If $\bar{Z}=\sum_{a=1}^AZ_a$ is a diagonal matrix, then $\Omega$ itself is diagonal, and we can take $\tilde{\Omega}=\Omega$, $\mu=1$. If $\bar{Z}$ is not diagonal, let $\mu$ be the largest number such that, for some probability distribution $(p_1,...,p_s)$, it holds that
\begin{equation}
\mbox{diag}(p_1,...,p_s)-\mu\bar{Z}\geq 0.
\end{equation}
The fact that
\begin{equation}
\mbox{diag}\left(\frac{1}{s},...,\frac{1}{s}\right)-\frac{1}{sz_{\max}}\bar{Z}=\frac{1}{s}(\id-\frac{1}{z_{\max}}\bar{Z})\geq 0,
\end{equation}
where $z_{\max{}}>0$ denotes the largest eigenvalue of $\bar{Z}$, implies that $\mu >0$.

Then, the instrument $\tilde{\Omega}=(Q,\tilde{Z})$, with
\begin{equation}
\tilde{Z}_a= \left\lbrace \begin{matrix} \mu Z_a, & \mbox{ for }a=1,...,A,\nonumber\\ 
\mbox{diag}(p_1,...,p_n)-\mu\bar{Z}, & \mbox{ for }a=A+1 \end{matrix} \right.
\end{equation}
is diagonal by construction and provides a heralded implementation of $\Omega$ with success probability $\mu$.

Another reason why diagonal Weyl instruments are relevant is that Weyl POVMs and even general POVMs can be asymptotically implemented through a diagonal Weyl instrument defined over the same subspace of quadratures. Our main result is:

\begin{theorem}
\label{theo:diagonal_POVM}
    Given any POVM $(M_a)_a \in B(\H)^A$ and any irreducible representation $\pi$ of $\hat{\W}$ on $B(\H)$, there exists a sequence $(\Omega^n)_n$ of diagonal Weyl instruments such that
\begin{equation}
    \slim{n\to\infty} \pi(\Omega^n_a(1)) = M_a, \forall a.
\end{equation}
    Moreover, if $(M_a)_a \in \A_\pi(\Q)^A$ for some subspace of quadratures $\Q$, then the instruments $(\Omega^n)_n$ can be defined over $\Q$.
\end{theorem}

Theorem \ref{theo:diagonal_POVM} is a consequence of the following structural result for Weyl channels, proven in \Cref{app:structural}.
\begin{lemma}
\label{lemma_asymp}
Let $\Omega$ be a Weyl channel over $\Q$. For arbitrarily small $\epsilon>0$ there exist Weyl channels $\Omega', \Omega''$ over $\Q$ such that
\begin{equation}
\Omega_\epsilon:=(1-\epsilon)\Omega\circ\Omega'+\epsilon\Omega''
\label{approx_channel}
\end{equation}
is diagonal.
\end{lemma}

With Lemma \ref{lemma_asymp}, Theorem \ref{theo:diagonal_POVM} is proven as follows:

\begin{proof}[Proof of Theorem \ref{theo:diagonal_POVM}] Let $(M_a)_{a=1}^A$ be a Weyl POVM with matrix representation $(Q,Z)$, with $Q_1,...,Q_n\in {\cal Q}$. Let $(\Omega_a)_a$ be the canonical instrument realization of $(M_a)_a$, namely,
\begin{equation}
\Omega_a(X)=\sum_{j,k}Z_a(j,k)e^{iQ_j}Xe^{-iQ_k}.
\end{equation}
Consider then the Weyl channel $\Omega=\sum_a\Omega_a$. By Lemma \ref{lemma_asymp}, we have that there exist Weyl channels $\Omega',\Omega''$ such that eq. (\ref{approx_channel}) holds. We thus define the Weyl instrument
\begin{equation}
\tilde{\Omega}_a:=(1-\epsilon) \Omega_a\circ\Omega'+\epsilon A^{-1}\Omega''.
\end{equation}
By eq. (\ref{approx_channel}), it holds that $\sum_a\tilde{\Omega}_a$ is diagonal. Moreover, from the equation above, the induced POVM $M'$ with $M'_a=\tilde{\Omega}_a(1) = (1-\epsilon)M_a + \epsilon A^{-1}1$ satisfies
\begin{equation}
\|M'_a-M_a\| \leq \epsilon (\norm{M_a} + A^{-1}\norm{1}) \leq \epsilon\left(1+A^{-1}\right),
\end{equation}
so that $\tilde{\Omega}$ comes arbitrarily close to implementing the Weyl POVM $M$. Combining this result with Proposition \ref{prop:approx_POVM}, we also get arbitrarily close to implementing any POVM $M \subset B(\H)^A$.\end{proof}

\section{Causality of Weyl instruments}
\label{sec:causality}

In this section we explore the causality properties of Weyl instruments and channels. We take them to act on the Weyl $\star$-algebra $\hat{\W}$ associated with a given QFT on a fixed spacetime $\M$, as described in \Cref{subsec:freeqft}. Since the Weyl algebra is linearly spanned by its generators (which we denote by $e^{i\Phi(h)}$, $h \in C_0^\infty(\M)$), to study the causality properties of a channel it suffices to investigate its action on $e^{i\Phi(h)}$.

In the following, for any pair of test functions $f,g\in C_0^\infty(\M)$, we say that $f,g$ are equivalent if $f-g=Kr$, for some $r\in C_0^\infty(\M)$, where $K$ denotes the hyperbolic differential operator defining the field's classical equation of motion. We denote by $[f]$ the resulting equivalence class and remark that $\Phi(f) = \Phi([f])$. We say that $s=[f]$ is \emph{localizable} in a region $\RR \subset \M$ iff $e^{i\Phi(s)} \in \hat{\W}(\RR)$, i.e., there exists $g \in C_0^\infty(\M)$ with $s = [g]$ and $\supp g \subset \RR$. We recall that, for a compact $\OO \subset \M$, $\OO^\pm := \M \setminus J^\mp(\OO)$ defines the in/out-region with respect to $\OO$. 

The proofs of most of our results on causality rely on one or several of the following lemmas: 
\begin{lemma}
\label{lemma:localization}
$e^{i\Phi(s)}$ is localizable in region $\RR$ iff there exists a Cauchy surface $\Sigma \subset \M$ such that $\supp Es \cap \Sigma \subset \RR$.
\end{lemma}

\begin{lemma}\label{lemma:psni_simpler}
Let $\Omega$ be a channel localized in a compact $\OO\subset \M$. $\Omega$ is causal with respect to $\OO$ iff, for all pairs of regions $\RR_{\pm}\subset\OO^{\pm}$, with $\overline{\RR_{+}}\subset D^+(\RR_{-})$, it holds that $\Omega(\A(\RR_{+}))\subset \A(\RR_-)$.
\end{lemma}
Note that, rather than demanding $\overline{\RR_{+}}\subset D(\RR_{-})$ like in the PSNI condition (\ref{PSNI_def}), we enforce the stronger constraint $\overline{\RR_{+}}\subset D^+(\RR_{-})$.

\begin{lemma}\label{lemma:weyllocality}
Let $\Omega$ be a Weyl instrument. Then, $\Omega$ is local to a compact $\OO$ iff its Kraus operators are localizable in any region $\RR$ containing $\OO$. In particular, $\Omega$ can be assumed to have a matrix representation $(Q,Z)$ with $Q_j = \Phi(f_j)$ and $\mbox{Supp}(Ef_j) \subset J(\OO)$.
\end{lemma}
The first lemma is proven in \cite[Section V]{Jubb_2022} or \cite[Thm.~3.3.1]{AQFT_adv}; the remaining two are original, and the reader can find a proof thereof in \Cref{app:lemmas_causal}.

\begin{remark}
Despite the intuition that the interaction zone $\OO$ is ``where the experimenter acts'', there exist Weyl instruments $(Q,Z)$, causal with respect to $\OO$, such that, for all $j$, there exists no $f_j\in C_0^\infty(\OO)$ with $Q_j = \Phi(f_j)$. Example: let $\OO=\OO_1\cup\OO_2$ as in Figure \ref{fig:weird_test_function}, and define the Weyl channel $\Omega(\bullet)=e^{i\Phi(f)}\bullet e^{-i\Phi(f)}$. 

\begin{figure}
\centering
\includegraphics[width=.45\textwidth]{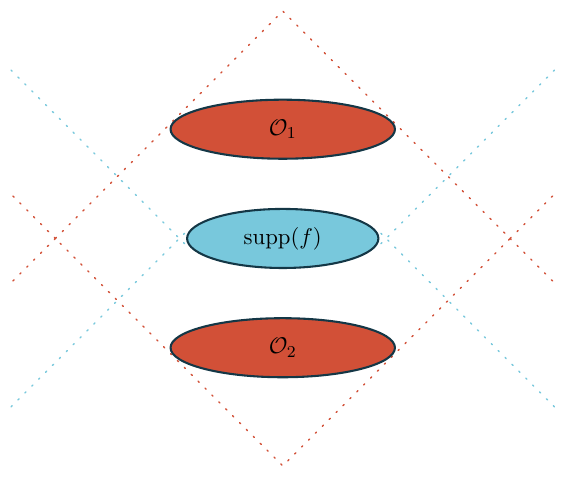}
\caption{A causal displacement channel whose test function cannot be represented in $\OO$.}
\label{fig:weird_test_function}
\end{figure}

Taking $f\in C_0^\infty(\M)$ such that $\mbox{Supp}(f)\subset \chull(\OO)$, we guarantee that $\Omega$ is local to $\OO$, as its Kraus operator will commute with any element of $\A(\chull(\OO)^\perp)=\A(\OO^\perp)$. Moreover, since $\Omega$ is Weyl diagonal, it is also causal with respect to $\OO$, see section \ref{sec:diagonalcausal}. However, as the picture shows, $f$ can be chosen such that $\mbox{Supp}(f) \not\subset D(\OO)$, in which case $[f]$ cannot be localized (in general) in $\OO$.

\end{remark}

\subsection{Characterization of causal Weyl instruments}
\label{sec:causal_charac}
We next provide a simple characterization of the set of Weyl instruments complying with the PSNI condition. Before proceeding, though, we need to introduce a couple of definitions. Consider a Weyl channel $\Omega$ with matrix representation $(Q,Z)$, with $Q_j=e^{-i\Phi(f_j)}$, $f_j\in C_0^\infty(\M)$, for all $j$, and $[f_j]\not=[f_k]$, for all $j\not=k$. We define the set ${\cal D}$ of differences between test functions
\begin{equation}
{\cal D}=\left\{[f_j-f_k]:j,k\right\}
\end{equation}
and, for any $d\in {\cal D}$, the associated set of semi-sums
\begin{equation}
{\cal S}(d)=\left\{\left[\frac{f_j+f_k}{2}\right]:[f_j-f_k]=d\right\}.
\end{equation}
Then, the action of $\Omega$ on a generator $e^{i\Phi(h)}$ can be expressed as:
\begin{equation}
\Omega(e^{i\Phi(h)})=\sum_{d\in {\cal D}}\Omega^{(d)}(e^{i\Phi(h)}),
\end{equation}
with
\begin{equation}
\Omega^{(d)}(e^{i\Phi(h)}):=\sum_{s\in {\cal S}(d)}c^d_se^{-i\EE{s}{h}}e^{i\Phi(h+d)},
\label{expr_omega_d}
\end{equation}
where, for $s=\left[\frac{f_j+f_k}{2}\right]$, $d=[f_j-f_k]$,
\begin{equation}
c^d_s:=Z(j,k)e^{i\frac{\EE{f_j}{f_k}}{2}}.
\end{equation}
We warn the reader that the linear map $\Omega^{(d)}$ is neither an instrument nor an outcome thereof: in fact, for $d\not=0$, $\Omega^{(d)}(\bullet)$ maps Hermitian to non-Hermitian elements of $\hat{\W}$.

With this notation, the normalization condition $\Omega(1)=1$ can be expressed as
\begin{align}
\sum_{s\in {\cal S}(d)}c^d_s=&0, \mbox{ for } d\not=0,\nonumber\\
&1, \mbox{ for } d=0.
\end{align}

Now, consider a potential interaction zone $\OO\subset \M$ for $\Omega$, with $J(\OO)\supset \cup_j\mbox{Supp}(Ef_j)$. Given $d\in {\cal D}$, call ${\cal P}(d)$ the set of partitions $\{A_1,...,A_n\}$ of ${\cal S}(d)$ such that there exists a region $\RR_{+}\subset \OO^+$ satisfying these two conditions:
\begin{enumerate}
    \item $J^{-}(\overline{\RR_{+}})\not\supset \mbox{int}\left(\mbox{Supp}(Ed)\right)\cap \OO^-=:G_d$.
    \item For $j,k=1,...,n$, $j>k$, there exists $x\in \RR_{+}$ with $(Es)(x)\not=(Et)(x)$, for all $s\in A_j$, $t\in A_k$.
\end{enumerate}
\begin{remark}
\label{remark:G_d}
$G_d\subset J^-(\OO)$. This follows from the relations $\mbox{Supp}(Ed)\subset J^-(\OO)\cup J^+(\OO)$, $\OO^-=\M\setminus J^+(\OO)$.
\end{remark}

\begin{remark}
Establishing that some partition $\{A_1,...,A_n\}$ belongs to ${\cal P}(d)$ amounts to finding points $\{x_{jk}:j>k\}\subset \OO^+$ with $J^-\left(\{x_{jk}:j>k\}\right)\not\supset G_d$ and such that
\begin{equation}
(Es)(x_{jk})\not=(Et)(x_{jk}), \mbox{ for all } s\in A_j, t\in A_k.
\label{not_equal_points}
\end{equation}
 Indeed, if said region $\RR_{+}$ exists, then there must exist points $\{x_{jk}:j>k\}\subset \RR_{+}$ satisfying the equation above. Moreover, $J^-\left(\{x_{jk}:j>k\}\right)\subset J^-\left(\RR_{+}\right)\not\supset G_d$. Conversely, if there exist points $\{x_{jk}:j>k\}\subset \OO^+$ with $J^-\left(\{x_{jk}:j>k\}\right)\not\supset G_d$ satisfying eq. (\ref{not_equal_points}), then there exist sufficiently small neighborhoods $B_{jk}$ of $x_{jk}$ within $\OO^+$ such that $J^-\left(\overline{B}\right)\not\supset G_d$, with $B=\cup_{j>k}B_{jk}$. In that case, taking $\RR_+=\chull(B)$, it follows that $\RR_+ \subset \OO^+$ and $J^- (\RR_+)=J^-(B)\not\supset G_d$.
\end{remark}

In Appendix \ref{app:causal_charact}, we derive the following characterization of causal Weyl channels.
\begin{theorem}
\label{theo:causal}
Let $\Omega$ be a Weyl channel, local with respect to $\OO$. Then $\Omega$ is causal with respect to $\OO$ iff, for any non-zero $d\in {\cal D}$ and any partition $\{A_1,...,A_n\}\in {\cal P}(d)$, it holds that
\begin{equation}
\sum_{s\in A_j}c_s^d=0, \mbox{ for } j=1,...,n.
\label{eq:null_sum}
\end{equation}
\end{theorem}
For Weyl channels $\Omega$ with an interaction zone $\OO$ for which the sets of partitions $\P(d)$ are computable, determining the causality of $\Omega$ with respect to $\OO$ is therefore a decidable problem. The same considerations apply to Weyl instruments.

\subsection{Diagonal Weyl instruments}
\label{sec:diagonalcausal}
In the previous section, we provided a full characterization of causal Weyl instruments. In this section, we focus on diagonal Weyl instruments (see \Cref{sec:diagonal}), for which the membership problem is trivial. Indeed, due to relation (\ref{charac_diagonal_alt}), any diagonal Weyl instrument is PSNI with respect to any region to which it is local. E.g.: if $\Omega$ has the matrix representation $(Q,Z)$ with $Q_j = \Phi(f_j)$, then $\Omega$ is PSNI with respect to any compact containing $\cup_j \supp f_j$.

How rare are diagonal Weyl instruments among the greater set of causal Weyl instruments? To address this question in a meaningful way, we must first define a notion of typicality over test functions.
\begin{defin}\label{def:generic}
Given a tuple of regular compact sets $(\K_j\subset \M)_{j=1}^m$, we say that $(f_j \in C^\infty(\K_j))_j$ are \emph{generic} test functions with respect to $(\K_j)_j$ if, for any open set $V\subset \M$, with
\begin{equation}
V\cap J(\K_j)\not=\emptyset, \forall j, 
\label{causal_cond_generic}
\end{equation}
it holds that the functions
\begin{equation}
\left\{\left.E f_j\right|_{V}\right\}_j
\end{equation}
are linearly independent. We say that a QFT is \emph{generic} if, for any $m\in\N$ and any regular compact sets $(\K_j)_{j=1}^m \subset \M$, there exist test functions $(f_j)_j$ that are generic with respect to $(\K_j)_j$. Similarly, we say that test functions and QFTs are \emph{null-generic} if we replace $J(\K_j)$ in \eqref{causal_cond_generic} by $N(\K_j)$.
\end{defin}
In Appendix \ref{app:generic}, we prove that the Klein-Gordon fields in Minkowski spacetime in 1+1 and 1+3 dimensions are generic QFTs. Similarly, we show that any scalar QFT with a commutator function $E$ associated with a normally hyperbolic equation of motion is null-generic. This includes the Klein-Gordon theory in $d \geq 2$ spacetime dimensions.

The reason to call any such $n$-tuple of functions ``generic" is that, in any (null) generic QFT, (null) generic test functions are dense among the set of test functions defined in $C_0^\infty(\K_1)\times ...\times C_0^\infty(\K_m)$, see Remark \ref{remark:density}.

Coming back to Weyl instruments, recall they they are always defined over a finite-dimensional subspace of quadratures $\Q$, with $\Q=\mbox{span}\{[f_j]\}_{j=1}^m$. Let us assume that $\mbox{Supp}(f_j)\subset \OO$, for all $j$, for some regular compact $\OO\subset \M$. We next show that, if $\{f_j\}_j$ are (null) generic with respect to the compacts $(\OO,...,\OO)$, then any causal Weyl instrument over $\Q$ with interaction zone $\OO$ must be diagonal. This is a consequence of the following result:
\begin{theorem}
\label{theo:generic_diagonal}
Let $\OO\subset \M$ be a regular compact, and let $\Q$ be a subspace spanned by the quadratures $\Phi(f_j), j=1,...,n$, with $f_j\in C_0^\infty(\OO),\forall j$. Let $\RR\subset \OO^+$ be a region and let $\T\subset {\cal O}^-$ be an open set, with $J^-(\overline{\RR})\cap \T=\emptyset$. Define the sets of functions:
\begin{equation}
{\cal F}_\RR:=\left\{\left.Ef_j\right|_{\RR}\right\}_{j=1}^n,{\cal F}_\T:=\left\{\left.Ef_j\right|_{\T}\right\}_{j=1}^n.
\label{cond_lin_indep}
\end{equation}
If the functions in ${\cal F}_\RR$ and in ${\cal F}_\T$ are linearly independent, then the only causal Weyl channels over $\Q$ with interaction zone $\OO$ are diagonal.
\end{theorem}

Now, for any regular compact $\OO\subset\M$, it is always possible to find a region $\RR\subset \OO^+$ and an open set $\T\subset \OO^-$ satisfying $J(\overline{\RR})\cap\T=\emptyset$, $\RR\cap J(\OO)\not=\emptyset$ ($\RR\cap N(\OO)\not=\emptyset$) and $\T\cap J(\OO)\not=\emptyset$ ($\T\cap N(\OO)\not=\emptyset$). Hence, any tuple of (null) generic test functions $(f_j)_j \in C_0^\infty(\OO)^n$ will satisfy eq. (\ref{cond_lin_indep}). Putting everything together, we have that, for the Klein-Gordon field, causal Weyl channels over ``typical" subspaces of quadratures $\Q$ are all diagonal.

\begin{proof}[Proof of Theorem \ref{theo:generic_diagonal}]
Let $\Omega$ be a Weyl channel over $\Q$, with matrix representation $(Q,Z)$ such that $Q_k=\Phi(\vec{t}_k\cdot \vec{f}), \vec{t}_k \in \R^n, \forall k$, with $\vec{f}=(f_1,...,f_n)$. Note that the condition $Q_j\not=Q_k$, for all $j\not=k$ implies, by linear independence of the $f$'s, that $\vec{t}_j\not=\vec{t}_k$. Let $d:=[\vec{d}\cdot \vec{f}]\in {\cal D}(\Omega)$, with $d\not=0$ (equivalently, $\vec{d}\not=0$). By the linear independence of the functions in ${\cal F}_{\T}$, we have that $\left.Ed\right|_{\T}\not=0$, or $\T\cap \supp(Ed)\not=\emptyset$. Since $J^-(\overline{\RR})\cap \T=\emptyset$, $\T\subset \OO^-$, it follows that $G_d = \supp(Ed)\cap \OO^-\not\subset J^-(\overline{\RR})$.

Now, let $\S(d)=\{s^l\}_{l=1}^L$, with $s^l=\vec{u}^l\cdot\vec{f},\forall l$. For $j>k$, it holds that $\vec{u}^j-\vec{u}^k\not=0$. By the linear independence of the functions in ${\cal F}_{\RR}$ it follows that 
\begin{equation}
E\left.(\vec{u}^j-\vec{u}^k)\cdot \vec{f}\right|_{\RR}\not=0.    
\end{equation}
In other words, for any $j>k$, there exists a point $x_{jk}\in \RR$ such that 
\begin{equation}
Es^j(x_{jk})\not=Es^k(x_{jk}).
\end{equation}
This implies that the partition $\{\{1\}, \{2\},...,\{L\}\}$ belongs to $\P(d)$. If $\Omega$ is causal, it follows, by Theorem \ref{theo:causal}, that $c_s^d=0,\forall s\in \S(d)$. Since this holds for all non-zero $d\in {\cal S}$, channel $\Omega$ must be diagonal.
    
\end{proof}

\subsection{Communication networks of causal Weyl instruments}
\label{sec:networks}
Given a causal Weyl instrument $\Omega$, and a family thereof $\Xi=\{\Xi^a\}_a$, it is tempting to use the measurement outcome of $\Omega$ to decide which instrument $\{\Xi^a\}_a$ to implement. The resulting instrument $\Omega'$ would then have a double outcome $(a,b)$, associated to the map:
\begin{equation}
\Omega'_{(a,b)}(\bullet)=\Omega_{a}\circ \Xi^a_b(\bullet).
\end{equation}
What we have just done is \emph{feed-forward} $\Omega$'s measurement outcome to $\Xi$. It is straightforward to verify that $(\Omega'_{(a,b)})_{a,b}$ is indeed an instrument. The question is under which conditions $\Omega'$ is causal, given that its constituent instruments $\Omega,\{\Xi^a\}_a$ are all causal.

Assume that all the instruments in $\{\Xi^a\}_a$ are causal with respect to the same interaction zone $\OO(\Xi)$, and let $\Omega$ be causal with respect to $\OO(\Omega)$. Following the intuition, suggested in \cite{Jubb_2022}, that the outcomes of a QFT instrument with interaction zone $\OO$ are only accessible to observers in the total future of $\OO$, we would expect $\Omega'$ to be causal if $\OO(\Omega),\OO(\Xi)$ are causally orderable and $\OO(\Xi)\subset \OO(\Omega)^\vee$; see \Cref{fig:network} (left) for a simple example.

We can go beyond the composition of two instruments and think of integrating multiple families of instruments within a complex communication network, see \Cref{fig:network} (right) for a more complicated example. 
\begin{figure}[ht]
\centering
\includegraphics[width=0.3\textwidth]{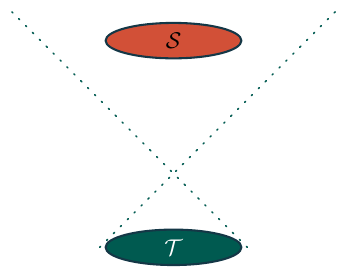}
\includegraphics[width=0.45\textwidth]{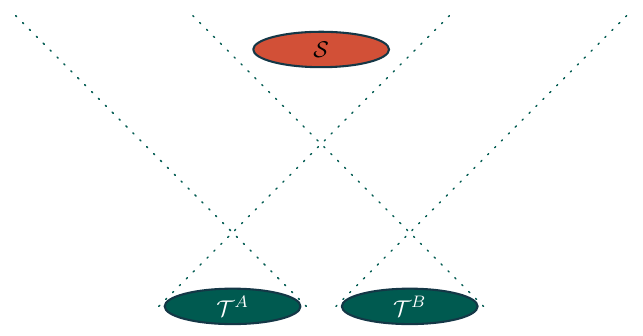}
\caption{Two examples of communication networks feeding forward outcomes of instruments localized in the interaction zone $\T$ (left) or $\T^A$ and $\T^B$ (right) to instruments localized in the interaction zone $\S$.}
\label{fig:network}
\end{figure}

In an $n$-node network, one would start from a finite list of causally orderable compact subsets of the spacetime manifold $\M$, $\OO_1\prec....\prec \OO_n$. Each network node $j\in \{1,...,n\}$ would be associated to a causal Weyl instrument $\Omega^j$ with interaction zone $\OO_j$ and outcomes $c_j$, or a family thereof $(\Omega^{j,c(\wedge_j)})_j$ conditioned on the outcomes $c(\wedge_j)$ which are accessible at node $j$. Here we define $\wedge_j := \{ l \in \{ 1,\ldots,n\} : \OO_j \subset \OO_l^\vee \}$ and, for any subset $I \subset \{ 1,\ldots,n\}$, $c(I)$ denotes the subvector of $c := (c_1,\ldots,c_n)$ associated with the indices $I$. The effective instrument $\Omega := (\Omega_c)_c$ implemented on the QFT would then be
\begin{equation}
\Omega_c(\bullet):= \Omega^{1,c(\wedge_1)}_{c_1} \circ \ldots \circ \Omega^{n,c(\wedge_n)}_{c_n}(\bullet).
\label{composed_instr}
\end{equation}

If our interpretation of the outcomes of a causal Weyl instrument is correct, then one would expect $\Omega$ to be causal as well. This is actually the case.

\begin{theorem}\label{theorem:causal_network}
Let $(\OO_1,...,\OO_n)$ be a vector of compacts in causal order from past to future and, for $j=1,\ldots,n$, let the Weyl instruments $\left\{(\Omega^{j,c(\wedge_j)}_{c_j})_{c_j}:c(\wedge_j)\right\}$ be causal with interaction zone $\OO_j$. Then, the instrument (\ref{composed_instr}) is causal with respect to $\bigcup_j \OO_j$.
\end{theorem}
The reader can find a proof and a remark on the generalization of this result to arbitrary normal instruments at the end of \Cref{app:causal_networks}.

Theorem \ref{theorem:causal_network} allows constructing instances of non-diagonal causal Weyl instruments through the composition of diagonal instruments. For example, let $\S,\T\subset \M$ be two regular compacts such that $\S\cap \T=\emptyset$ and $\S \subset \T^\vee$, and let $f_+,f_-\in C_0^\infty(\S)$, $g\in C_0^\infty(\T)$. We next define the diagonal instrument $\Omega$, with matrix representation $((1,e^{i\Phi(g)}), (Z_+,Z_-))$, where
\begin{equation}
Z_\pm:=\frac{1}{4}\begin{pmatrix}
        1 & \pm1 \\
        \pm1 & 1
    \end{pmatrix}.
\end{equation}
Consider now the channel $\Omega'$, given by:
\begin{equation}\label{eq:non_diag_ex}
\Omega'(\bullet)=\sum_{a=\pm}\Omega_a\left( e^{i\Phi(f_a)}\bullet e^{-i\Phi(f_a)}\right).
\end{equation}
$\Omega'$ is the result of feedforwarding the measurement outcome $a$ of $\Omega$, with interaction zone $\T$, to the family of diagonal channels $(e^{i\Phi(f_a)}\bullet e^{-i\Phi(f_a)})_a$, with interaction zone $\S$. Since all channels are diagonal, and thus causal, and $\S \subset \T^\vee\cap \T^+$, it follows that $\Omega'$ is also causal. However, due to $[f_+]\not=[f_-]$, its expansion shows genuinely off-diagonal terms:
\begin{align}
\Omega'(\bullet)=&\frac{1}{4}\sum_{a=\pm}\left(e^{i\Phi(f_{a})}\bullet e^{-i\Phi(f_{a})}+e^{i\Phi(f_{a}+g)}\bullet e^{-i\Phi(f_{a}+g)}\right)\nonumber\\
&+a\left(e^{-\frac{iE(f_a,g)}{2}}e^{i\Phi(f_{a})}\bullet e^{-i\Phi(f_{a}+g)}+e^{\frac{iE(f_a,g)}{2}}e^{i\Phi(f_{a}+g)}\bullet e^{-i\Phi(f_{a})}\right).
\end{align}

In conclusion, even though causal Weyl instruments defined over \emph{generic} subspaces of quadratures are all diagonal, this is not true for causal Weyl instruments defined over \emph{general} subspaces.

\section{Weyl instruments in the FV framework}
\label{sec:implementability}

In this section, we will analyze how Weyl instruments fit within the FV formalism presented in \Cref{subsec:fv}. We will prove, in \Cref{subsec:withinfv}, that, for Klein-Gordon fields, all diagonal Weyl instruments admit an asymptotic FV realization. We next show, in \Cref{sec:causal_not_FV}, that certain causal Weyl channels cannot be approximated arbitrarily well through asymptotic FV schemes, even if we allow the coupling zone to be arbitrarily large (but bounded). Finally, in \Cref{sec:undecidability} we consider the problem of characterizing which Weyl instruments can be asymptotically induced through the composition of several FV schemes. In this regard, we construct a countable family of Weyl channels for which we prove that there is no general algorithm to tell apart implementable from far-from-implementable channels. 

As standing assumptions, we suppose our target theory to be the Weyl algebra of the real scalar field $\Phi$, as defined in \Cref{subsec:freeqft}. For our FV schemes, we will invoke probe theories of the same type and denote the associated probe fields by $\Psi$.

\subsection{Weyl instruments with an FV realization}\label{subsec:withinfv}

First, we will show that all diagonal Weyl instruments admit an asymptotic FV realization. From the results of \Cref{sec:weylops}, it thus follows that the FV framework allows the asymptotic realization of arbitrary non-demolition measurements and the heralded implementation of arbitrary QFT instruments, causal or not.

We assume that the target QFT $\Phi$ is a Klein-Gordon field in an arbitrary hyperbolic spacetime. As stated in section \ref{subsec:aqft}, we only consider representations $\pi_\Phi$ of $\Phi$ that stem, via the GNS construction, from a quasi-free state with distributional covariance on the abstract Weyl algebra. Denoting the representing Hilbert space by $\H_\Phi$ and the representing vector by $\ket{\phi}$, we suppose that the Weyl instruments act on $B(\H_\Phi)$. Note that, by the Reeh-Schlieder property, we have that $\overline{\operatorname{span} e^{i\Phi(f)} \ket{\phi}} = \H_\Phi$. Moreover, we suppose an analogous setup for the probe theory $\Psi$ and its representation. To lighten up the notation, in this section we will omit most representation symbols. That is, by expressions of the form $e^{i\Phi(f)}, e^{i\Psi(f)}$ we will respectively mean $\pi_\Phi(e^{i\Phi(f)})$, $\pi_\Psi(e^{i\Psi(f)})$.

Let $\Omega=(\Omega_a)_{a=1}^A$ be a diagonal Weyl instrument over $\Q$, and let $(Q,Z)$ be a matrix representation thereof such that $Q_j = \Phi(f_j)\in\Q$, $f_j \in C^\infty_0(\M)$ with $\supp f_j \subset \OO$ for all $j=1,\ldots,s$ and
\begin{equation}
\bar{Z}=\sum_aZ_a=\mbox{diag}(p_1,...,p_s), \quad \sum_j p_j = 1, \, p_j \geq 0.   
\end{equation}
To prove that $\Omega$ can be asymptotically realized in the FV framework, we next construct a probe state $\omega'_\lambda$, a probe POVM $(M^\lambda_a)_a$ and a local scattering morphism $\Theta_\lambda$ such that, for any $g\in C_0^\infty(\M)$,
\begin{equation}
    \slim{\lambda \to 0} \epsilon_\lambda\circ\Theta_\lambda\left(M_a^\lambda \otimes e^{i\Phi(g)}\right)=\Omega_a(e^{i\Phi(g)}),
\end{equation}
with $\epsilon_\lambda(A_\Psi\otimes B_\Phi):=\omega_\lambda'(A)\otimes B_\Phi$.

To this aim, let $h_1,...,h_{2n}\in C_0^\infty(\OO)$, such that the vector of quadratures $\vec{R}=(R_1,...,R_{2n})$, with $R_j:=\Psi(h_j)$, forms $n$ independent modes $M_j := (R_{2j-1},R_{2j})$, $j=1,..,n$, and such that $\Q\subset \mbox{span}\{R_j\}_j$.

We next define a projector that only acts non-trivially in those $n$ modes:
\begin{equation}
\Pi:=\frac{1}{(2\pi)^{n}}\int d\vec{\xi} e^{-\frac{\vec{\xi}^2}{4}}e^{-i\vec{\xi}\cdot \vec{R}}.
\end{equation}
It can be verified, by the Weyl relations, that $\Pi^2=\Pi$. In fact, $\Pi$ is a projector onto the all-zeros number state of the modes $(M_j)_j$, i.e.,
\begin{equation}
\Pi=\proj{0}\otimes \id_{\H^\perp}, 
\label{proj}
\end{equation}
where $\H_\Psi^\perp$ denotes the Hilbert space corresponding to the remaining infinitely many modes of the probe QFT and the state $\ket{0}$ satisfies
\begin{equation}
\bra{0}e^{i\vec{\xi}\cdot\vec{R}}\ket{0}=e^{-\frac{\vec{\xi}^2}{4}}, \forall \vec{\xi}\in\R^{2n}.
\label{values_vacuum}
\end{equation}

Let $\omega$ be the (quasi-free) vacuum state of the probe system, i.e.,
\begin{equation}
\omega\left(e^{i\Psi(f)}\right)=e^{-\frac{\Gamma(f,f)}{4}},
\end{equation}
with covariance form $\Gamma$ defined in terms of the anticommutator $\{,\}$ as:
\begin{equation}
\Gamma(f,g)=\langle \{\Psi(f),\Psi(g)\}\rangle
\end{equation}
By the Reeh-Schlieder property, it holds that $\omega(\Pi)\not=0$. Define then the state
\begin{equation}
\omega'(\bullet)=\frac{\omega(\Pi\bullet\Pi)}{\omega(\Pi)}.
\end{equation}
$\omega'$ can be understood as the result of effecting a heterodyne measurement on $\omega$ and post-selecting on the measurement result $\vec{0}$. Let $f,g\in C^\infty_0(\M)$ be such that $[\Psi(f),R_k]=[\Psi(g),R_k]=0$, for all $k$. Standard quantum optics tells us\footnote{See, e.g., eq. (99) of \cite{brask2022gaussianstatesoperations}, noting that the author's definition of covariance matrix is ours divided by two.} that $\omega'$ is also quasi-free, with covariance matrix $\Gamma'$ satisfying:
\begin{align}
&\Gamma'(h_j,h_k) =\delta_{jk},\\
&\Gamma'(\vec{h},f) =\Gamma'(\vec{h},g)=0,\nonumber\\
&\Gamma'(f,g)=\Gamma(f,g)- \Gamma(f,\vec{h})^T(\tilde{\Gamma}+\id)^{-1}\Gamma(g,\vec{h}),
\label{cova_prime}
\end{align}
where $\tilde{\Gamma}$ is the $2n\times 2n$ sub-covariance matrix of $\omega$ for the $n$ modes $\{(R_{2k-1},R_{2k})\}_k$, i.e.,
\begin{equation}
\tilde{\Gamma}_{jk}=\Gamma(h_j,h_k).
\end{equation}

For $\lambda\in\R$ sufficiently small, we choose the initial state $\omega_\lambda'$ of the probe to be
\begin{equation}
\omega'_{\lambda}(\bullet)=\omega'(S_\lambda^* \bullet S_\lambda),\qquad S_\lambda:=\sum_j\sqrt{p_j}e^{i\frac{\Psi(f_j)\cos(\lambda)}{\lambda}}
\end{equation}
Note that $\omega'_\lambda$ is not normalized. Indeed, for $j=1,...,s$, let $Q_j=\vec{f}_j\cdot \vec{R}$. From eq. (\ref{cova_prime}), we have that
\begin{align}
&\omega_\lambda'(1)=\sum_{j,k}\sqrt{p_jp_k}e^{\frac{i}{2}\vec{f}_j\sigma \vec{f}_k}e^{-\frac{(\vec{f}_j-\vec{f}_k)^2}{\lambda^2}}\not=1,
\end{align}
where\footnote{Note that, for $\vec{\alpha},\vec{\beta}\in \R^{2n}$, it holds that $E(\vec{\alpha}\cdot\vec{h},\vec{\beta}\cdot\vec{h})=\vec{\alpha}\cdot\sigma\vec{\beta}$.}
\begin{equation}
\sigma:=\bigoplus_{k=1}^n\left(\begin{array}{cc}0&1\\-1&0\end{array}\right).
\end{equation}

In the limit $\lambda\to 0$, the terms with $j\not=k$ vanish, and we are left with $\sum_j p_j=1$. That is, $\lim_{\lambda\to 0}\omega'_\lambda(1)=1$. Since we are presenting an asymptotic scheme to realize $\Omega$ in the limit $\lambda\to 0$, for simplicity we skip the normalization of $\omega'_\lambda$.

Next, let us define the probe observable. First, we define
\begin{equation}
\tilde{Z}_a(j,k):=\frac{1}{\sqrt{p_jp_k}}Z_a(j,k),    
\end{equation}
and let $\mu^\lambda$ be the largest non-negative real number such that
\begin{equation}
\mu^\lambda \sum_{j=1}^s e^{i\frac{\Psi(f_j)}{\lambda}}\Pi e^{-i\frac{\Psi(f_j)}{\lambda}}\leq \id_\Psi.
\end{equation}
From eqs. (\ref{proj}) and (\ref{values_vacuum}), in the limit $\lambda\to 0$, the projectors $(\Pi_j)_j$, with
\begin{equation}
\Pi_j:=e^{i\frac{\Psi(f_j)}{\lambda}}\Pi e^{-i\frac{\Psi(f_j)}{\lambda}}    
\end{equation}
become orthogonal, which implies that $\lim_{\lambda\to 0}\mu^\lambda =1$. Written as a function of $\tilde{Z},\mu^\lambda$, the POVM we propose to take effect on the probe is:
\begin{align}
&M_a^\lambda=\mu^\lambda\sum_{j,k}\tilde{Z}_a(j,k)e^{i\frac{\Psi(f_j)}{\lambda}}\Pi e^{-i\frac{\Psi(f_k)}{\lambda}}, a=1,...,A;\nonumber\\
&M_{A+1}^\lambda=\id_\Psi-\sum_{a=1}^sM^\lambda_a.
\end{align}
As we will see, in the limit $\lambda\to 0$, the probability that this measurement scheme returns outcome $A+1$ is zero. In this scheme, the processing zone, i.e., the zone where the measurement of the probe takes place, will be $\OO$ itself.

We next detail the interaction between probe and target QFTs. Following the measurement scheme of \cite{fewster2025measurementpreparationprotocolsquantum}, we choose three regions $L,L_-, L_+\subset M$ and two non-intersecting Cauchy surfaces $\Sigma_+, \Sigma_-\subset\M$, with $\Sigma_+\succ \Sigma_-$, such that 
\begin{enumerate}
    \item $L_+, L_-$, are, respectively, in the future of $\Sigma_+$ and past of $\Sigma_-$.
    \item $L_+\cup L_-\subset L\subset D^+(L_-)$.
    
    \item $\OO\subset L_+$.
\end{enumerate}
Call $S$ the spacetime region between surfaces $\Sigma_+,\Sigma_-$. Then, the interaction zone can be any compact $\K$ that contains a neighborhood of $L\cap S$. This implies, in particular, that $\OO\subset \K^+$. A sufficient condition for a compact $\K$ that guarantees the existence of $L_\pm, L, S$ is that $\K=\overline{\K_0}$, for some bounded region $\K_0$ with $\OO\subset D^+(\K_0)\cap \K_0^+$.

We use the same interaction as that described in \cite{fewster2025measurementpreparationprotocolsquantum}\footnote{The notation in \cite{fewster2025measurementpreparationprotocolsquantum} is, however, different: the author calls ``$s$" what we here denote by ``$\lambda$".}, which depends on a function of the form $e^{i\lambda\theta(x)}$, where $\theta(x)$ is a smooth function such that $\theta(x)=0$ for $x\in J^+(\Sigma_+)$, and $\theta(x)=1$, for $x\in J^-(\Sigma_-)$. The interaction term is of the form $\lambda$ `times' $r$, where $r$ is a local term quadratic in the fields and its first differentials, that vanishes outside $\K$ and is analytic in $\lambda$. The interaction thus vanishes for $\lambda\to 0$.

For general $g\in C_0^\infty(\K^+)$, by basic perturbation theory it holds that
\begin{align}
&\Theta_\lambda\left(\Psi(g)\right)=\Psi(g_\lambda)+\lambda\Phi(g_\lambda'),\nonumber\\
&\Theta_\lambda\left(\Phi(g)\right)=\Phi(\tilde{g}_\lambda)+\lambda\Phi(\tilde{g}_\lambda'),
\end{align}
for $g_\lambda, g'_\lambda, \tilde{g}_\lambda, \tilde{g}'_\lambda \in C_0^\infty(\K^+)$ such that the limits $\lim_{\lambda\to 0}$ of $g_\lambda, g'_\lambda, \tilde{g}_\lambda, \tilde{g}'_\lambda$ exist within $C_0^\infty({\cal K}^+)$ and such that $\lim_{\lambda\to 0}g_\lambda = \lim_{\lambda\to 0} \tilde{g}_\lambda=g$.

Remarkably, the results of \cite{fewster2025measurementpreparationprotocolsquantum} imply that 
\begin{align}
&\Theta_\lambda\left(\Psi(f)\right)=\cos(\lambda)\Psi(f)+\sin(\lambda)\Phi(f),\nonumber\\
&\Theta_\lambda\left(\Phi(f)\right)=\cos(\lambda)\Phi(f)-\sin(\lambda)\Psi(f),
\end{align}
for all $f\in C^\infty_0(\OO)$.

Now that the FV scheme has been specified, let us calculate the effective instrument induced on the target system. Through painstaking calculations, it can be shown that
\begin{prop}
\label{prop:lengthy_calculation}
For all $g\in C_0^\infty(\M)$, the identity
\begin{equation}
\slim{\lambda\to 0}\epsilon_\lambda\circ\Theta_\lambda\left(e^{i\frac{\Psi(f_j)}{\lambda}}\Pi e^{-i\frac{\Psi(f_k)}{\lambda}} \otimes e^{i\Phi(g)}\right)=\sqrt{p_jp_k}e^{i\Phi(f_j)}e^{i\Phi(g)}e^{-i\Phi(f_k)}.
\end{equation}
\end{prop}
See Appendix \ref{app:lengthy_calculation} for a proof.

It follows that: 
\begin{align}
&\slim{\lambda\to 0}\epsilon_\lambda\circ\Theta_\lambda\left(M_a^\lambda\otimes e^{i\Phi(g)}\right)=\nonumber\\
&=\slim{\lambda\to 0}\mu^\lambda\sum_{j,k}\tilde{Z}_a(j,k) \epsilon_\lambda\circ\Theta_\lambda\left([e^{i\frac{\Psi(f_j)}{\lambda}}\Pi e^{-i\frac{\Psi(f_k)}{\lambda}}]\otimes e^{i\Phi(g)}\right)\nonumber\\
&=\sum_{j,k}\tilde{Z}_a(j,k)\sqrt{p_jp_k}e^{i\Phi(f_j)}e^{i\Phi(g)}e^{-i\Phi(f_k)} \nonumber\\
&=\sum_{j,k}Z_a(j,k)e^{i\Phi(f_j)}e^{i\Phi(g)}e^{-i\Phi(f_k)}\nonumber\\
&=\Omega_a\left(e^{i\Phi(g)}\right).
\end{align}

The FV scheme detailed above is therefore an asymptotic FV realization of $\Omega$.

The fact that all diagonal Weyl instruments are asymptotically realizable, together with Theorem \ref{theo:diagonal_POVM}, implies that, for any subspace of quadratures $\Q$, any POVM $(M_a)_a \subset \A(\Q)^A$ can be realized within the FV framework in a \emph{non-demolition way}, i.e., through an FV-realizable instrument $\Omega=(\Omega_a)_a$ with the property that
\begin{equation}
\overline{\Omega}(X) = \sum_a\Omega_a(X)=X,
\end{equation}
for any $X$ commuting with $\A(\Q)$. In this sense these instruments are \emph{non-disturbing} outside of $\A(\Q)$. Moreover, from Proposition \ref{prop:approx_instr} and the discussion in \Cref{sec:diagonal}, any Weyl instrument $\Omega$, causal or not, can be realized in a heralded way through an asymptotic FV scheme, and any normal instrument whatsoever can be \emph{approximately} realized probabilistically.

\subsection{Causal Weyl channels without an FV realization}
\label{sec:causal_not_FV}
We have just shown that any diagonal Weyl instrument over $\Q$ can be asymptotically realized within the FV framework. Since for a generic, finite-dimensional subspace of quadratures $\Q$ diagonal Weyl instruments are the only causal Weyl instruments over $\Q$, an optimistic reader might thus conclude that perhaps \emph{all} causal Weyl instruments admit an FV realization. 

In this section we show this not to be the case. In fact, there exists a simple family of causal Weyl channels that cannot be approximated arbitrarily well through an FV scheme, no matter how large we take its interaction zone $\OO$ to be.

Our proof exploits an analogy between QFT and quantum information science, firstly pointed out by Gisin and del Santo \cite{Gisin_2024}. Consider the set of bipartite quantum instruments, i.e., instruments of the form $(\Omega_{ab})_{ab}$, with $\Omega_{ab}:B(\H_A)\otimes B(\H_B)\to B(\H_A)\otimes B(\H_B)$, mapping a bipartite quantum system into itself. If $\Omega$ is to be implemented by two separate parties Alice and Bob, respectively acting on the Hilbert spaces $\H_A$, $\H_B$, and obtaining the results $a,b$, then $\Omega$ must be of the form:
\begin{equation}
\Omega_{ab}(\bullet)=(\bra{\psi}_{A'B'}\otimes\id_{AB})\xi_a^{AA'}\otimes_{2,4} \eta_b^{BB'}(\id_{A'B'}\otimes\bullet_{AB})(\ket{\psi}_{A'B'}\otimes\id_{AB}),
\end{equation}
where $A'$, $B'$ are ancillary systems under the respective control of Alice and Bob; $\ket{\psi}_{A'B'}$ is the (in general, entangled) initial state of systems $A'B'$; $(\xi^{AA'}_a)_a$, $(\eta^{BB'}_b)_b$ are instruments respectively acting on $B(\H_{A'}\otimes\H_A)$, $B(\H_{B'}\otimes\H_B)$; and $(A' \otimes A) \otimes_{2,4} (B' \otimes B) = A' \otimes B' \otimes A \otimes B$. In the following, we call instruments of the form $\Omega$ \emph{entanglement-assisted bipartite}\footnote{In the quantum information literature, they are actually called \emph{localizable} \cite{Gisin_2024}, \cite{causal_vs_localizable}, but that name collides with standard QFT terminology.}. It is a celebrated result \cite{Vaidman} that, for any bipartite POVM $(M^{AB}_c)_c \subset B(\H_A\otimes\H_B)$, there exists an entanglement-assisted bipartite instrument $(\Omega_{ab})_{ab}$ and a processing function $f$ such that
\begin{equation}
M^{AB}_c=\sum_{a,b:f(a,b)=c}\Omega_{ab}(1).
\end{equation}
In other words: any bipartite POVM can be implemented by realizing an entanglement-assisted instrument and postprocessing its local outcomes $a,b$. The resulting instrument $(\Omega_c)_c$ is said to admit an entanglement-assisted realization.

In this quantum information scenario, the analog of Sorkin's paradox arises when we in addition expect an entanglement-assisted realization to also reproduce Lüder's rule, i.e., when we demand that $(\Omega_c)_c$ act as:
\begin{equation}
\Omega_c(\bullet)=\sqrt{M^{AB}_c}\bullet \sqrt{M^{AB}_c}
\label{instr_cenutrio}
\end{equation}
While the instrument above satisfies $\Omega_c(1)=M^{AB}_c$, in general it does not admit an entanglement-assisted realization. In fact, in most cases an instrument of the form (\ref{instr_cenutrio}) will be signalling, meaning that there exists an initial state $\rho_{AB}$, two channels $\{D^\alpha:B(\H_A)\to B(\H_A)\}_{\alpha=0,1}$ and a POVM $(N_b)_b \in B(\H_B)^M$ such that the distribution of Bob's outcome $b$
\begin{equation}
P(b|\alpha)=\tr\left\{\rho_{AB} \left(D^\alpha\otimes\id_B\right)\circ \bigg( \sum_c \Omega_c \bigg) (\id_A\otimes N_b) \right\}
\end{equation}
depends on $\alpha$. If Alice implements operation $D^\alpha$ before the instrument $\Omega$ acts upon the joint system $\H_{A}\otimes\H_{B}$, she is therefore able to transmit Bob information about her bit choice $\alpha$. If $\Omega$ can be implemented faster than the time it takes for light to travel from Alice's to Bob's lab, then Einstein's causality is violated.

The authors of \cite{Gisin_2024} interpreted the FV framework \cite{local_meas_QFT} as a restriction on the set of QFT instruments, which in a quantum information scenario would correspond to demanding bipartite instruments to be entanglement-assisted. It is easy to see that entanglement-assisted instruments preclude signalling, and thus the quantum information analog of the Sorkin paradox is avoided.

Let us next focus on channels, i.e., instruments with just one outcome. Quantum information theory has studied the relation between entanglement-assisted bipartite channels and non-signalling bipartite channels, i.e., bipartite channels $\Omega:B(\H_{A})\otimes B(\H_B)\to B(\H_{A})\otimes B(\H_B)$ with the property that
\begin{align}
&\Omega(\id_A \otimes\bullet)=\id_A\otimes\Omega^B(\bullet),\nonumber\\
&\Omega(\bullet\otimes\id_B)=\Omega^A(\bullet)\otimes \id_B,
\end{align}
for some channels $\Omega^A$, $\Omega^B$. Since such channels do not violate causality, in principle they might accept an entanglement-assisted realization. 

However, in general they do not: while all entanglement-assisted channels are non-signalling, the converse is not true. This is shown in \cite{causal_vs_localizable}, through the following argument. Fix a Bell scenario $(M,M)$, let $P\in C_{\mathrm{ns}}$, $P\not\in C_{\mathrm{qa}}$, and consider the channel $\Omega:B(\C^M)^{\otimes 2}\to B(\C^M)^{\otimes 2}$ given by:
\begin{equation}
\Omega(\bullet)=\sum_{a,b,\alpha,\beta=0}^{M-1}P(a,b|\alpha,\beta)\ket{\alpha}\bra{a}\otimes\ket{\beta}\bra{b}\bullet    \ket{a}\bra{\alpha}\otimes\ket{b}\bra{\beta},
\end{equation}
where $\{\ket{n}\}_{n=0}^{M-1}$ is an orthonormal basis for $\C^M$. It can be verified that this channel is non-signalling. However, it cannot be realized through entanglement-assisted operations. Indeed, let $\rho_{AB}=\proj{0}^{\otimes 2}$ be the initial quantum state and define the unitary $U\ket{n}=\ket{n+1\mbox{ (mod } M\mbox{)}}$. Then we have that
\begin{equation}
\tr\left(\rho_{AB} U^{-\alpha}\otimes U^{-\beta}\Omega(\proj{a}\otimes\proj{b})U^{\alpha}\otimes U^{\beta}\right)=P(a,b|\alpha,\beta).
\end{equation}
In words: by carrying out local operations $U^\alpha, U^\beta$ before the action of channel $\Omega$, and then conducting local measurements $(\proj{a})_a$, $(\proj{b})_b$ immediately afterwards, Alice and Bob are able to generate a distribution $P\not\in C_{\mathrm{qa}}$. It is easy to prove, though, that, with local operations and local measurements the only distributions Alice and Bob can generate with an entanglement-assisted bipartite channel are in $C_{\mathrm{qs}}$. We arrive at a contradiction.

Now, replace ``entanglement-assisted bipartite channel'' by ``FV-realizable Weyl channel'' and ``non-signalling channel'' by ``causal Weyl channel'' and we are back to QFT. To prove the existence of causal Weyl channels that do not have an FV realization, we will thus mimic the quantum information argument. Namely, first we will show, in \Cref{sec:nl_conn}, that Bell experiments involving an FV-realizable channel can only give rise to correlations in $C_{\mathrm{qc}}$. Next, we will introduce, in \Cref{sec:box_channels}, the notion of box channels: a family of causal Weyl channels that are in one-to-one correspondence with distributions $P\in C_{\mathrm{ns}}$. We will further show that specific Bell experiments can generate the hidden distribution $P$ defining the box channel. Finally, in \Cref{subsec:unrealizability} we will define a class of box channels for which such $P$-generating Bell experiments can be realized arbitrarily far away from the zone where the Kraus operators of the channel are localized. Thus, any (causal) box channel with distribution $P\not\in C_{\mathrm{qc}}$ will be unrealizable within the FV framework (evne asymptotically), independently of how large we take its interaction zone to be.

\subsubsection{Connection with quantum nonlocality}
\label{sec:nl_conn}
We next model a Bell experiment around the action of a channel with interaction zone $\OO$. 

\begin{defin}[Bell setup]
    Given a subset $\OO \subset \M$, we say that regions $\RR_\pm^X$, $\U^X\subset \M$, for $X=A,B$, form a \emph{Bell setup around} $\OO$ if they satisfy the Bell conditions:
    \begin{align}
    &\RR_\pm^X\subset \OO^\pm, \overline{\RR^X_+}\subset D^+(\RR^X_-),\U^X\subset \RR_{-}^X,\mbox{ for }X=A,B,\nonumber\\
    &\overline{\RR^A_+}\perp \overline{\RR^B_+},\nonumber\\
    &\overline{\U^A}\perp \overline{\RR^B_-},\overline{\U^B}\perp \overline{\RR^A_-},
    \label{conds_Bell_QFT}
    \end{align}
    We say that the Bell setup is \emph{spacelike} if, in addition, 
    \begin{equation}\label{cond_split}
        \overline{\RR^A_-} \perp \overline{\RR^B_-}.
    \end{equation}
\end{defin}
The reader can see these relations depicted in \Cref{fig:Bell_exp_QFT}. Note that they also imply $\overline{\U^A}\perp \overline{\U^B}$.
\begin{figure}[H]
\centering
\includegraphics[width=0.7\textwidth]{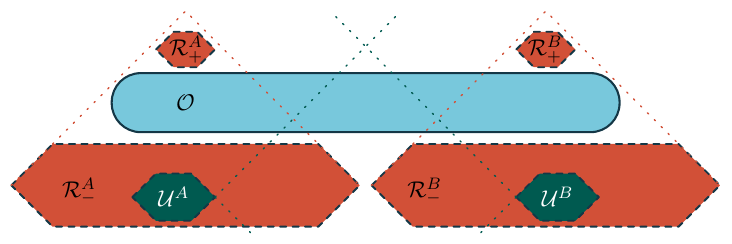}
\caption{A spacelike Bell setup around $\OO$.}
\label{fig:Bell_exp_QFT}
\end{figure}

Intuitively, in the Bell experiment Alice (Bob) will implement an operation labeled by $\alpha$ ($\beta$), local to region $\U^A$ ($\U^B$), and then conduct a measurement localized in $\RR^A_+$ ($\RR^B_+$), obtaining an outcome $\alpha$ ($\beta$). We next prove that, if $\Omega$ admits an asymptotic FV realization, the resulting distribution $P(a,b|\alpha,\beta)$ must be quantum.
\begin{prop}
\label{prop:FV_box}
Let $\Omega$ admit an FV realization in the interaction zone $\OO$ and let the regions $\RR_\pm^X, \U^X$, with $X=A,B$, form a Bell setup around $\OO$. Let $(M^A_a)_{a=0}^{M-1}\in B(\H)^M$, $(M^B_b)_{b=0}^{M-1}\in B(\H)^M$ be POVMs respectively localizable in $\RR^A_+$, $\RR^B_+$, and let $(\Lambda^{A,\alpha})_{\alpha=0}^{N-1}$ $(\Lambda^{B,\beta})_{\beta=0}^{N-1}$ be normal channels respectively local to $\U^A, \U^B$. Then, for any QFT state $\omega$, the distribution $P(a,b|\alpha,\beta)$ defined by
\begin{equation}
P(a,b|\alpha,\beta):= \omega\circ\Lambda^{A,\alpha}\circ\Lambda^{B,\beta}\circ \Omega(M^A_aM^B_b)
\label{Bell_corrs}
\end{equation}
satisfies $P\in C_{\mathrm{qc}}$. The result also holds for $\Omega$ asymptotically FV-realizable if the POVM elements of $M^A, M^B$ live in $\pi(\hat{\W})$.
\end{prop}
\begin{proof}
Suppose that $\Omega$ is FV realizable. Then, there exists a probe theory $\A_P$, initially in state $\varphi$, and a scattering morphism $\Theta:\A_S\otimes\A_P\to \A_S\otimes\A_P$ describing the interaction between probe and target fields. Since $\Omega$ describes a channel, we do not need to consider any measurement on the probe. Now, let $(A^\alpha_j)_j$ ($(B^\beta_k)_k$) be the Kraus operators of instrument $\Lambda^{A,\alpha}$ ($\Lambda^{B,\beta}$) and define the instruments $\tilde{\Lambda}^{A,\alpha}, \tilde{\Lambda}^{B,\beta}$ through 
\begin{align}
&\tilde{\Lambda}^{A,\alpha}(s\otimes t)=\Lambda^{A,\alpha}(s)\otimes t,\nonumber\\
&\tilde{\Lambda}^{B,\beta}(s\otimes t)=\Lambda^{B,\beta}(s)\otimes t.
\end{align}
It is immediate that the Kraus operators of these instruments are, respectively,
\begin{equation}
\tilde{A}^\alpha_j = A^\alpha_j \otimes \id_P, \tilde{B}^\beta_k = B^\beta_k \otimes \id_P.
\end{equation}

With this notation, we have that
\begin{align}\label{bellboxcorrelations}
P(a,b|\alpha,\beta)&=\omega\otimes\varphi\left(\tilde\Lambda^{A,\alpha}\circ\tilde\Lambda^{B,\beta}\circ\Theta\left(M^A_aM^B_b\otimes\id_P\right)\right),\nonumber\\
&=\omega\otimes\varphi\left(\tilde\Lambda^{A,\alpha}\circ\tilde\Lambda^{B,\beta}\left(\tilde{M}^A_a\tilde{M}^B_b\right)\right),\nonumber\\
&=\sum_{j,k}\omega\otimes\varphi\left(\tilde{A}^\alpha_j\tilde{B}^\beta_k\tilde{M}^A_a\tilde{M}^B_b(\tilde{B}^\beta_k)^*(\tilde{A}^\alpha_j)^*\right)
\end{align}
where we defined the POVMs
\begin{equation}
\tilde{M}^X_x:=\Theta\left(M^X_x\otimes\id_P\right), X=A,B, \, x=a,b.
\end{equation}
For $X=A,B$, due to the condition $\overline{\RR^X_+}\subset D^+(\RR_-^X)$, we have that $\tilde{M}^X_x$ is localizable in $\RR_-^X$. From the conditions $\overline{\U^A}\perp \overline{\U^B}$, $\overline{\U^B}\perp \overline{\RR^A_{-}}$, $\overline{\U^A}\perp \overline{\RR^B_{-}}$ and Lemma \ref{lemma:weyllocality}, we further infer that
\begin{equation}
[\tilde{A}, \tilde{B}]=[\tilde{B}, \tilde{M}^A_a]=[\tilde{A}, \tilde{M}^B_b]=0,
\end{equation}
where $\tilde{A}$ stands for $\tilde{A}_j^\alpha$ and $(\tilde{A}_j^\alpha)^*$; and $\tilde{B}$, for $\tilde{B}_k^\beta$ and $(\tilde{B}_k^\beta)^*$. In addition, since $\Theta$ is an isomorphism, we have that $[\tilde{M}^A_a,\tilde{M}^B_b]=[M^A_a,M^B_b]=0$. These commutation relations allow us to express $P(a,b|\alpha,\beta)$ as
\begin{equation}
P(a,b|\alpha,\beta)=\omega\otimes \varphi\left(\tilde{M}^A_{a|\alpha}\tilde{M}^B_{b|\beta}\right),
\end{equation}
with 
\begin{equation}
\tilde{M}^A_{a|\alpha}=\tilde\Lambda^{A,\alpha}\left(\tilde{M}^A_a\right),\tilde{M}^B_{b|\beta}=\tilde\Lambda^{B,\beta}\left(\tilde{M}^B_b\right),
\end{equation}
Noting that $(\tilde{M}^A_{a|\alpha})_a$ $(\tilde{M}^B_{b|\beta})_b$ are POVMs for all $\alpha,\beta$ and that
$[\tilde{M}^A_{a|\alpha},\tilde{M}^B_{b|\beta}]=0$, we have that $P(a,b|\alpha,\beta)\in C_{\mathrm{qc}}$. Since $C_{\mathrm{qc}}$ is closed, if the POVM elements of $M^A, M^B$ belong to $\pi(\hat{\W})$, the statement of the Proposition also holds for asymptotic FV realizations.    
\end{proof}

\subsubsection{Box channels}
\label{sec:box_channels}
For arbitrary $f\in C_0^\infty(\M)$, $n\in \N$, we define the Weyl instruments:
\begin{align}
&D^{f}(\bullet):=e^{i\Phi(f)}\bullet e^{-i\Phi(f)},\nonumber\\
&\Omega^{f,n}_\gamma(\bullet):=(\Lambda^{f,n}_\gamma)^*\bullet \Lambda^{f,n}_\gamma,\mbox{ for } \gamma=0,...,n-1,
\end{align}
where
\begin{equation}
\Lambda^{f,n}_\gamma:=\frac{1}{n}\sum_{k=0}^{n-1}e^{i\frac{2\pi k(\gamma-\Phi(f))}{n}}.
\label{kraus_disc_meas}
\end{equation}

Note that $D^{f,n}$ and $\Omega^{f,n}$ are diagonal Weyl instruments, and thus causal, for interaction zone $\OO=\supp f$. That $\Omega^{f,n}$ is a diagonal is a direct consequence of the identity:
\begin{equation}\label{eq:kronecker}
    \frac{1}{n} \sum_{\gamma=0}^{n-1} e^{i\frac{2\pi(l-k)\gamma}{n}} = \delta_{kl},\forall k,l\in\{0,...,n-1\}.
\end{equation}
More explicitly, 
\begin{align}
\sum_{\gamma=0}^{n-1}\Omega^{f,n}_\gamma(\bullet)&=\sum_{\gamma=0}^{n-1}(\Lambda^{f,n}_\gamma)^*\bullet \Lambda^{f,n}_\gamma\nonumber\\
&=\frac{1}{n^2}\sum_{j,k=0}^{n-1}e^{i\frac{2\pi j\Phi(f)}{n}}\bullet e^{-i\frac{2\pi k\Phi(f)}{n}}\left(\sum_{\gamma=0}^{n-1}e^{i2\pi\gamma(k-j)}\right)\nonumber\\
&=\frac{1}{n}\sum_{j,k=0}^{n-1}e^{i\frac{2\pi j\Phi(f)}{n}}\bullet e^{-i\frac{2\pi k\Phi(f)}{n}}\delta_{j,k}\nonumber\\
&=\frac{1}{n}\sum_{j=0}^{n-1}e^{i\frac{2\pi j\Phi(f)}{n}}\bullet e^{-i\frac{2\pi j\Phi(f)}{n}}.
\end{align}

For QFT states $\omega$ satisfying
\begin{equation}
\omega((\phi(f)-\tilde{\gamma})^2)\ll 1,
\end{equation}
for some $\tilde{\gamma},\in \{0,...,n-1\}$, it is the case that
\begin{equation}
\omega\circ\Omega^{f,n}_\gamma(\bullet) \approx \omega(\bullet)\delta_{\gamma,\tilde{\gamma}}.
\label{meas_rel}
\end{equation}
Indeed, for such states $\omega$ the distribution $\omega(d\Pi(\Phi(f)))$ is highly peaked at $\tilde{\gamma}$. Thus, the right action of the Kraus operator $\Lambda^{f,n}_\gamma$, as defined in (\ref{kraus_disc_meas}), is of the form
\begin{equation}
\omega(\bullet \Lambda^{f,n}_\gamma)\approx \omega(\bullet)\frac{1}{n}\sum_{k=0}^{n-1}e^{i\frac{2\pi k(\gamma-\tilde{\gamma})}{n}}=\omega(\bullet)\delta_{\gamma,\tilde{\gamma}},
\end{equation}
and similarly for the left action of $(\Lambda^{f,n}_\gamma)^*$. Now, let $\hat{f}\in C^\infty_0(\M)$, with
\begin{equation}
[\Phi(\hat{f}),\Phi(f)]=i,    
\end{equation}
and let $\omega$ be a QFT state with
\begin{equation}
\omega(\Phi(\hat{f})^2)\ll 1.
\end{equation}
Then, for $\gamma\in\{0,...,n\}$, the state 
\begin{equation}
\omega^\gamma(\bullet)=\omega\circ D^{\gamma f}(\bullet)    
\end{equation}
will satisfy
\begin{equation}
\omega^\gamma((\Phi(\hat{f})-\gamma)^2)\ll 1.
\end{equation}
From eq. (\ref{meas_rel}) and the above equation, it follows that
\begin{equation}
\omega\circ D^{\tilde{\gamma}f}\circ \Omega^{\hat{f},n}_\gamma (\bullet) \approx \omega\circ D^{\gamma f}(\bullet) \delta_{\gamma,\tilde{\gamma}}.
\label{comp_displ_meas_rel}
\end{equation}

Next, we define another ingredient:
\begin{defin}[Box setup]
We say that the compacts $\S^X,\T^X \subset \M$ and test functions $f_X,g_X$, for $X=A,B$, form a \emph{box setup} if they satisfy the conditions:
\begin{align}
&\supp f_X \subset \S^X, \supp g_X \subset \T^X, \nonumber\\
&\S^X\subset (\T^X)^\vee,\nonumber\\
&\S^A\perp \S^B,\T^A\perp \T^B.
\label{conds_box_channel}
\end{align}
We call it \emph{spacelike} if, in addition, 
\begin{equation}
\T^A\perp \S^B, \T^B\perp \S^A.
\label{space-like_conds}
\end{equation}
\end{defin}
An example of a spacelike box setup is depicted in \Cref{fig:box_channel_regions}.
\begin{figure}
\centering
\includegraphics[width=0.6\textwidth]{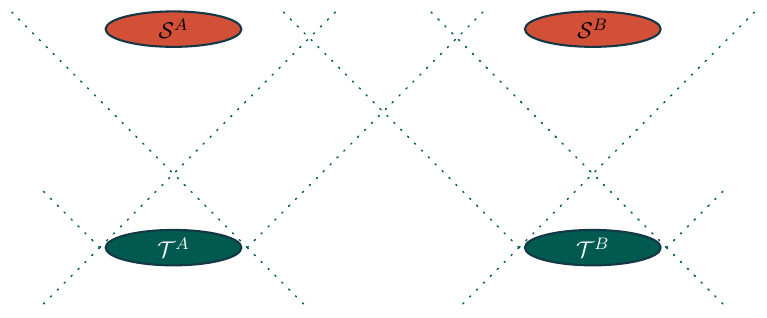}
\caption{Spacetime zones for a spacelike box channel.}
\label{fig:box_channel_regions}
\end{figure}

Now, we define a class of channels that we will be playing with within this section and the next.
\begin{defin}
For $X=A,B$ consider compacts $\S^X, \T^X\subset \M$ and test functions $f_X, g_X\in C_0^\infty(\M)$ forming a box setup.
Let $(M,N)$ be a Bell scenario, and let $P(a,b|\alpha,\beta)\in C_{\mathrm{ns}}$ be a correlation therein. The Weyl channel defined by
\begin{equation}
\Omega:=\sum_{a,b,\alpha,\beta}P(a,b|\alpha,\beta)\Omega^A_{\alpha}\circ\Omega^B_{\beta}\circ D^{a f_A}\circ D^{b f_B},
\label{counterexample_channel}
\end{equation}
with $\Omega^A:=\Omega^{g_A,N}$, $\Omega^B:=\Omega^{g_B,N}$, will be called a \emph{box channel with distribution $P$}.
\end{defin}

\begin{remark}
    Box channels can be seen as a generalization of the non-diagonal example given in \eqref{eq:non_diag_ex}. Specifically, in the case where $P(a,b|\alpha,\beta)=\delta_{a\alpha}\delta_{b\beta}$ and the box setup is spacelike, the associated box channel with distribution $P$ reduces to the composition of two mutually commuting instruments
    \begin{equation}
        \Omega'^A=\sum_a \Omega^A_a\circ D^{af_A}, \;\; \Omega'^B=\sum_b \Omega^B_b\circ D^{bf_B},
    \end{equation}
    of the form given in eq  \eqref{eq:non_diag_ex}. As in that example, for some choices of $f_A, f_B, g_A, g_B$ these are non-diagonal. 
\end{remark}

The next proposition states that all box channels are causal.
\begin{prop}
\label{prop:causal_box}
Any box channel is causal for any interaction zone $\OO$ containing $\S^A, \S^B, \T^A, \T^B$, no matter the choice of $f_A, f_B, g_A, g_B$.
\end{prop}

\begin{proof}
Let $\RR_{\pm}$ be regions in $\OO^\pm$, respectively, with $\overline{\RR_{+}}\subset D^+(\RR_{-})$, and choose an arbitrary $h\in C_0^\infty(\RR^+)$. We distinguish four cases:
\begin{enumerate}
    \item $\overline{\RR_{+}}\cap (\T^A)^\vee\not=\emptyset,\overline{\RR_+}\cap (\T^B)^\vee\not=\emptyset$. In that case, for any $x\in \T^A$, there is $y\in \overline{\RR_{+}}$ with a causal, past-directed curve starting in $y$ and ending in $x$. Continuing the curve, we would arrive at $\RR_{-}$, due to $\overline{\RR_{+}}\subset D^+(\RR_{-})$, $\RR_-\subset\OO^-$, which implies that $\T^A\subset D^+(\RR_{-})$, and so $\Phi(g_A)$ can be localized in $\RR_{-}$. Similarly, $\Phi(g_B)$ can be localized in $\RR_{-}$. Moreover,
    $O_{ab}:=D^{a f_A}\circ D^{b f_B}(e^{i\Phi(h)})$ is proportional to $e^{i\Phi(h)}$ and thus localizable in $\overline{\RR_{+}}$ and, by $\overline{\RR_{+}}\subset D^+(\RR_{-})$, also in $\RR_{-}$. The effect of the channel $\Omega^A_{\alpha}\circ\Omega^B_{\beta}$ on $O_{ab}$ is to multiply it by elements localizable in $\RR_{-}$. We conclude that
    \begin{equation}
    \Omega\left(e^{i\Phi(h)}\right)=\sum_{a,b,\alpha,\beta}P(a,b|\alpha,\beta)\Omega^A_{\alpha}\circ\Omega^B_{\beta}(O_{ab})
    \end{equation}
    is localizable in $\RR_{-}$.
    
    \item $\overline{\RR_{+}}\cap (\T^A)^\vee \not=\emptyset,\overline{\RR_{+}}\cap (\T^B)^\vee=\emptyset$. In that case, $\overline{\RR_{+}}\perp \S^B$, and so $[\Phi(f_B),\Phi(h)]=0$. Thus,
    \begin{align}
    &\Omega(e^{i\Phi(h)})=\sum_{a,b,\alpha,\beta}P(a,b|\alpha,\beta)\Omega^A_{\alpha}\circ\Omega^B_{\beta}\circ D^{a f_A}
    (e^{i\Phi(h)})\nonumber\\
    &=\sum_{a,\alpha}P(a|\alpha)\Omega^A_{\alpha}\circ\sum_\beta\Omega^B_{\beta}\circ D^{a f_A}(e^{i\Phi(h)})\nonumber\\
    &=\sum_\beta\Omega^B_{\beta}\circ\sum_{a,\alpha}P(a|\alpha)\Omega^A_{\alpha}\circ D^{a f_A}(e^{i\Phi(h)})\nonumber\\
    &=\overline{\Omega}^B\circ\sum_{a,\alpha}P(a|\alpha)\Omega^A_{\alpha}\circ D^{a f_A}(e^{i\Phi(h)}),
    \end{align}
    where we made use of $\T^A\perp\T^B$ to argue that we can permute the order of $\Omega^A,\Omega^B$ and of the no-signalling condition $\sum_{b}P(a,b|\alpha,\beta)=P(a|\alpha)$, independent of $\beta$. Define
    \begin{equation}
    \tilde{\Omega}^A(\bullet):=\sum_{a,\alpha}P(a|\alpha)\Omega^A_{\alpha}\circ D^{a f_A}(\bullet).
    \end{equation}
    This is the channel that results from feed-forwarding the outcome $a$ of the diagonal Weyl instrument
    \begin{equation}
    \hat{\Omega}^A_a:=\sum_{\alpha}P(a|\alpha)\Omega^A_\alpha
    \end{equation}
    to the channel $D^{a f_A}$. It follows that $\tilde{\Omega}^A_a$ is a causal channel with interaction zone ${\cal S}^A\cup {\cal T}^A$, and thus $\sum_a\tilde{\Omega}_a^A(e^{i\Phi(h)})\in\A(\RR_-)$. Hence, $\Omega(e^{i\Phi(h)})=\overline{\Omega}^B\circ \overline{\tilde{\Omega}}^A(e^{i\Phi(h)})\in \A(\RR_-)$, since $\Omega^B$ is diagonal.

    \item $\overline{\RR_{+}}\cap (\T^A)^\vee=\emptyset,\overline{\RR_{+}}\cap (\T^B)^\vee \not=\emptyset$. This case is symmetric to the previous one, with the exchange $A\leftrightarrow B$, so $\Omega(e^{i\Phi(h)})$ is also localizable in $\RR_{-}$.
    \item $\overline{\RR_{+}}\cap (\T^A)^\vee=\overline{\RR_{+}}\cap (\T^B)^\vee=\emptyset$. In this case, $[\Phi(h), \Phi(f_A)]=[\Phi(h), \Phi(f_B)]=0$, and so
    \begin{align}
    &\Omega(e^{i\Phi(h)})=\sum_{a,b,\alpha,\beta}P(a,b|\alpha,\beta)\Omega^A_{\alpha}\circ\Omega^B_{\beta}(e^{i\Phi(h)})\nonumber\\
    &=\overline{\Omega}^A\circ\overline{\Omega}^B(e^{i\Phi(h)}).
    \end{align}
    Since both $\Omega_A,\Omega_B$ are diagonal, their composition is also diagonal, and so $\Omega(e^{i\Phi(h)})$ is localizable in $\RR_{+}$ and thus in $\RR_{-}$.
\end{enumerate}
    
\end{proof}

We next prove an important property of box channels: under very mild conditions, their defining distribution $P$ can be observed in a Bell experiment. Such conditions are expressed as follows:
\begin{defin}[Bell-box compatibility]
    A Bell setup with regions $\RR_\pm^X$, $\U^X$, $X=A,B$, is called \emph{compatible} with a box setup (or channel) formed from test functions $f_X$,$g_X$, $X=A,B$, if the following symplectic condition holds:
    \begin{align}
    &\mbox{For }X=A,B,\exists\hat{f}_X\in C_0^\infty(\RR^X_{+}),\hat{g}_X\in C_0^\infty(\U^X),\mbox{ such that}\nonumber\\
    &\EE{\hat{f}_X}{f_X} = \EE{g_X}{\hat{g}_X} = 1, \nonumber\\
    &\EE{\hat{f}_X}{g_X} = \EE{\hat{f}_X}{\hat{g}_X}=0,\nonumber\\
    &\EE{s_A}{t_B} = 0,\mbox{ for }s,t\in\{f,\hat{f},g,\hat{g}\}.
    \label{sympl_cond}
    \end{align}
\end{defin}
\begin{prop}
\label{prop:recovery_P}
Let $\Omega$ be a box channel with distribution $P$, interaction zone $\OO$ and test functions $f_X$, $g_X$, for $X=A,B$ and let the regions $\RR_\pm^X$, $\U^X$ form a compatible Bell setup around $\OO$. 

Then, there exists a family of states $\omega_\lambda$, Weyl POVMs $M^A, M^B$, respectively localized in $\RR_+^A,\RR^B_+$, and displacement channels $\{D^{A,\alpha}\}_{\alpha}$, $\{D^{B,\beta}\}_{\beta}$, respectively localized in $\U^A, \U^B$, such that
\begin{equation}
P(a,b|\alpha,\beta)= \lim_{\lambda\to 0}\omega_\lambda\circ D^{A,\alpha}\circ D^{B,\beta}\circ \Omega(M^A_aM^B_b).
\label{final_distr}
\end{equation}    
\end{prop}
The conditions in the statement of the proposition are not void: in \Cref{subsec:unrealizability}, we show that typical scalar QFTs realize them all.
\begin{proof}
By the symplectic condition (\ref{sympl_cond}), the quadratures $\Phi(\hat{f}_A), \Phi(g_A), \Phi(\hat{f}_B), \Phi(g_B)$ all commute. This implies that, for any $\lambda>0$, there exists a quasifree state $\omega_\lambda$ such that
\begin{equation}
\omega_\lambda(\Phi(g_X)^2),\omega_\lambda(\Phi(\hat{f}_X)^2)\leq\lambda^2,X=A,B,
\end{equation}
and such that $\omega_\lambda(Y)$ remains bounded in $\lambda$ for all quadratures $Y$ commuting with $\Phi(g_X)$ and $\Phi(\hat{f}_X)$, $X=A,B$; see, e.g., \cite{field_meas} to learn how to construct such states by applying squeezing operations to an initial, arbitrary quasi-free state.

Define next the POVMs
\begin{align}
&(M^A_a)_{a=0}^{M-1}:=(\Omega^{\hat{f}_A, M}_a(1))_{a=0}^{M-1}\nonumber\\
&(M^B_b)_{b=0}^{M-1}:=(\Omega^{\hat{f}_B, M}_b(1))_{b=0}^{M-1}.
\label{POVMs_Bell}
\end{align}
and the displacement channels
\begin{align}
&D^{A,\alpha}:=D^{\alpha\hat{g}_A},\alpha=0,...,N-1,\nonumber\\
&D^{B,\beta}:=D^{\beta\hat{g}_B},\beta=0,...,N-1,
\label{disps_Bell}
\end{align}
which clearly satisfy the localization conditions in the statement of the proposition.
The resulting distribution is:
\begin{equation}
P_\lambda(a,b|\alpha,\beta):=\omega_\lambda\circ D^{A,\alpha} \circ D^{B,\beta} \circ \Omega\left(M^A_aM^B_b\right),
\label{approx_box}
\end{equation}
for $a,b\in\{0,...,M-1\}$, $\alpha,\beta\in\{0,...,N-1\}$.

To complete the proof, we just need to show that
\begin{align}
&\lim_{\lambda \to 0} \omega_\lambda \left( D^{\alpha\hat{g}_A} \circ D^{\beta\hat{g}_B}\circ  \Omega^{g_A}_{\alpha'}\circ  \Omega^{g_B}_{\beta'}\circ  D^{a'f_A}\circ  D^{b'f_B}\circ \Omega_a^{\hat{f}_A} \circ \Omega_b^{\hat{f}_B} (1) \right)\nonumber\\
&=\lim_{\lambda \to 0} \omega_\lambda \left( D^{\alpha\hat{g}_A}\circ  \Omega^{g_A}_{\alpha'}\circ D^{\beta\hat{g}_B}\circ  \Omega^{g_B}_{\beta'}\circ  D^{a'f_A}\circ \Omega_a^{\hat{f}_A}\circ  D^{b'f_B}\circ \Omega_b^{\hat{f}_B} (1) \right)\nonumber\\
&= \delta_{aa'} \delta_{bb'} \delta_{\alpha\alpha'} \delta_{\beta\beta'},
    \label{boxdelta}
\end{align}
as this will imply 
\begin{equation}
\lim_{\lambda\to 0}P_\lambda(a,b|\alpha,\beta)=P(a,b|\alpha,\beta).
\label{limit_box}
\end{equation}

For $\lambda\ll 1$, eq. (\ref{boxdelta}) follows from the recursive application of relation (\ref{comp_displ_meas_rel}), together with the commutation relations (\ref{sympl_cond}) and the easily verifiable fact that, for any $f,g\in C_0^\infty(\M)$, $D^{f}\circ D^{g}=D^{g}\circ D^{f}$. Indeed, invoking eq. (\ref{comp_displ_meas_rel}) twice, the argument of the limit on the second line of eq. (\ref{boxdelta}) can be approximated by
\begin{equation}
\omega_\lambda \left( D^{\alpha\hat{g}_A}\circ  D^{\beta\hat{g}_B}\circ D^{a'f_A}\circ \Omega_a^{\hat{f}_A}\circ  D^{b'f_B}\circ \Omega_b^{\hat{f}_B} (1) \right)\delta_{\alpha,\alpha'}\delta_{\beta,\beta'}.
\end{equation}
Due to eq. (\ref{sympl_cond}) and the commutation relations between the displacement channels, we can permute $D^{a'f_A}\circ \Omega_a^{\hat{f}_A}\circ  D^{b'f_B}\circ \Omega_b^{\hat{f}_B}$ by $D^{\alpha\hat{g}_A}\circ  D^{\beta\hat{g}_B}$, which, applying (\ref{comp_displ_meas_rel}) twice more, leaves us with the two remaining Kronecker deltas and four displacement channels acting on the identity (and thus returning $1$).

The proposition is proven.

\end{proof}

\subsubsection{FV-unrealizability for arbitrary interaction zones}\label{subsec:unrealizability}
Having reached this point, it is easy to provide an example of a causal Weyl channel that does not admit an FV realization. Let $\OO \subset \M$ be a compact and consider zones, regions and test functions $\S^X,\T^X,\RR_{\pm}^X,\U^X, f_X,g_X$, $X = A,B$ such that $\S^X,\T^X,f_X,g_X$ form a box setup contained in $\OO$ and $\RR_\pm^X,\U^X$ form a compatible Bell setup around $\OO$. If we choose $P\in C_{\mathrm{ns}}, P\not\in C_{\mathrm{qc}}$ and consider the corresponding box channel $\Omega_P$ (using the box setup $\S^X,\T^X,f_X,g_X$), then, by Proposition \ref{prop:causal_box}, $\Omega_P$ is causal. Now, by Proposition \ref{prop:recovery_P} there exists a family of Bell experiments that generates distributions arbitrarily close to $P\not\in C_{\mathrm{qc}}$. This implies, by Proposition \ref{prop:FV_box}, that $\Omega_P$ is not FV-realizable, even asymptotically, since the corresponding POVM elements are all in $\pi(\hat{\W})$.

The above argument does show the FV-unrealizability of $\Omega_P$ for a given interaction zone $\OO$, provided that one can find zones, regions and test functions forming Bell and box setups and complying with the symplectic conditions. However, even if these conditions hold for a given interaction zone $\OO$, they might fail for larger $\OO$, in which case channel $\Omega_P$ could become FV-realizable. In this regard, in this section we argue that typical scalar QFTs satisfy the following property.

\begin{defin}[$U$ property]\label{def:uprop}
A scalar QFT satisfies \emph{the $U$ property} ($U$ from ``unbounded'') if, for $X=A,B$ there exist $\S^X,\T^X\subset \M, f_X\in C_0^\infty(\S^X), g_X\in C_0^\infty(\T^X)$ forming a spacelike box setup such that, for any compact $\OO\supset \bigcup_{X=A,B}\S^X\cup\T^X$, there exists a compatible Bell setup around $\OO$.    
\end{defin}
By the argument above, in any QFT with the $U$ property it is possible to define a causal Weyl channel that cannot be asymptotically FV-realized, no matter how large we take its interaction zone $\OO$ to be.

Luckily, the simplest QFT on the market happens to satisfy the $U$ property.
\begin{prop}
\label{prop:Bell_massless}
The massless Klein-Gordon field in Minkowski spacetime in $1+1$ dimensions has the $U$ property.
\end{prop}
The proof makes use of the explicit expression for the commutator function $E$ of the massless scalar field, and the reader can find it in Appendix \ref{app:massless}.

Exploiting genericity allows us to derive a variant of Proposition \ref{prop:Bell_massless} for arbitrary spacetimes.
\begin{prop}
\label{prop:Bell_massive}
Consider a generic scalar QFT, and let the associated spacetime $\M$ be such that there exist $\S^X,\T^X\subset \M$ forming a spacelike box setup, as well as unbounded regions $\J^X,\K^X\subset \M$ with a single connected component such that
\begin{align}
&J^+(\S^X)\supset \J^X, J^-(\T^X)\supset \K^X,\nonumber\\
&\overline{\J^A}\perp \overline{\J^B},\overline{\K^A}\perp \overline{\K^B},\overline{\J^A}\perp \overline{\K^B},\overline{\J^B}\perp \overline{\K^A}.
\label{conds_unbounded_O}
\end{align}
Then, the considered QFT has the $U$ property.
\end{prop}
The proposition is proven in Appendix \ref{app:massive}. An example in Minkowski spacetime satisfying the geometric constraints in (\ref{conds_unbounded_O}) is depicted in Figure \ref{fig:tsirel_viol}.
\begin{figure}[H]
\centering
\includegraphics[width=0.6\textwidth]{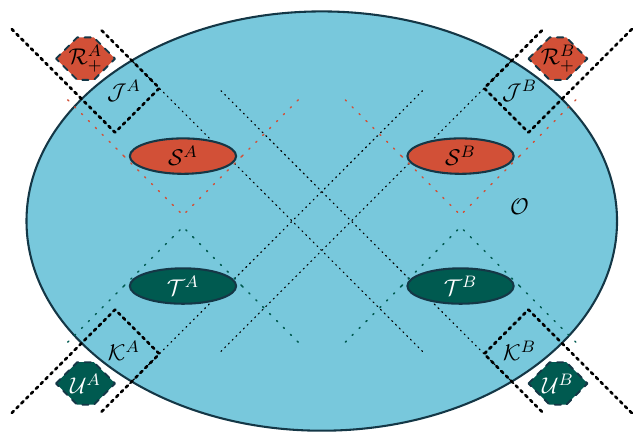}
\caption{Realization of conditions (\ref{conds_unbounded_O}) in Minkowski spacetime.}
\label{fig:tsirel_viol}
\end{figure}

\begin{remark}
\label{remark:anticip_friendly}
The constructions presented in the proofs of Propositions \ref{prop:Bell_massless}, \ref{prop:Bell_massive} allow one to set the compacts $\S^A,\T^A$ (and also $\S^B,\T^B$) arbitrarily far apart in time from each other and also arbitrarily far apart in space from $\S^B,\T^B$. For a fixed interaction region $\OO$, this allows integrating spacelike box and Bell setups satisfying the extra condition $\overline{\RR_{-}^A}\perp \overline{\RR_{-}^B}$. This last condition and the fact that, for $X=A,B$, $\S^X$ can be set apart in time from $\T^X$, are required in \Cref{sec:undecidability} to explore the undecidability of a special class of box channels, called FV-friendly, which, by the argument above, can also be realized in certain scalar QFTs.
\end{remark}

Given a generic free scalar QFT (under mild spacetime conditions) or the Klein-Gordon field in Minkowski spacetime with 1+1 dimensions, it is therefore possible to define causal Weyl channels that do not admit an FV realization, no matter how large we take their interaction region to be. Is it possible to quantify how non-FV-realizable any such channel $\Omega$ is? In the next section, we introduce a quantitative measure of approximate FV-realizability for QFT channels. For now, we conclude by deriving a simple linear constraint for FV-realizable channels that certain box channels violate. 

Consider a box channel $\Omega_P$ with $P\in C_{\mathrm{ns}},P\not\in C_{\mathrm{qc}}$. In the proof of Proposition \ref{prop:recovery_P}, it is explained how to build a family of states $\omega_\lambda$, Weyl POVMs $M^A, M^B$ and displacement channels $\{D^{A,\alpha}\}_{\alpha}$, $\{D^{B,\beta}\}_{\beta}$ such that
\begin{equation}
P(a,b|\alpha,\beta)= \lim_{\lambda\to 0}\omega_\lambda\circ D^{A,\alpha}\circ D^{B,\beta}\circ \Omega_P(M^A_aM^B_b);
\label{limit_distr_2}
\end{equation}
see eqs. (\ref{POVMs_Bell}) and (\ref{disps_Bell}) for the explicit expressions of POVMs and displacement channels in terms of the test functions $f_X,\hat{f}_X, g_X,\hat{g}_X$, $X=A,B$.

Since $C_{\mathrm{qc}}$ is convex and closed, by the Hahn-Banach theorem there exists a linear functional $\eta$ and $R\in \R$ such that \footnote{E.g.: by the Tsirelson bound (\ref{tsirelson_bound}), one can take $\eta(a,b,\alpha,\beta):=\frac{1}{4}\delta_{a\oplus b,\alpha\beta}$ and $R:=\frac{2+\sqrt{2}}{2}$.}
\begin{align}
&\eta(\tilde{P}):=\sum_{a,b,\alpha,\beta}\eta(a,b,\alpha,\beta)\tilde{P}(a,b|\alpha,\beta)\leq R,\forall \tilde{P}\in C_{\mathrm{qc}},\nonumber\\
&\eta(P)>R.
\label{witness_NL}
\end{align}
From eq. (\ref{limit_distr_2}) and the second line of (\ref{witness_NL}), it follows that
\begin{equation}
\lim_{\lambda\to 0}\omega_\lambda\circ (B(\Gamma,\Omega_P))>R,
\label{violation_witness}
\end{equation}
with
\begin{equation}
B(\Omega):=\sum_{a,b,\alpha,\beta}\eta(a,b,\alpha,\beta)D^{A,\alpha}\circ D^{B,\beta}\circ \Omega(M^A_aM^B_b).
\end{equation}
However, by Proposition \ref{prop:FV_box} and the first relation of (\ref{witness_NL}), for $\Omega$ FV realizable we have that
\begin{equation}
\omega\left(B(\Omega)\right)\leq R.
\end{equation}
Since this relation holds for all states $\omega$, it follows that
\begin{equation}
B(\Omega)\leq R,
\end{equation}
for all FV-realizable channels $\Omega$. This condition is closed, and, since the arguments of $\Omega$ in $B(\Omega)$ are elements of $\pi(\hat{\W})$, it also applies to asymptotically FV-realizable channels. By eq. (\ref{violation_witness}) this operator inequality is violated, though, by channel $\Omega_P$: the top of the spectrum of $B(\Omega_P)$ thus quantifies how much non-FV-realizable $\Omega_P$ is.

\subsection{Undecidability of the implementation of Weyl channels through FV schemes}
\label{sec:undecidability}
We just saw that there exist causal Weyl channels that do not admit an asymptotic FV representation. A natural, follow-up problem is thus how to characterize the set of FV-realizable causal Weyl channels. In this regard, we next present a strong no-go result.

First, in \Cref{sec:implement}, we propose a natural way to generate QFT channels by integrating FV schemes within a communication network that includes ancillary QFTs. We then introduce a notion of (network) implementability and approximate (network) implementability of a QFT channel. Importantly, we make the assumption that all QFTs involved in the realization of the channel (namely, the target, probe and ancillary QFTs) satisfy the split property. 

Then, in the next subsections, we show that determining whether a Weyl channel is implementable or not approximately implementable is an undecidable problem.

Here is a proof roadmap: in \Cref{sec:Bell_implement}, we prove that a certain class of Bell experiments involving an implementable channel can only generate correlations in $C_{\mathrm{qa}}$. In \Cref{sec:FV_friendly_boxes}, we define a subclass of box channels, called FV-friendly, and use results from the previous section to relate the approximate implementability of an FV-friendly box channel with distribution $P$ with the distance between $P$ and the set of distributions $C_{\mathrm{qa}}$. Finally, in \Cref{sec:charac_impl} we exploit this relation to prove that a hypothetical algorithm to distinguish implementable from far-to-implementable FV-friendly box channels would allow solving \Cref{problem:MIP} (see Section \ref{sec:quant_nl}), regarding the characterization of the set $C_{\mathrm{qa}}$ of quantum correlations. Since \Cref{problem:MIP} is undecidable, this means that said algorithm cannot exist.

\subsubsection{Implementability of quantum channels through compositions of FV schemes}
\label{sec:implement}
Taking FV schemes as basic operations, one can, as we did in the case of causal instruments, study what kind of effective channels can be generated through the integration of several FV schemes within a communication network. A controversial point is under which circumstances one can feed the outcome of an FV instrument to another. In the next lines, we outline a prescription that we believe makes physical sense.

In an FV implementation of an instrument, the target is made to interact with a probe system at $\K_1$, and then the probe is measured at a processing zone $\K_2\subset \K_1^+$. Consequently, we will assume in the following that \emph{any observer overlapping the total future of $\K_2$ can make use of that outcome to switch on a further FV scheme}. This is the only real difference with respect to the formalism for communication networks introduced in \Cref{sec:networks}: namely, that the interaction zone $\K_1$ only very indirectly determines the processing zone $\K_2\subset \K_1^+$ where the measurement outcome is generated. To each network node we will thus associate a family of FV instruments $\{\Omega^\alpha\}_{\alpha=1}^A$ with the same interaction and processing zones $\K_1$, $\K_2$. The realization of such a node will in general involve $A$ probes: one for each instrument.

Besides the probes, we will allow FV operations over multiple \emph{ancillary QFTs}, which will play the role of a QFT quantum memory. The effective channel induced through the FV scheme will be the result of tracing out all such ancillas after the last step of the communication protocol.

Finally, we will restrict to FV schemes where target, probe and ancillary QFTs satisfy the split property, introduced in \Cref{subsec:aqft}. This assumption is required in \Cref{sec:FV_friendly_boxes} to establish an equivalence between the implementability of a box channel with distribution $P$ and the property $P\in C_{\mathrm{qa}}$. 

\begin{remark}
If we accept the split property as a fundamental axiom and FV schemes are meant to represent the set of all local QFT operations available, then any implementable QFT channel should also be FV-realizable. Whether the set of FV-realizable instruments is closed under network composition, as described above, is, however, unknown.
\end{remark}

With this model of composition in mind, given some scalar QFT with a representation $\pi$ of the $\star$-algebra $\hat{W}$, we next introduce the notions of (asymptotic) implementability and approximate implementability of a QFT channel.
\begin{defin}
Let $\Omega:B(\H)\to B(\H)$ be a channel, and let $(\Omega^j)_j$ be a sequence of channels, each of which admits an FV realization, with total interaction zone $\OO$, through the composition of several FV schemes involving probes and ancillary QFTs where the split property holds. If, for all states $\omega$ and all $O\in \pi(\hat{\W})$, it holds that
\begin{equation}
\lim_{j\to\infty}\omega\circ\Omega^j(O)=\omega\circ\Omega(O),
\end{equation}
then we say that $\Omega$ is implementable in the interaction zone $\OO$.
\end{defin}

\begin{defin}
We say that channel $\Omega:B(\H)\to B(\H)$ is $\epsilon$-implementable in interaction zone $\OO$ if there exists an implementable channel $\tilde{\Omega}$, with interaction zone $\OO$, such that, for all states $\omega$ and all $O\in \pi(\hat{\W})$, with $\|O\|\leq 1$, it holds that
\begin{equation}
|\omega\circ\Omega(O)-\omega\circ\tilde{\Omega}(O)|\leq \epsilon. 
\label{approx_O}\end{equation}
\end{defin}

\subsubsection{Bell tests in implementable channels}
\label{sec:Bell_implement}
In \Cref{app:FV_wirings}, we prove a variant of Proposition \ref{prop:FV_box} for \emph{spacelike} Bell setups (complying with (\ref{cond_split})) and \emph{implementable} QFT channels.
\begin{prop}
\label{prop:FV_wirings_box}
Let $\OO\subset \M$ be compact, and let regions $\RR_\pm^X, \U^X$, with $X=A,B$, form a spacelike Bell setup around $\OO$. Let $\Omega$ be implementable in the interaction zone $\OO$, let $\omega$ be a QFT state, let $(M^A_a)_{a=0}^{M-1}$, $(M^B_b)_{b=0}^{M-1}$ be Weyl POVMs respectively localizable in $\RR^A_+$, $\RR^B_+$, and let $(\Lambda^{A,\alpha})_{\alpha=0}^{N-1}$, $(\Lambda^{B,\beta})_{\beta=0}^{N-1}$ be normal channels respectively localizable in $\U^A, \U^B$. Then, the distribution $P(a,b|\alpha,\beta)$ defined by eq. (\ref{Bell_corrs}) satisfies $P\in C_{\mathrm{qa}}$.
\end{prop}

\subsubsection{FV-friendly box channels}
\label{sec:FV_friendly_boxes}
We next define a subclass of box channels whose implementability depends on whether their generating distribution $P$ belongs to $C_{\mathrm{qa}}$ or not. First, we must make sure that implementability is even conceivable: this forces us to demand some small margins around $\S^X,\T^X$ within $\OO$ to accommodate for extra spacetime zones needed to define the implementation. In particular, for $X=A,B$, we require to find compacts $\Q^X,\tilde{\Q}^X,\tilde{\T}^X\subset \OO$ satisfying:
\begin{align}
&\tilde{\T}^X=\overline{\tilde{\T}_0^X},\tilde{\Q}^X=\overline{\tilde{\Q}_0^X},\mbox{ for some regions } \tilde{\T}_0^X, \tilde{\Q}_0^X,\nonumber\\
&\T^X\subset D^+(\tilde{\T}_0^X) \cap (\tilde{\T}_0^X)^+,\nonumber\\
&\tilde{\Q}^X\subset (\T^X)^\vee,\nonumber\\
&\Q^X\subset D^+(\tilde{\Q}^X_0) \cap (\tilde{\Q}_0^X)^+,\nonumber\\
&\S^X\subset (\Q^X)^\vee.\nonumber\\
\label{spacetime_conds_undecid}
\end{align}
These relations are depicted in \Cref{fig:FV_undecid}. 
\begin{figure}[H]
\centering
\includegraphics[width=0.8\textwidth]{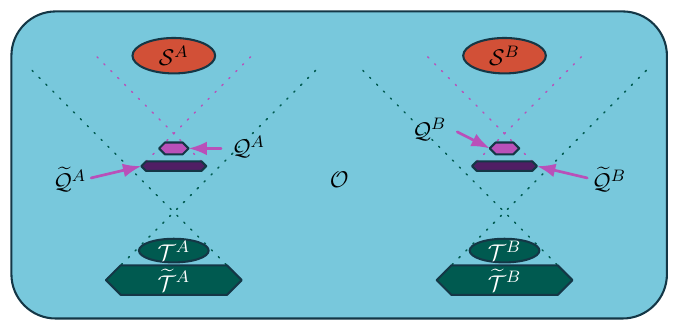}
\caption{Spacetime diagram for regions $\Q^X,\tilde{\Q}^X,\tilde{\T}^X$.}
\label{fig:FV_undecid}
\end{figure}

Now we can introduce our class of difficult-to-classify Weyl channels.
\begin{defin}
A box channel with distribution $P$ satisfying conditions (\ref{spacetime_conds_undecid}) for some compact $\Q^X,\tilde{\Q}^X,\tilde{\T}^X\subset \OO$ is called an FV-friendly box channel with distribution $P$.
\end{defin}
By Remark \ref{remark:anticip_friendly}, FV-friendly box channels exist for generic QFTs and the massless Klein Gordon field in Minkowski spacetime with $1+1$ dimensions. Being box channels, FV-friendly box channels are causal for all distributions $P\in C_{\mathrm{ns}}$. We next explore for which distributions $P$ such a channel is implementable.
\begin{prop}
\label{prop:FV_implement}
Let $\Omega$ be an FV-friendly box channel  in a Klein-Gordon field with box setup $\mathbb{S}$ and distribution $P(a,b|\alpha,\beta)\in C_{\mathrm{qa}}$. Let the compact $\OO\subset\M$ contain all the compact zones specifying $\mathbb{S}$. Then, $\Omega$ is implementable in the interaction zone $\OO$.
\end{prop}
\begin{proof}
Since $\overline{C_{\mathrm{qs}}}=C_{\mathrm{qa}}$, it is enough to prove that we can realize FV-friendly box channels with distribution $P\in C_{\mathrm{qs}}$. To do so, we will need an ancillary probe $\Psi$, also a scalar field, initially in some state $\psi$. For $X=A,B$, let $q^1_X,q^2_X\in C_0^\infty(\Q^X)$ define the mode $(\Psi(q^1_X),\Psi(q^2_X))$ and choose $\psi$ such that 
\begin{equation}\label{state_povm_oherence}
\psi(M^A_{a|\alpha}M^B_{b|\alpha})=P(a,b|\alpha,\beta),
\end{equation}
for some POVMs $(M^A_{\bullet|\alpha})_x$, $(M^B_{\bullet|\alpha})_x$, such that 
\begin{align}
&M^A_{a|\alpha}\in \A(\Q^A), \forall a,\alpha\nonumber\\
&M^B_{b|\beta}\in \A(\Q^B), \forall b,\beta,
\label{POVM_Neumann}
\end{align}
with $\Q^X:=\mbox{span}\{\Psi(q^1_X),\Psi(q^2_X)\}$.

The existence of a state and POVMs obeying conditions \eqref{state_povm_oherence},\eqref{POVM_Neumann} stems from the fact that $\A(\Q^A+\Q^B)$ is isomorphic to $B(\hat{H})^{\otimes 2}$, for any separable, infinite dimensional Hilbert space $\hat{H}$. In this representation, $\A(\Q^A)$ would act trivially on the second factor; and $\A(\Q^B)$, on the first. There is enough room, therefore, to realize any distribution in $C_{\mathrm{qs}}$. Due to eq. (\ref{POVM_Neumann}), for any $\alpha$, POVM $(M^A_{a|\alpha})_a$ can be realized through an FV scheme with interaction and processing zones $\tilde{\Q}^A,\Q^A$, see \cite{fewster2025measurementpreparationprotocolsquantum}. Likewise for POVM $(M^B_{b|\beta})_b$.

We are ready to present an FV communication protocol that implements $\Omega$. 
\begin{enumerate}
    \item Measure $\Omega^A_\alpha$ and, independently, $\Omega^B_\beta$. Since both instruments are Weyl diagonal, by Section \ref{subsec:withinfv} this operation can be realized through FV schemes with respective interaction zones $\tilde{\T}^A$, $\tilde{\T}^B$. Zones $\T^A,\T^B$ are used to respectively process the outcomes $\alpha,\beta$.
    \item Feed outcome $\alpha$ to implement the POVM $(M^A_{a|\alpha})_a$ on the ancillary QFT. Similarly, feed outcome $\beta$ to implement the POVM $(M^B_{b|\beta})_a$. 
    \item The outcomes $a,b$ of the two POVMs are respectively fed to $D^{a f_A}$, $D^{b f_B}$. This is feasible because $\S^A$, $\S^B$ are respectively contained in the total future of $\Q^A,\Q^B$ and because displacements within $\S^X$ are FV-realizable through an FV scheme with interaction zone $\S^X$.
    \item Trace out the ancillary QFT.
\end{enumerate}
It is immediate to verify that this protocol is sound.

\end{proof}
Building on Proposition \ref{prop:FV_implement}, one can connect the $\epsilon$-implementability of an FV-friendly box channel with how far away from $C_{\mathrm{qa}}$ its distribution $P$ is.
\begin{prop}
\label{prop:epsilon_imp2quan}
Let $\Omega$ be an FV-friendly box channel in a Klein-Gordon field, with interaction region $\OO$ compatible with a spacelike Bell setup around it, and distribution $P\in C_{\mathrm{ns}}$ for a Bell scenario with $M$ outcomes. Let 
\begin{equation}
\mu:=\min_{\tilde{P}\in C_{\mathrm{qa}}}\|P-\tilde{P}\|.
\label{min_dist_as}
\end{equation}
Then, $\Omega$ is $\mu$-implementable. Moreover, if $\Omega$ is $\epsilon$-implementable, then $\mu\leq \epsilon M^2$.
\end{prop}
The reader can find a proof in \Cref{app:epsilon_imp}.

\subsubsection{The approximate characterization of implementable channels}
\label{sec:charac_impl}
Consider the following family of decision problems, labeled by a prerequisite: 
\begin{table}[H]
    \centering
    \begin{tabular}{|ll|}
        \hline
        PREREQUISITE: & a spacetime $\M$, a Klein-Gordon scalar QFT defined therein, a compact interaction zone $\OO\subset\M$,\\
        & a spacelike box setup $\mathbb{S}:=\{\S^X,\T^X, f^X, g^X\}_{X=A,B}$ with $\OO\supset \bigcup_{X=A,B}\S^X\cup\T^X$, \\
        & complying with the friendliness condition (\ref{spacetime_conds_undecid})\\
        & and a spacelike Bell setup $\{\RR^X_{\pm}, U^X\}$ around $\OO$ compatible with $\mathbb{S}$.\\
        \hline
        INPUT:  & A bipartite Bell scenario $(M,N)$, a correlation $P\in C_{ns}$ therein, with $P(a,b|\alpha,\beta)\in\mathbb{Q}$, for all $a,b,\alpha,\beta$.\\
        \hline
        OUTPUT: &``YES", if the box channel $\Omega_P$ with distribution $P$ and box setup $\mathbb{S}$ is implementable in $\OO$,\\
        &``NO", if $\Omega_P$ is not $\epsilon$-implementable in $\OO$, with $\epsilon:=\frac{1}{3M^2}$,\\
        &``YES" or ``NO", otherwise.\\
        \hline
    \end{tabular}
    \captionsetup{name=Problem}
    \caption{Implementability of FV-friendly box channels}
    \label{undecid_problem}
\end{table}

We are ready to state our no-go result regarding the characterization of FV-implementable channels.
\begin{theorem}
\label{theo:undecidable}
For any fixed prerequisite, \Cref{undecid_problem} is undecidable. That is, no Turing machine can solve all its instances.
\end{theorem}
To prove the theorem, we will need the following result.
\begin{lemma}
\label{lemma:non_loc}
Given a Bell scenario $(M,N)$, define the maximally mixed box
\begin{equation}
P^I(a,b|\alpha,\beta):=\frac{1}{M^2},\forall a,b,\alpha,\beta  
\end{equation}
and let $P\in C_{\mathrm{qa}}$. Then, there exists a computable quantity $R(M,N)\in\mathbb{Q}^+$, such that, for any $\mu\in [0,1]$, the distribution $P^\mu:=(1-\mu)P+\mu P ^I$ satisfies
\begin{equation}
\forall \tilde{P}\in C_{\mathrm{ns}}:\|P^\mu-\tilde{P}\|\leq \mu R(M,N)\Rightarrow\tilde{P}\in C_{\mathrm{qa}}.
\label{perturb_local}
\end{equation}

\end{lemma}
The lemma is folklore among the quantum non-locality community, see, e.g., \cite{shrinking_factor}, for a discussion thereabout in $(2,N)$ Bell scenarios. Nevertheless, we include a general proof in \Cref{app:non_loc} for completeness.

\begin{proof}[Proof of Theorem \ref{theo:undecidable}]
Suppose that, for a given prerequisite, there exists an algorithm to solve Problem \ref{undecid_problem}. Given a Bell scenario $(M,N)$, compute $R(M,N)$ and define
\begin{equation}
\mu:=\frac{1}{12(2+R(M,N)}.
\end{equation}
Next, create a (finite) $\mu R(M,N)$-net\footnote{Given a normed vector space $V$ and a set $S\subset V$, a $\delta$-net of $S$ is a countable set $\{s_j\}_j\subset S$ such that, for all $s\in S$, there exists $j$ with $\|s_j-s\|\leq \delta$. If $V$ is finite-dimensional, any such net can be assumed finite.} $\{P^j\}_{j=1}^n$ for the set of all non-signalling boxes in the $(M,N)$ Bell scenario (with respect to the box norm (\ref{def:box_norm})), with each $P_j$ having rational probabilities. Since in the considered Bell scenario $C_{ns}$ has finitely many extreme points, all of which are computable and have rational probabilities, generating the net $\{P^j\}_{j=1}^n$ is a computable task.

Call $\Omega^j$ the FV-friendly box channel with distribution $P^j$, and use the algorithm to classify $\{\Omega^j\}_j$. Call $Q\subset\{1,...,n\}$ the set of all instances for which the solution was ``YES". 

Now, given a nonlocal game $G$, consider the problem of computing the value $C_{\mathrm{qa}}(G)$. We can estimate this quantity by calculating
\begin{equation}
\tilde{C}(G):=\max_{j\in Q}G(P^j).
\label{approx_game_value}
\end{equation}
Indeed, let $P\in C_{\mathrm{qa}}$ be such that $C_{\mathrm{qa}}(G)=G(P)$ and define $P_\mu:=(1-\mu)P+\mu P^I$. Since $P_\mu\in C_{\mathrm{ns}}$, there exists $j$ such that $\|P_\mu-P^j\|\leq \mu R(M,N)$. By Lemma \ref{lemma:non_loc}, it then follows that $P^j\in C_{\mathrm{qa}}$. This implies, by Proposition \ref{prop:FV_implement}, that $\Omega^j$ is implementable, and so $j\in Q$. Therefore,
\begin{align}
&\tilde{C}(G)\geq G(P^j)\geq G(P)-\|P-P^j\|\nonumber\\
&\geq C_{\mathrm{qa}}(G) -\|P-P_\mu\|-\|P_\mu-P^j\|\nonumber\\
&\geq C_{\mathrm{qa}}(G)-\mu(2+R(M,N))\nonumber\\
&=C_{\mathrm{qa}}(G)-\frac{1}{12},
\end{align}
where we used the inequalities $\|P-P_\mu\|\leq 2\mu$, $\|P_\mu-P^j\|\leq \mu R(M,N)$.

Now, let $j\in Q$ be the maximizer of eq. (\ref{approx_game_value}). Then, the channel $\Omega^j$ must be $\epsilon$-implementable; otherwise it would have been marked as ``NO". By Proposition \ref{prop:epsilon_imp2quan}, we thus have that there exists $\tilde{P}\in C_{\mathrm{qa}}$ such that
\begin{equation}
\|P^j-\tilde{P}\|\leq \epsilon M^2,
\end{equation}
which implies that
\begin{equation}
\tilde{C}(G)=G(P_j)\leq G(\tilde{P})+\epsilon M^2\leq C_{\mathrm{qa}}(G)+\frac{1}{3}.
\end{equation}

Hence, we can solve \Cref{problem:MIP}. For, if $C_{\mathrm{qa}}(G)<\frac{1}{2}$, we would observe that
\begin{equation}
\tilde{C}(G)< \frac{1}{2}+\frac{1}{3}=\frac{5}{6}.
\end{equation}
If, on the contrary, $C_{\mathrm{qa}}(G)=1$, we would have that
\begin{equation}
\tilde{C}(G)\geq 1-\frac{1}{12}=\frac{11}{12}>\frac{5}{6}.
\end{equation}
Since \Cref{problem:MIP} is undecidable, it follows that no algorithm can solve \Cref{undecid_problem}.

\end{proof}

\section{Discussion: the set of physical operations}
\label{sec:discussion}
Despite the fact that QFT is almost one century old, there is an ongoing debate with regards to modeling local operations within the theory. Fewster and Verch proposed a framework to that effect, which complies with relativistic covariance and at the same time is a QFT generalization of non-relativistic quantum measurement theory \cite{local_meas_QFT}. Jubb \cite{Jubb_2022} and Oeckl \cite{Oeckl_2026}, on the other hand, propose that all those operations compatible with relativistic causality should be regarded as feasible \emph{a priori}. The first approach therefore postulates an ad-hoc measurement model, while the second derives a set of operations from a fundamental physical principle, namely Einstein's causality. In \Cref{sec:causal_not_FV} we saw that these two proposals predict different sets of local operations, and so the FV framework cannot be derived from causality considerations alone.

To these authors, it is unclear whether the FV framework characterizes the set of all local operations available in QFT experiments: for all we know, the actual set of local operations could be smaller, larger, or have a non-trivial intersection with the set of FV-implementable instruments. We find it unlikely, however, that Einstein's causality will be questioned in the future, so we do not expect that non-causal instruments will play a role in a forthcoming QFT measurement theory. In fact, before this work some of us believed that Einstein's causality could be telling the whole story: namely, that a QFT operation is physically realizable iff it is causal. 

The results of \Cref{sec:box_channels} have made us revise our views. There we learned that, if all causal instruments were realizable\footnote{To observe an arbitrary distribution $P\in C_{ns}$ in a Bell experiment, it is not necessary that our QFT measurement theory posits the local realizability of \emph{all} causal QFT instruments. From the proof of Proposition \ref{prop:recovery_P}, it suffices to be able to: (a) prepare a squeezed, quasi-free local state; (b) carry out local displacements of the quantum  field; (c) implement a box channel with distribution $P$; and (d) measure the smeared field (in a demolition or a non-demolition way). Note that (a), (b), (d) are overwhelmingly regarded as physical; in fact, they describe standard operations in the quantum optics lab.}, then the correlations $P(a,b|\alpha,\beta)$ observed in a bipartite Bell experiment would only be limited by the non-signalling condition (\ref{NS_conds}). In particular, Bell experiments could generate behaviors beyond $C_{\mathrm{qc}}$, the largest of the three quantum sets of correlations considered in Bell non-locality.

The possibility of generating supra-quantum correlations (namely, distributions $P\in C_{\mathrm{ns}},P\not\in C_{\mathrm{qc}}$) in the course of a Bell experiment is an old topic in quantum foundations. While there is evidence that correlations slightly stronger than quantum would not violate any fundamental physical principle \cite{almost_quantum}, large deviations of the quantum limits are impossible, unless we are willing to let go several compelling physical intuitions. 

One of those is the principle of non-trivial communication complexity (NTCC) \cite{non_trivial}. To explain NTCC, picture a scenario where two separate parties, Alice and Bob, are each respectively handed the strings of bits $x,y\in\{0,1\}^n$ and asked to compute $f(x,y)$, for some Boolean function $f:\{0,1\}^n\times \{0,1\}^n\to \{0,1\}$. In general, this task cannot be accomplished without communication, and a large body of literature explores how much communication is needed to compute $f(x,y)$ with high probability. In a world where communication complexity were trivial, there would exist a probability $p>\frac{1}{2}$ such that, for \emph{any} $n\in\N$ and \emph{any} Boolean function $f$, by exchanging just one bit of communication, Alice and Bob could compute $f(x,y)$ with probability greater than $p$ for all inputs $x,y\in \{0,1\}^n$. To a computer scientist's ears, this is simply too good to be true. Consequently, the authors of \cite{non_trivial} postulate that our world has non-trivial communication complexity. 

NTCC is satisfied by non-relativistic quantum mechanics. Suppose, indeed, that Alice and Bob wish to compute the ``inner product function'' $f(x,y)=\sum_{i=1}^nx_iy_i\mbox{ (mod }2)$ with probability greater than $p$ for all inputs. As shown in \cite{inner_product,almost_quantum}, if Alice and Bob only use quantum resources (i.e., if they have access to entangled quantum states, which they can locally measure any way they like) and are just allowed $1$ bit of communication, then $p$ must satisfy 
\begin{equation}
p\leq \frac{1}{2}+\frac{1}{2^{\frac{n+1}{2}}}.    
\end{equation}
Taking the limit $n\to \infty$, $p$ is thus bounded by $\frac{1}{2}$, which precludes the non-trivial computation of the inner product function.

Contrarily, NTCC can be violated if one has access to arbitrary non-signalling boxes. Let $f:\{0,1\}^n\times \{0,1\}^n\to \{0,1\}$ be an arbitrary Boolean function and consider, in the Bell scenario $(2,2^n)$, the following box:
\begin{equation}
P_{f}(a,b|\alpha,\beta):=\frac{1}{2}\delta_{a\oplus b, f(\alpha,\beta)}.
\end{equation}
One can easily check that $P_f\in C_{ns}$. Now, suppose that Alice and Bob could generate $P_f$ in the course of a Bell experiment. Then Alice and Bob could run the following communication protocol: 
\begin{enumerate}
    \item Alice and Bob receive the inputs $x,y\in\{0,1\}^n$.
    \item Alice inputs $x$ in her device (namely, she sets $\alpha=x$), obtaining the result $a\in\{0,1\}$. Similarly, Bob sets $\beta=y$, obtaining the result $b\in\{0,1\}$. 
    \item Alice sends $a$ to Bob, who, by computing $a\oplus b$ obtains $f(x,y)$.
\end{enumerate}
A world where, for any $f$, the box $P_f\in C_{ns}$ were approximately realizable up to arbitrary precision would thus allow computing any Boolean function with one bit of communication with probability $p$ arbitrarily close to $1$. That world would violate NTCC.

Similarly, a world where any non-signalling distribution $P$ could be approximately observed in the course of a Bell experiment would further break the principles of No Advantage for Non-Local Computation \cite{no_advantage} and Information Causality \cite{inf_causality}. Moreover, if multiple independent instances of the process of generating $P$ could be run in parallel (e.g., by making them act on independent field modes), that world would also violate the principles of macroscopic locality \cite{mac_loc} and local orthogonality \cite{local_orth}.

Adopting Einstein's causality as the only physical principle limiting QFT local operations, as proposed in \cite{Jubb_2022, Oeckl_2026}, would therefore force us to abandon many well-established physical beliefs. If we accept that not all causal QFT operations are physical, then two obvious questions come to mind:
\begin{itemize}
\item Is there a natural QFT principle that constrains the set of feasible local operations beyond causal instruments?
\item Does any physical principle single out the set FV-implementable operations \emph{exactly}?
\end{itemize}

\section{Conclusion}
\label{sec:conclusion}
In this paper, we have studied the problem of characterizing which local operations are feasible in a QFT. Existing literature only proposes two candidate sets: the set of FV-realizable instruments, which embody a non-relativistic extension of von Neumann's measurement model; and the set of causal instruments \cite{Jubb_2022,Oeckl_2026}, which solely follows from the principle of Einstein's causality. Prior to our work, it was just known that all FV-realizable operations were causal \cite{impossible_meas}; that non-disturbing Gaussian measurements of Klein-Gordon scalar fields were FV-realizable \cite{field_meas}; and that the FV framework allowed conducting arbitrary local measurements of Klein-Gordon scalar fields in a demolition way \cite{fewster2025measurementpreparationprotocolsquantum}.

In this context, we first showed that a large class of causal operations, namely, all diagonal Weyl instruments, are FV-realizable in Klein-Gordon scalar fields. Such instruments allow the probabilistic, heralded realization of arbitrary Weyl instruments, causal or not. Moreover, they suffice to conduct local measurements in a non-demolition way. Standard tasks in a quantum optics laboratory, such as counting the photons of a given light mode without perturbing the rest, might therefore be modeled within the FV framework. 

Next, contrary to our early expectations, we found that the sets of FV-realizable and causal instruments differ. While they coincide for \emph{generic} subspaces of quadratures ${\mathcal Q}$, for \emph{general} ${\mathcal Q}$ we nonetheless constructed instances of causal QFT channels that cannot be approximated arbitrarily well through sequences of FV-realizable maps. 

This naturally prompted the question of which QFT channels/instruments are FV-realizable. In this direction, we proved that an approximate characterization of network-FV implementable Weyl channels is impossible, due to undecidability issues. This no-go result need not be interpreted as a flaw of the FV framework, though, for many physically relevant sets are similarly uncomputable, e.g., the set of bipartite quantum correlations $C_{qa}$ \cite{mipre}.

Finally, we argued that the set of causal channels includes many QFT operations that, if physically realizable, even approximately, would violate a number of physical principles \cite{non_trivial, no_advantage,inf_causality, mac_loc, local_orth}. Those principles have nothing to do with QFT, but were formulated more than a decade ago in the literature of quantum foundations. Each such principle therefore represents a compelling reason not to regard every causal instrument as a potentially feasible QFT operation.

Our work motivates two topics for further research. First, it is important to elucidate if FV operations are closed under causal composition. If they are, then our undecidability result extends to the characterization of FV-realizable channels; if they are not, then the set of FV operations cannot be regarded as physical. Second, if the set of FV instruments is, as we suspect, closed under causal composition, then it would be very convenient to prove that it follows from a set of well-accepted physical principles. With regards to a future measurement theory of QFT, any such physical derivation would promote the status of the FV framework from \emph{possible} to \emph{inevitable}.

\begin{acknowledgments}
This work was funded in whole or in part by the Austrian Science Fund (FWF) 10.55776/COE1 and the European Union – NextGenerationEU. M.N. was supported by the Swedish Research Council under grant no. 2021-06594 while in residence at Institut Mittag-Leffler in Djursholm, Sweden during April, 2026.
\end{acknowledgments}
\bibliographystyle{apsrev4-2}
\bibliography{article}

\begin{appendix}
\crefalias{section}{appsec}

\section{Archimedeanity of the Weyl algebra: Proof of Lemma \ref{lemma:Archimedeanity}}
\label{app:Archimedean}

We first prove that for any Hermitian $w\in \hat{\W}$ there exists $\lambda > 0$ such that $\lambda- w$ is a sums of Hermitian squares, hence proving that $\pi(w)\leq \lambda$. If we restrict to $w$ of the form $v v^\ast$ for some $v \in \hat{\W}$, this property is known as \emph{Archimedeanity} in non-commutative algebraic geometry \cite{Schmudgen2009}. The argument to prove the claim is simple: first, note that any Hermitian element $w$ of $\hat{\W}$ can be written as
\begin{equation}
w=\sum_{j=1}^N \left( \mu_j e^{-iR_j}+\mu_j^*e^{iR_j}\right),
\end{equation}
where $\{R_j\}_j$ are quadratures and $\{\mu_j\}_j\subset\C$. Define $w_j:=\mu_j e^{-iR_j}+\mu_j^*e^{iR_j}$. It can then be verified that
\begin{equation}
1+|\mu_j|^2-w_j=(1-\mu_j e^{-iR_j})(1-\mu_j e^{-iR_j})^*,
\end{equation}
which implies that $N+\sum_{j=1}^N |\mu_j|^2- w$ is SOS.

Next, we prove that Lemma \ref{lemma:Archimedeanity} holds when $\Q$ is the space of all quadratures. As explained in the main text, $\hat{\W}$ is simple. Hence, for any $w\in \hat{W}$, the condition $\pi(w)>0$ implies that $\tilde{\pi}(w)>0$ in all representations $\tilde{\pi}$ of $\hat{\W}$. Together with the archimedeanity of $\hat{\W}$, this implies that $w$ admits an SOS decomposition, see \cite[Prop.~12]{Schmudgen2009}. Since $\hat{\W}$ was defined over an arbitrary symplectic space, our derivation also holds whenever $\Q$ is a symplectic subspace of the space of all quadratures, i.e., when $\Q$ is the span of a number of independent modes.

It remains to check our statement in the case where $\Q$ is the span of the quadratures $\vec{R}=(R_1,...,R_{2m})$, which we can organize in $m$ modes, and the extra quadratures $\vec{S}=(S_1,...,S_r)$, which commute with all the $R$'s. Then it is possible to find $T_1,...,T_r$ commuting with all the $R$'s such that $(S_1, T_1, S_2,T_2,...,S_r,T_r)$ are $r$ modes. Define the extended quadrature space $\tilde{\Q}=\Q+\mbox{span}_\R\{T_j:j\}$. Then, by the previous arguments, there exist $\{d_j\}_j\subset\hat{\W}(\tilde{\Q})$ such that eq. (\ref{SOS_decomp_q}) holds and one can express $w$ as
\begin{equation}
w=\sum_{j,k}Z(j,k)e^{-i(\vec{f}_j\cdot \vec{R}+\vec{g}_j\cdot \vec{S}+\vec{h}_j\cdot \vec{T})}e^{i(\vec{f}_k\cdot \vec{R}+\vec{g}_k\cdot \vec{S}+\vec{h}_k\cdot \vec{T})},
\label{decomp_w}
\end{equation}
for some positive semidefinite matrix $Z$ and real vectors $\vec{f},\vec{g},\vec{h}$.

Since $w\in\hat{\W}(\Q)$, and $S_1, ...,S_r$ commute with $\Q$, for all $\vec{\theta}\in\R^r$, it holds that
\begin{align}
w&=e^{-i\vec{\theta}\cdot \vec{S}}we^{i\vec{\theta}\cdot \vec{S}} \nonumber\\
&=\sum_{j,k}Z(j,k)e^{-i(\vec{f}_j\cdot \vec{R}+\vec{g}_j\cdot \vec{S}+\vec{h}_j\cdot \vec{T})}e^{i(\vec{f}_k\cdot \vec{R}+\vec{g}_k\cdot \vec{S}+\vec{h}_k\cdot \vec{T})}e^{-i\vec{\theta}\cdot(\vec{h}_j-\vec{h}_k)}.
\label{gran_theta}
\end{align}

It is easy to find a finitely supported distribution $p(\vec{\theta})d\vec{\theta}$ such that
\begin{equation}
\int e^{-i\vec{\theta}\cdot(\vec{h}_j-\vec{h}_k)} p(\vec{\theta})d\vec{\theta} =0, \mbox{ for } \vec{h}_j\not=\vec{h}_k.
\end{equation}
Indeed, it suffices to take $\vec{\theta}=\sum_{j,k}\vec{\theta}_{jk}$, where $\{\vec{\theta}_{jk}:j<k,\vec{h}_j\not=\vec{h}_k\}$ are two-valued independent random variables with $\langle e^{-i\vec{\theta}_{jk}\cdot(\vec{h}_j-\vec{h}_k)}\rangle=0$. 

Computing the average on $\vec{\theta}$ of eq. (\ref{gran_theta}) ---this amounts to taking a finite convex combination---, we find that we can replace $Z$ in eq. (\ref{decomp_w}) by another positive semidefinite matrix $\tilde{Z}$ with the particularity that 
\begin{equation}
\tilde{Z}=\bigoplus_{\vec{h}} Z\bigg|_{\vec{h}},
\end{equation}
where $Z\bigg|_{\vec{h}}$ represents the matrix $Z$, restricted to the rows and columns $\{j:\vec{h}_j=\vec{h}\}$. Thus, 
\begin{equation}
w=\sum_{\vec{h}}\sum_jd_j(\vec{h})d_j(\vec{h})^*,
\end{equation}
with 
\begin{equation}
d_j(\vec{h})=\tilde{d}_j(\vec{h})e^{-i\vec{h}\cdot \vec{T}},
\end{equation}
for some $\tilde{d}_j(\vec{h})\in \hat{W}(\Q)$. Hence, 
\begin{equation}
w=\sum_{\vec{h}}\sum_j\tilde{d}_j(\vec{h})e^{-i\vec{h}\cdot \vec{T}}e^{i\vec{h}\cdot \vec{T}}\tilde{d}_j(\vec{h})^*=\sum_{\vec{h}}\sum_j\tilde{d}_j(\vec{h})\tilde{d}_j(\vec{h})^*.
\end{equation}
This completes the proof.

\section{Approximability of general POVMs and instruments through their Weyl counterparts}\label{app:approx}
The goal of this Appendix is to prove Propositions \ref{prop:approx_POVM}, \ref{prop:approx_instr} from the main text. To do so, we will rely on Lemma \ref{lemma:Archimedeanity} and the following two results:
\begin{lemma}
\label{lemma:bounded_convergence}
Let $(c^k)_k$ be a net of bounded operators that strongly converges to the bounded operator $c$. Then, there exists $C>0$ such that $\|c^k\|\leq C$, for all $k$.
\end{lemma}
\begin{proof}
The lemma is a direct consequence of the Uniform Boundedness Principle \cite{Rudin1991FunctionalAnalysis}, which states that, given a collection ${\cal O}$ of continuous linear functionals from some Banach space ${\cal X}$ to a metric space ${\cal Y}$, if it holds that $\sup_{O\in {\cal O}}\|Ox\|<\infty$, for all $x\in {\cal X}$, then $\sup_{O\in {\cal O}}\|O\|<\infty$. In our case, for any vector $\ket{\psi}\in\H$ it holds that $(c^k\ket{\psi})_k$ converges and thus the supremum of $\{\|c^k\ket{\psi}\|\}_k$ is bounded. Therefore, $\sup_k\|c^k\|<\infty$.    
\end{proof}

\begin{lemma}
\label{lemma:fromnetstosequences}
Let $(c^\alpha)_\alpha$ be a net of bounded operators on a separable Hilbert space $\H$ such that $\slim{\alpha\to\infty} c^\alpha=c$, $\slim{\alpha\to\infty} (c^\alpha)^*=c^*$, for some $c\in B(\H)$. Then, there exists a sequence $(c^j)_{j\in\N}$ of bounded operators (taken from the net) such that $\slim{j\to\infty} c^j=c$, $\slim{j\to\infty} (c^j)^*=c^*$.
\end{lemma}
\begin{proof}
    Let $\H$, $(c^{\alpha})_\alpha$, $c$ satisfy the conditions of the lemma, and let $\{\ket{n}\}_{n\in\N}$ be an orthonormal basis for $\H$. For each $j\in\N$, choose $\alpha_j$ such that, for $\alpha\geq \alpha_j$, it holds that $\|(c_j-c)\ket{n}\|,\|((c_j)^*-c^*)\ket{n}\|\leq\frac{1}{j}$, for $n=1,...,j$. Define $c^j:=c^{\alpha_j}$. From the above we have that $\|(c^j-c)\ket{\psi}\|, \|((c^j)^*-c^*)\ket{\psi}\|\leq\frac{1}{\sqrt{j}}$, for any $\ket{\psi}\in \H_j:=\mbox{span}\{\ket{n}:n=1,...,j\}$. Now, consider any $\ket{\psi}\in\H$. For any $\epsilon>0$, there exists $j$ such that $\|\ket{\psi}-\ket{\tilde{\psi}}\|\leq \epsilon$, for some $\ket{\tilde{\psi}}\in\H_j$. By Lemma \ref{lemma:bounded_convergence}, we know that there exists $C>0$ with $|c^j|\leq C$ for all $j$. For $k\geq j$, this implies that 
    \begin{align}
    \|(c^k-c)\ket{\psi}\|\leq \|(c^k-c)\ket{\tilde{\psi}}\|+2\epsilon C\leq \frac{1}{k}+2\epsilon C.  
    \end{align}
    Since $\epsilon$ was arbitrary, this implies that $\lim_{k\to\infty}\|(c^k-c)\ket{\psi}\|=0$ and thus $(c^j)_j$ converges to $c$ strongly. Mutatis mutandis, the same argument shows that $((c^j)^*)_j$ converges to $c^*$.
\end{proof}

In the following, we fix $\pi$ to be an irreducible representation of $\hat{\W}$ on $B(\H)$, i.e., $\A_\pi = \overline{\pi(\hat{\W})} = B(\H)$, where the overline indicates closure with respect to the strong operator topology. Thus, implied by the above Lemmas \ref{lemma:bounded_convergence}, \ref{lemma:fromnetstosequences}, for any element $c$ of $B(\H)$ we find a uniformly bounded sequence $(c_j)_j$ within $\hat{\W}$ such that $\slim{j\to\infty} \pi(c_j) = c$. We can even demand $\sstarlim{j\to \infty} \pi(c_j) = c$, where $\sstarlim{\,}$ refers to the strong$^\ast$-topology, which implies that $\slim{j\to\infty} \pi(c_j) = c$ and $\slim{j\to\infty} \pi(c_j)^\ast = c^\ast$. Further, we state here once that, if $c \in \A_\pi(\Q)$ for some subspace of quadratures $\Q$, then $(c_j)_j$ can be chosen as a subset of $\hat{\W}(\Q)$ instead of merely $\hat{\W}$, thus implying the latter parts of the statements in the propositions.

\begin{proof}[Proof of Proposition \ref{prop:approx_POVM}]
Since $M_a\geq 0$, there exists $c_a\in B(\H)$ such that $c_a(c_a)^*=M_a$ and a sequence $(c_a^j)_j$ of elements in $\hat{\W}$ such that 
\begin{equation}
\sstarlim{j\to \infty}\pi(c_a^j)=c_a.  
\end{equation}
Generally, though,
\begin{equation}
    c^j :=\sum_a c_a^j (c_a^j)^\ast \neq 1,    
\end{equation}
which implies that $(c_a^j (c_a^j)^\ast)_a$ is not a Weyl POVM. However, since the limiting tuple $(M_a)_a$ is a POVM, we do have that $\slim{j\to \infty} \pi(c^j) = 1$, so, by Lemma \ref{lemma:bounded_convergence}, there exists $C \geq 1$ such that
\begin{equation}
\pi(c^j)\leq C, \forall j.
\label{bound_norm}
\end{equation}

For $\epsilon>0$, let $\tilde{p}(t;\epsilon)$ be a polynomial in $t$ satisfying 
\begin{equation}\label{eq:polydemand}
    |\tilde{p}(t;\epsilon) - \tilde{p}(t)| < \epsilon, \quad \tilde{p}(t) := \left\lbrace \begin{matrix} 1, & \quad t\in [0,1] \\ t^{-\frac{1}{2}}, & \quad t \in [1,C] \end{matrix} \right. ,
\end{equation}
The existence of $\tilde{p}(t;\epsilon)$ for any $\epsilon > 0$ follows from the Stone-Weierstrass theorem.

Using $|t\tilde{p}(t)| \leq \sqrt{C}$ and $|t\tilde{p}(t)^2| \leq 1$, we find that
\begin{equation}
    t \tilde{p}(t;\epsilon)^2 \leq t \tilde{p}(t)^2 + t |\tilde{p}(t;\epsilon)^2- \tilde{p}(t)^2| < 1 + \epsilon', t\in [0,C]  
\end{equation}
with $\epsilon' := 2\epsilon \sqrt{C} + \epsilon^2 C$. Thus, the polynomial
\begin{equation}
    p(t;\epsilon) := (1+\epsilon')^{-1} \tilde{p}(t;\epsilon)
\end{equation}
satisfies $tp(t;\epsilon)^2<1$, for $t\in [0,C]$. 

Now, define
\begin{align}
&c_a^j(\epsilon) := p(c^j;\epsilon) c_a^j \in \hat{\W},\nonumber\\
&c^j(\epsilon) := \sum_a c_a^j(\epsilon) (c_a^j(\epsilon))^\ast = c^j p(c^j;\epsilon)^2.
\end{align}
Then it follows that
\begin{equation}\label{eq:papprox}
\pi(c^j(\epsilon)) = \left. t p(t;\epsilon)^2 \right|_{t=\pi(c^j)} < 1.
\end{equation}
Lemma \ref{lemma:Archimedeanity} then implies that $1-c^j(\epsilon)$ admits a Hermitian sum of squares decomposition (SOS) so that
\begin{equation}
    M^j_a(\epsilon) := c^j_a(\epsilon) (c_a^j(\epsilon))^\ast + \delta_{a1} (1- c^j(\epsilon))    
\end{equation}
forms a Weyl POVM; in particular,
\begin{equation}
    \sum_a M_a^j(\epsilon) = 1.
\end{equation}

We next show that, for any sequence of positive numbers $(\epsilon_j)_j$ with $\epsilon_j \to 0$, it holds that
\begin{equation}\label{eq:slimcj}
    \sstarlim{j\to \infty} \pi(c_a^j(\epsilon_j)) = c_a,
\end{equation}
implying $\slim{j\to \infty} \pi(M_a^j(\epsilon)) = M_a^j$.

For simplicity we set $\hat{c}^j := \pi(c^j)$ and $\hat{c}_a^j := \pi(c_a^j)$. For arbitrary $\psi \in \H$ and $\epsilon \leq \epsilon_0$, we estimate
\begin{align}
    \norm{(\hat{c}_a^j(\epsilon) - \hat{c}_a)\psi} &\leq \norm{p(\hat{c}^j;\epsilon) (\hat{c}_a^j - \hat{c}_a)\psi} + \norm{(p(\hat{c}^j;\epsilon) -1)\hat{c}_a \psi} \\
    &\leq A_1 \norm{(\hat{c}_a^j - \hat{c}_a)\psi} +  A_2 \epsilon\norm{\hat{c}_a\psi} + A_3 \norm{(\tilde{p}(\hat{c}^j) - 1)\hat{c}_a \psi}
    \label{norma_contenida}
\end{align}
for constants $A_1,A_2,A_3 > 0$ (independent from $\epsilon$, $j$ and $\psi$), where we used that $p(\cdot;\cdot)$ is uniformly bounded on $[0,C]\times [0,\epsilon_0]$ to obtain $A_1$ and we inferred from \eqref{eq:polydemand} that
\begin{equation}
    \norm{(p(\hat{c}^j;\epsilon) - 1)\chi} \leq (1+\epsilon')^{-1} (\epsilon+\epsilon') \norm{\chi} + (1+\epsilon')^{-1} \norm{(\tilde{p}(c^j)-1)\chi}.
\end{equation}
Setting $\epsilon=\epsilon_j$ in eq. (\ref{norma_contenida}), we take the limit $j\to\infty$. Using that $\tilde{p}$ is continuous, with $\tilde{p}(1) = 1$, we have that the second line of eq. (\ref{norma_contenida}) vanishes, and thus
\begin{equation}
    \norm{(\hat{c}_a^j(\epsilon_{j}) - \hat{c}_a)\psi} \to_{j \to \infty} 0.
\end{equation}
The same argument can be easily adapted to prove that
\begin{equation}
    \norm{(\hat{c}_a^j(\epsilon_{j})^\ast - \hat{c}^\ast_a)\psi} \to_{j \to \infty} 0.
\end{equation}
\end{proof}

\begin{proof}[Proof of Proposition \ref{prop:approx_instr}] Following the proof of Proposition \ref{prop:approx_POVM}, it is easy to show that any $A$-outcome instrument $\tilde{\Omega}^{(N)}$ (on $B(\H)$) such that $\tilde{\Omega}^{(N)}_a$ has at most $N$ Kraus operators can be approximated strongly through a sequence of Weyl instruments. Now, consider a general normal instrument $\Omega$, with
\begin{equation}
\Omega_a(\bullet)=\sum_{j=1}^{\infty}K^{j}_a\bullet (K^{j}_a)^*, \quad K^j_a \in B(\H),
\label{normal_app}
\end{equation}
and define its $N$-truncation $\tilde{\Omega}^{(N)}$ as
\begin{equation}
\tilde{\Omega}^{(N)}_a(\bullet):=\sum_{j=1}^{N} K^{j}_a\bullet (K^{j}_a)^* + \delta_{a1} \sqrt{R^N}\bullet \sqrt{R^N}
\end{equation}
with
\begin{equation}
R^N:=\sum_{j=N+1}^\infty \sum_a K_a^j(K_a^j)^*.
\end{equation}
Then $\tilde{\Omega}^{(N)}$ is an instrument for all $N$. Moreover, due to the strong convergence of (\ref{normal_app}), we have that
\begin{equation}
\slim{N\to \infty} \tilde{\Omega}_a^{(N)}(X)=\Omega_a(X)
\label{conv_finite_terms}
\end{equation}
strongly, for all $X\in B(\H)$ and $a \neq 1$. For $a=1$ one needs, in addition, to consider the term $\sqrt{R^N}X\sqrt{R^N}$. Let $\psi,\chi \in\H$ with $\norm{\chi} = 1$. Then, for $X\in B(\H)$, and any sequence of unit vectors $(\ket{\chi_N})_N$, we have that
\begin{equation}
    |\braket{\chi_N}{ \sqrt{R^N} X \sqrt{R^N} \psi}| \leq  \norm{X} \sqrt{\braket{\chi_N}{R^N \chi_N}\braket{\psi}{R^N \psi}} \leq  \norm{X} \sqrt{\braket{\psi}{R^N \psi}},
\end{equation}
using that $\braket{\chi_N}{R^N \chi_N} \leq \braket{\chi_N}{R^0 \chi_N} = 1$. By the strong convergence of the series that defines $\Omega_a$, the r.h.s. tends to zero as $N$ tends to infinity. Choosing $\ket{\chi_N}$ to be parallel to $\sqrt{R^N} X \sqrt{R^N}\ket{\psi}$, the above implies that $\slim{N\to\infty}\sqrt{R^N}X\sqrt{R^N}= 0$.  

The fact that for, each such $\tilde{\Omega}^{(N)}$, there exists a converging sequence of Weyl instruments implies, through a diagonal argument, that there exists a sequence of Weyl instruments $(\Omega^n)_n$ that converges to $\Omega$ strongly.
\end{proof}

\section{Asymptotic implementation of Weyl POVMs via diagonal Weyl instruments: proof of Lemma \ref{lemma_asymp}}
\label{app:structural}
First a word on notation: in this Appendix, given a vector $\vec{z}\in\R^n$, by $\|\vec{z}\|$ we will denote its $l_\infty$ norm, i.e., $\|\vec{z}\|:=\max_j|z_j|$. For $\vec{z}\in\R^n,\delta>0$, the closed ball of radius $\delta$ centered at $\vec{z}$ will be denoted $B(\vec{z},\delta)$. Namely, $B(\vec{z},\delta)=\{\vec{y}:\|\vec{y}-\vec{z}\|\leq \delta\}$.

Before we proceed to the proof of Lemma \ref{lemma_asymp}, we need to define a special representation for Weyl instruments.
\begin{defin}
Let $\Omega$ be a Weyl instrument over $\Q$, and for some $n \in \NN$ let $R_1,...,R_n\in\Q$ be linearly independent over the integers, i.e., if there exist integers $\{z_j\}\subset \Z$ such that
\begin{equation}
\sum_j z_jR_j=0,
\end{equation}
then $z_1=...=z_n=0$.
Then, if for some $m\in\NN$ and some $(2m+1)^n\times (2m+1)^n$ positive semidefinite matrices $(Z_a)_a$ it holds that
\begin{equation}
\Omega_a(X)=\sum_{\vec{j},\vec{k}\in \{-m,...,m\}^n}Z_a(\vec{j},\vec{k})e^{i\vec{j}^T\cdot\vec{R}} Xe^{-i\vec{k}^T\cdot\vec{R}},
\end{equation}
for $a=1,...,A$, then we say that $(Z,R)$, with $Z=(Z_1,...,Z_A)$ and $R=(R_1,...,R_n)$, is a \emph{lattice representation} of $\Omega$.
\end{defin}
\begin{remark}
Any Weyl instrument over a finite subspace of quadratures $\Q$ admits a lattice representation. Indeed, given a standard representation (\ref{SOS_intr}) of $\Omega$ in terms of the quadratures $Q_1,...,Q_s \in \Q$, it is enough to find a subset of indices $J\subset\{1,...,s\}$ such that $\{Q_j:j\in J\}$ are linearly independent over the integers and, for any $j\not\in J$, $Q_j$ can be expressed as a linear combination of $\{Q_j:j\in J\}$ with integer coefficients.
\end{remark}

In lattice representation, the normalization condition (\ref{znorm}) reads:
\begin{equation}
1=\sum_a\Omega_a(1)=\sum_{\vec{j},\vec{k}}\bar{Z}(\vec{j},\vec{k})e^{i(\vec{j}-\vec{k})^T\cdot\vec{R}}e^{i\frac{\vec{j}^T\cdot\sigma\vec{k}}{2}},
\label{normal_instr2}
\end{equation}
where $\bar{Z}:=\sum_aZ_a$ and $\sigma$ denotes the symplectic matrix of $R$, i.e., $[R_j, R_k]=i\sigma_{jk}$. Due to the linear independence of $R_1,...,R_n$ over the integers, it follows that the exponents $(\vec{j}^T \cdot \vec{R})_{\vec{j}}$ are mutually different. As a consequence, the elements $(e^{i \vec{j}^T \cdot \vec{R}})_{\vec{j}}$ are linearly independent (over $\C$) and condition \eqref{normal_instr2} is equivalent to:
\begin{align}
&1=\sum_{\vec{j}}\overline{Z}(\vec{j},\vec{j}),\nonumber\\
&0=\sum_{\vec{j},\vec{k} \in\Z^n: \vec{j}-\vec{k} = \vec{d}, \|\vec{j}\|, \|\vec{k}\|\leq m}\overline{Z}(\vec{j},\vec{k})e^{-\frac{i}{4}(\vec{j}+\vec{k})^T\cdot\sigma (\vec{j}-\vec{k})},\quad \forall \vec{d}\in\Z^n,\vec{d}\neq 0.
\label{sum_one_primes}
\end{align}
The above equalities are equivalent to
\begin{equation}
    \delta_{\vec{j}\vec{k}} =\sum_{\vec{c}\in\Z^n:\|\vec{j}-\vec{c}\|, \|\vec{k}-\vec{c}\|\leq m}\overline{Z}(\vec{j}-\vec{c},\vec{k}-\vec{c})e^{\frac{i}{2}\vec{c}^T\cdot\sigma (\vec{j}-\vec{k})},\quad \forall \vec{j},\vec{k}\in\Z^n
    \label{sum_lemma}
\end{equation}
where the global phase factor $e^{-\frac{i}{4}(\vec{j}+\vec{k})^T \cdot \sigma (\vec{j}-\vec{k})}$ has been factored out.

With this notation in place, we are ready to prove the Lemma.
\begin{proof}[Proof of Lemma \ref{lemma_asymp}]
Let $(Z,R)$ be a lattice representation of $\Omega$, with $Z$ a $(2m+1)^n \times (2m+1)^n$ positive semidefinite matrix and $R=(R_1,...,R_n)$. Fix $N\in\NN$, define the (causal) Weyl channel
\begin{equation}
\Omega'(X):=\frac{1}{(2N+1)^n}\sum_{\|\vec{c}\| \leq N} e^{i\vec{c}^T\cdot \vec{R}}Xe^{-i\vec{c}^T\cdot \vec{R}}.
\end{equation}
With this definition, we have that:
\begin{align}
&\Omega\circ \Omega'(X)=\frac{1}{(2N+1)^n}\sum_{\|\vec{j}\|,\|\vec{k}\|\leq m}\sum_{\|\vec{c}\| \leq N}Z(\vec{j},\vec{k})e^{i\vec{j}^T\cdot\vec{R}}e^{i\vec{c}^T\cdot\vec{R}}Xe^{-i\vec{c}^T\cdot\vec{R}}e^{-i\vec{k}^T\cdot\vec{R}}\nonumber\\
&=\frac{1}{(2N+1)^n}\sum_{\|\vec{j}\|,\|\vec{k}\|\leq m}\sum_{\|\vec{c}\| \leq N}Z(\vec{j},\vec{k})e^{i(\vec{j}+\vec{c})^T\cdot\vec{R}}Xe^{-i(\vec{k}+\vec{c})^T\cdot\vec{R}}e^{-\frac{i(\vec{j}-\vec{k})^T\sigma\vec{c}}{2}}.
\end{align}
Changing our summation variables to $\vec{j}' := \vec{j}+\vec{c}$ and $\vec{k}':=\vec{k}+\vec{c}$, the above expression equals
\begin{align}
&\sum_{\|\vec{j}'\|,\|\vec{k}'\| \leq N+m}\frac{1}{(2N+1)^n}\sum_{\substack{\vec{c}:\|\vec{c}\|\leq N,\\\|\vec{j}'-\vec{c}\|, \|\vec{k}'-\vec{c}\|\leq m}}Z(\vec{j}'-\vec{c}, \vec{k}'-\vec{c})e^{i(\vec{j}')^T\cdot\vec{R}}Xe^{-i(\vec{k}')^T\cdot\vec{R}}e^{-\frac{i}{2}(\vec{j}'-\vec{k}')^T \sigma \vec{c}}\nonumber\\
&=\sum_{\|\vec{j}\|,\|\vec{k}\| \leq N+m} \tilde{Z}(\vec{j},\vec{k}) \, e^{i\vec{j}^T\cdot\vec{R}}Xe^{-i\vec{k}^T\cdot\vec{R}},
\label{comp_int_expr}
\end{align}
with 
\begin{equation}
\tilde{Z}(\vec{j},\vec{k}):= \frac{1}{(2N+1)^n}\sum_{\substack{\vec{c}:\|\vec{c}\|\leq N,\\\|\vec{j}-\vec{c}\|, \|\vec{k}-\vec{c}\|\leq m}}Z(\vec{j}-\vec{c}, \vec{k}-\vec{c}) e^{-\frac{i}{2}(\vec{j}-\vec{k})^T \sigma \vec{c}},
\end{equation}
keeping the dependence on $N$ and $m$ implicit. 

We note that $\tilde{Z}(\vec{j},\vec{k})$ matches the r.h.s. of \eqref{sum_lemma} apart from the prefactor and the restriction, $\|\vec{c}\| \leq N$. Consequently, for all pairs of unequal vectors $(\vec{j},\vec{k})$ for which
\begin{equation}
\|\vec{j}-\vec{c}\|, \|\vec{k}-\vec{c}\|\leq m
\end{equation}
implies that $\|\vec{c}\|\leq N$, it holds that $\tilde{Z}(\vec{j},\vec{k})=0$. The set of such vectors corresponds to
\begin{equation}
A:=\{(\vec{j},\vec{k}): B(\vec{j},m)\cap B(\vec{k},m)\subset B(\vec{0}, N), ~ \vec{j}\neq \vec{k}\}.    
\end{equation}
In particular, any pair of unequal vectors $(\vec{j},\vec{k})$ such that either $\vec{j}\in B(\vec{0}, N-m)$ or $\vec{k}\in B(\vec{0}, N-m)$ will comply. Let us thus define the `bulk' and 'boundary' sets of vectors 
\bea
    V &:=& B(\vec{0}, N-m)\cap \Z^n, \\
    B &:=& B(\vec{0}, N+m)\cap \Z^n \setminus V,
\eea
and note that 
\begin{equation}
\tilde{Z}(\vec{j},\vec{k})=\frac{1}{(2N+1)^n} \delta_{\vec{j}\vec{k}}, \quad \text{if }\vec{j} \in V \vee \vec{k}\in V.
\end{equation}

Therefore,
\begin{equation}
\tilde{Z}=\tilde{Z}_V\oplus \tilde{Z}_B,
\end{equation}
where $\tilde{Z}_V$ is a diagonal matrix. Note also that
\begin{equation}
    \tr \tilde{Z}_V = \frac{|V|}{(2N+1)^n} = (1- \frac{2m}{2N+1})^n = 1 - \frac{2mn}{2N+1} + O(N^{-2})
\end{equation}
so that
\begin{equation}
    \tr \tilde{Z}_B = 1 - \tr \tilde{Z}_V = O(\tfrac{1}{N}).
\end{equation}
That is good news, because $\tilde{Z}_B$ is the only part of $\tilde{Z}$ that is not causal. We next construct another positive semidefinite matrix $\hat{Z}_B$, with trace $\tr(\hat{Z}_B)=O\left(\frac{1}{N}\right)$, such that $\tilde{Z}_B+\hat{Z}_B$ is diagonal. The channel represented by the matrix
\begin{align}
&\frac{1}{\tr(\tilde{Z}_V)+\tr(\tilde{Z}_B)+\tr(\hat{Z}_B)}\left(\tilde{Z}_V\oplus (\tilde{Z}_B+\hat{Z}_B)\right)=\nonumber\\
&\frac{1}{1+\tr(\hat{Z}_B)}\left(\tilde{Z}_V\oplus (\tilde{Z}_B+\hat{Z}_B)\right)
\end{align}
is therefore diagonal, and can be expressed as
\begin{align}
\frac{1}{1+\tr(\hat{Z}_B)}\Omega\circ\Omega'+\frac{\tr(\hat{Z}_B)}{1+\tr(\hat{Z}_B)}\Omega'',
\end{align}
where $\Omega''$ is the (possibly non-diagonal) Weyl channel defined by the matrix $\frac{0 \oplus \hat{Z}_B}{\tr(\hat{Z}_B)}$. As $N$ grows, the term $\tr(\hat{Z}_B)$ tends to zero, recovering the statement of the lemma.

To build $\hat{Z}_B$, first note that the condition $\|\vec{j}-\vec{c}\|, \|\vec{k}-\vec{c}\|\leq m$ implies
\begin{equation}
\|\vec{j}-\vec{k}\|\leq 2m.
\end{equation}
Thus, it holds that $\tilde{Z}(\vec{j},\vec{k}) = 0$ for $\|\vec{j}-\vec{k}\|> 2m$ and, accordingly,
\begin{equation}
    \tilde{Z}_B(\vec{j},\vec{k})=0 \text{ for } \vec{j}, \vec{k}\in B, \|\vec{j}-\vec{k}\|> 2m.
\end{equation}
Namely, the structure of correlations in $\tilde{Z}_B$ is local: just near neighbours have nonzero correlations.

We partition the boundary set $B$ according to $B = \cup_{c=1}^{\alpha} B_c$ into $\alpha:=(2m+1)^n$ sets, where we defined $B_c = B \cap (\vec{c} + (2m+1) \Z^n)$ for $c \in \{ 1, 2,\ldots ,\alpha\}$ and the vector $\vec{c}=(c_1,\ldots,c_n)^t$ is defined via $c = \sum_{j=1}^{n} c_j (2m+1)^{j-1}$. Note that there are no correlations within each of these sets, i.e., $\tilde{Z}_B(\vec{j},\vec{k}) = 0$ if $\vec{j},\vec{k} \in B_c$ and $\vec{j}\neq \vec{k}$. Now, let $\Lambda(r_1,...,r_{\alpha})$ be the diagonal matrix such that
\begin{align}
\Lambda_{\vec{j},\vec{j}}(r_1,...,r_{\alpha})=r_c,\mbox{ if }\vec{j}\in B_c.
\end{align}
Then we define the positive semidefinite matrices
\bea
\check{Z}_B &:=& \sum_{\vec{r}\in \{-1,1\}^\alpha}\Lambda(r_1,...,r_{\alpha})\tilde{Z}_B\Lambda(r_1,...,r_{\alpha})^\ast, \\
\hat{Z}_B &:=& \sum_{\vec{r}\in \{-1,1\}^\alpha\setminus \{(1,1,...,1)\}}\Lambda(r_1,...,r_{\alpha})\tilde{Z}_B\Lambda(r_1,...,r_{\alpha})^\ast,
\eea
and note that
\begin{equation}
    \check{Z}_B = \hat{Z}_B + \tilde{Z}_B,
\end{equation}
is a diagonal matrix. Indeed, let $\vec{j}\in B_c,\vec{k}\in B_d$, with $\vec{j}\not=\vec{k}$. Then,
\begin{equation}
\check{Z}_B(\vec{j},\vec{k})=\sum_{\vec{r}\in \{-1,1\}^\alpha}r_cr_d\tilde{Z}_B(\vec{j},\vec{k})
\end{equation}
If $c=d$, we have already noted that $\tilde{Z}_B(\vec{j},\vec{k})=0$. If $c\not=d$, then the sum over values of $r_c, r_d$ will be zero. Either way, $\check{Z}_B(\vec{j},\vec{k})=0$. We also check that 
\begin{equation}
\tr(\hat{Z}_B)=(2^\alpha-1)\tr(\tilde{Z}_B)=O(\tfrac{1}{N}).
\end{equation}
It remains to check that
\begin{equation}
\Omega''(X):=\frac{1}{\tr(\hat{Z}_B)}\sum_{\vec{j}, \vec{k}}\hat{Z}_B(\vec{j}, \vec{k})e^{i\vec{j}^T\cdot\vec{R}}Xe^{-i\vec{k}^T\cdot\vec{R}}
\end{equation}
satisfies the channel normalization condition $\Omega''(1)=1$. This follows from the fact that $(\tr(\hat{Z}_B))^{-1} \hat{Z}_B$ has the correct normalization and that $\hat{Z}_B = \tilde{Z}_B - \check{Z}_B$ decomposes into representing matrices of maps that are proportional to a channel.
\end{proof}

\section{Causality: two useful lemmas}\label{app:lemmas_causal}

In this section we prove two lemmas, formulated in Section \ref{sec:causality}, that are invoked throughout this article in the context of the PSNI condition \eqref{PSNI_def}. The first one, Lemma \ref{lemma:psni_simpler}, allows simplifying the PSNI condition, by restricting to pairs of regions $\RR_\pm \subset \OO^\pm$ with $\overline{\RR}_+ \subset D^+(\RR_-)$ instead of merely $\overline{\RR}_+ \subset D(\RR_-)$. The second one, Lemma \ref{lemma:weyllocality}, establishes that the Kraus operators of any Weyl instrument local to a compact $\OO\subset\M$ are localizable in any region containing $\OO$.

Lemma \ref{lemma:psni_simpler} follows from the next two propositions:
\begin{prop}
\label{prop:D_plus}
Let $\RR\subset \M$ be a region. Then, $D^+(\RR)$ is also a region.
\end{prop}
\begin{proof}
$D^+(\RR)$ is clearly open, so we just need to prove that it is causally convex. Let $x_0,x_1\in D^+(\RR)$, and let $\gamma:\R\to\M$ be a causal, past-directed curve with $\gamma(0)=x_0, \gamma(1)=x_1$. We next prove that $y=\gamma(1/2)\in D^+(\RR)$. Since $x_0,x_1\in D^+(\RR)$, there exist $t_0\geq 0$, $t_1\geq 1$ such that $\gamma(t_0), \gamma(t_1)\in \RR$. Suppose that $t_0\geq \frac{1}{2}$. In that case, starting on $y$, any past-inextendible curve $\gamma'(t)$ with $\gamma'(1/2)=y$ will hit $\RR$, since the curve $\{\gamma(t):t\in [0,\frac{1}{2}]\}\cup \{\gamma'(t):t\geq\frac{1}{2}\}$ is a past-directed curve starting from $x_0$ that has not hit $\RR$ by $t=\frac{1}{2}$. This implies that $y\in D^+(\RR)$. Conversely, suppose that $t_0< \frac{1}{2}$. Then, $y$ is on a time-like curve that connects $\gamma(t_0)\in\RR$ with $\gamma(t_1)\in\RR$. Since $\RR$ is a region, this implies that $y\in \RR\subset D^+(\RR)$.
\end{proof}

\begin{prop}\label{proposition:psniplus}
Let $\OO$ be a compact subset of a globally hyperbolic spacetime $\M$ and consider arbitrary regions $\RR_\pm \subset \OO^\pm$ with $\overline{\RR}_+ \subset D(\RR_-)$. Then, there exists a region $\T_+ \subset \OO^+$ with $\overline{\T}_+ \subset D^+(\RR_-)$ such that $\RR_+ \subset D(\T_+)$.
\end{prop}
\def\RRp{\mathcal{T}}
\begin{proof}
    Since $\overline{\RR}_+ \subset D(\RR_-)$ it follows that there is $\RR'_- \subset \OO^-$ with $\overline{\RR'_-} \subset \RR_-$ and $\RR_+ \subset D(\RR'_-)$. Then, let $\RR_{+\pm} = \RR_+ \cap D^\pm(\RR'_-)$ and $\RRp_{+\pm} = J^+(\RR_{+\pm}) \cap D^+(\RR'_-)$. Clearly, $\RRp_{+\pm} \subset J^+(\RR_{+\pm}) \subset J^+(\RR_+) \subset \OO^+$ and $\RRp_{+\pm} \subset D^+(\RR'_-)$, thus $\overline{\RRp_{+\pm}} \subset D^+(\RR_-)$. It is also clear that $\RR_{++} \subset \RRp_{++} \subset D(\RRp_{++})$. Now, for $\RR_{+-}$, if $\gamma$ is a causal future-inextendible curve starting from a point in $\RR_{+-}$ then it will intersect $\RR'_-$ by definition and the intersection point will lie in $\RRp_{+-}$. Thus, it also holds that $\RR_{+-} \subset D(\RRp_{+-})$. We next define $\tilde{\T}_+:=\RRp_{+-} \cup \RRp_{++}$, $\T_+ := \chull \left(\tilde{\T}_+\right)$. Since $\OO^+, D^+(\RR_-)$ are regions (the last one, by Proposition \ref{prop:D_plus}), the relations $\tilde{\T}_+\subset \OO^+$ and $\overline{\tilde{\T}_+} \subset D^+(\RR_-)$ imply that $\T_+\subset \OO^+$ and $\overline{\T_+} \subset D^+(\RR_-)$. Finally, we have that 
    \begin{equation}
    \RR_+=\RR_{++}\cup \RR_{+-}\subset D(\RRp_{+-}) \cup D(\RRp_{++}) \subset D(\RRp_{+-} \cup \RRp_{++})\subset D(\T_+).
    \end{equation}
    The proposition is proven.
\end{proof}
\begin{proof}[Proof of Lemma \ref{lemma:psni_simpler}]
If $\Omega$ is a causal channel, then it trivially satisfies the conditions of the lemma. Conversely, suppose that $\Omega$ satisfies the conditions of the lemma, and let $\RR_\pm\subset\OO^\pm$, with $\overline{\RR_+}\subset D(\RR_-)$. By Proposition \ref{proposition:psniplus}, there exists a region $\T_+\subset \OO^+$ with $\overline{\T_+}\subset D^+(\RR_-)$, $\RR_+\subset D(\T_+)$. Then, on one hand we have that $\Omega(\A(\T_+)) \subset \A(\RR_-)$ by assumption. On the other hand, $\A(\RR_+) \subset \A(\T_+)$ by the timeslice condition and isotony, which implies $\Omega(\A(\RR_+))\subset \Omega(\A(\T_+))$. Hence, $\Omega(\A(\RR_+)) \subset \A(\RR_-)$. Since $\RR_{\pm}$ were arbitrary, $\Omega$ is causal.
\end{proof}

To prove Lemma \ref{lemma:weyllocality}, we start by showing that the Kraus operators of normal instruments local to a compact $\OO$ can be localized in any (connected) region containing $\OO$.
\begin{prop}\label{prop:normallocality}
Let $\Omega$ be a normal instrument, i.e., of the form\footnote{In case $\A$ is not just a $\star$-algebra but also a von-Neumann algebra, we may allow for sums over countably many items.} $\Omega_a(\bullet)=\sum_j K_a^j \bullet (K_a^j)^*$ with $K_a^j \in \A$. Then, $\Omega$ is local to a compact $\OO\subset \M$ iff $K_a^j \in \A(\OO^\perp)^\perp,\forall a,j$. In particular, if $K_a^j \in \A(\OO),\forall a,j$, then $\Omega$ is local to $\OO$. If $\Omega$ is local to $\OO$ and $\A$ has the Haag property, then $K_a^j \in \A(\RR)$ for any connected region $\RR \supset \OO$.
\end{prop}
\begin{proof}
To shorten notation, let us write $n=(j,a)$ and $K_n := K_a^j$.

By microcausality it is clear that if $K_n \in \A(\OO^\perp)^\perp$ for all $n$, then $\Omega$ acts trivially on $\A(\OO^\perp)$ and is therefore local to $\OO$. Further, $K_n \in \A(\OO)$ is clearly sufficient for this conclusion since $\A(\OO)\subset \A(\OO^\perp)^\perp$ and, if we already know that $K_n \in \A(\OO^\perp)^\perp$, then $K_n \in \A(\RR)$ for any connected region $\RR \supset \OO$ follows by the Haag property. 

It remains to check $K_n \in \A(\OO^\perp)^\perp$ when assuming that $\Omega$ is local to $\OO$. Let $\RR$ be a region with $\RR\perp \OO$, and let $R \in \A(\RR)$ be arbitrary. It follows that
\begin{align}
    & \sum_n (K_n^\ast R - R K_n^\ast)^\ast (K_n^\ast R - R K_n^\ast) \nonumber \\
    &=\sum_n ( R^\ast K_n K_n^\ast R + K_n R^\ast R K_n^\ast - R^\ast K_n R K_n^\ast - K_n R^\ast K_n^\ast R) \nonumber \\
    &= 0
\end{align}
since $1$, $R$, $R^\ast$, and $R^\ast R$ are in $\A(\RR) \subset \A(\OO^\perp)$ implying that $\overline{\Omega} = \sum_n K_n \bullet K_n^\ast$ acts trivially on them. Thus, $K_n^\ast R - R K_n^\ast = 0$ or, equivalently,  $[K_n,R] = [K_n^\ast, R] = 0$ for all $j$ and all $R \in \A(\RR)$. Since $\RR\perp\OO$ was arbitrary, we conclude that $K_n$ and $K_n^\ast$ commute with all elements from $\A(\OO^\perp)$.
\end{proof}

\begin{proof}[Proof of Lemma \ref{lemma:weyllocality}]
If $\Omega$ is local to $\OO$, applying Proposition \ref{prop:normallocality} to $\Omega$, we have that any Kraus operator $K_a^j$ thereof must commute with $\A(\OO^\perp)$. In the setting of free scalar fields, as given in \Cref{subsec:freeqft}, this implies that $K_a^j\in \A(\RR)$ for any region $\RR\supset \OO$, connected or not.

Given a Weyl instrument $\Omega$, we can always choose its matrix representation $(Q,Z)$, with $Z=(Z_1,...,Z_A)$ to be such that, for all $j$, there exists $a$ such that $Z_a(j,j)\not=0$ (otherwise, we could simply remove the quadrature $Q_j$ from the matrix representation). In that case, Proposition \ref{prop:normallocality} implies that $[e^{iQ_j},e^{i\Phi(h)}]=0$, for all $h\in C_0^\infty(\OO^\perp)$. If $Q_j=\Phi(f_j)$, this implies that $\mbox{Supp}(f_j)\subset \OO$.

The lemma is proven.

\end{proof}

\section{Characterization of causal Weyl channels: proof of Theorem \ref{theo:causal}}
\label{app:causal_charact}
To prove the theorem, we need the following lemma.
\begin{lemma}
\label{lemma:causal}
Let $\OO\subset \M$ be a compact set, let $d\in C_0^\infty(\OO)$ and let $\RR_+\in \OO^+$ be a region. Then, there exists a region $\RR_{-}\subset \OO^-$, with $\overline{\RR_{+}}\subset D^+(\RR_{-})$, where $d$ is not localizable iff $J^-(\overline{\RR_{+}})\not\supset G_d$. 
\end{lemma}
\begin{proof}
Suppose that $G_d\subset J^-(\overline{\RR_+})$ and $\RR_{-}\subset \OO^-$, with $\overline{\RR_+}\subset D^+(\RR_{-})$. We next prove that $d$ is localizable in $\RR_{-}$. Let $\Sigma\subset \OO^-$ be a Cauchy surface such that $\Sigma\cap \RR_{-}$ is a Cauchy surface for $\RR_{-}$. From $\overline{\RR_+}\subset D^+(\RR_{-})$ we have that $J^-(\overline{\RR_{+}})\cap \Sigma$ is contained in $\RR_{-}$. Indeed, let $x\in \overline{\RR_+}$. If $x\prec \Sigma$, then $\Sigma\cap J^-(x)=\emptyset$, so let us assume that $x$ is in the future of $\Sigma$. Let $\gamma(t)$ be any past-directed, past-inextendible, time-like curve with $\gamma(0)=x$, and let $\tau$ be such that $y:=\gamma(\tau)\in \Sigma$. From $\overline{\RR^+}\subset D^+(\RR-)$ we have that there exists $\tau'\geq 0$ such that $\gamma(\tau')\in \RR_-$. Thus, there exists a time-like curve connecting a point in $\RR_-$ with $y\in \Sigma$. Since $\Sigma$ is a Cauchy surface for $\RR_-$, this implies that $y\in \RR_-$. Hence, $\RR_-\supset J^-(\overline{\RR_{+}})\cap \Sigma$.

Now, $G_d\subset J^-(\overline{\RR_+})$ implies that $\overline{G_d}\subset J^-(\overline{\RR_+})$, and so $\mbox{Supp}(Ed)\cap\OO^-\subset J^-(\overline{\RR_+})$. Consequently, 
\begin{equation}
\mbox{Supp}(Ed)\cap \Sigma=\mbox{Supp}(Ed)\cap\OO^-\cap \Sigma\subset J^-(\overline{\RR_+})\cap\Sigma\subset \RR_-.    
\end{equation}
By Lemma \ref{lemma:localization}, it follows that $d$ is localizable in $\RR_{-}$.

Suppose, on the contrary, that $G_d\not\subset J^-(\overline{\RR_+})$. We next construct $\RR_{-}\subset \OO^-$, with $\overline{\RR_+}\subset D^+(\RR_{-})$, such that $d$ is not localizable in $\RR_{-}$. Let $x\in G_d$, $x\not\in J^-(\overline{\RR_{+}})$ and let $\Sigma\subset \OO^-$ be any Cauchy surface that contains $x$. Then it holds that $x\not\in \overline{S_+}:=J^-(\overline{\RR_{+}})\cap \Sigma$, and hence there exists a region $\RR_-\subset \OO^-$, with $\RR_-\supset \overline{S_+}$, such that $x\not\in J(\overline{\RR_-})$. Since $x\in G_d$ and $G_d$ is an open set, there exists a neighborhood $B$ of $x$ such that $B\subset G_d$, with $B\perp \RR_-$. Choose $h\in C_0^\infty(B)$ such that $\langle h, Ed\rangle\not=0$. Then we have that $[\Phi(h),\Phi(d)]\not= 0$. Since $B\perp \RR_{-}$, this implies that $e^{i\Phi(d)}\not\in \A(\RR_{-})$ by Einstein's causality.

\end{proof}

\begin{proof}[Proof of Theorem \ref{theo:causal}]
We now prove the first implication of the theorem. Suppose that $\Omega$ is causal. Given any two partitions $P_1,P_2$, let us write $P_1>P_2$ if $P_1$ is a refinement on $P_2$, and, for any finite set of partitions $\P$, define $\mbox{Ref}(\P)$ as the set of maximally refined partitions within $\P$. Now, let $d\in {\cal D}$, with $d\not=0$. Let $\{A_1,...,A_n\}\in \mbox{Ref}({\cal P}(d))$, and let $\RR_{+}\subset\OO^+$, with $J^-(\overline{\RR_{+}})\not\supset G_d$, be such that, for some $\{x_{jk}\}_{j>k}\subset \RR_{+}$, it holds that $(Es)(x_{jk})\not=(Et)(x_{jk})$, for all $t\in A_j,s\in A_k$. Note that the maximality of $\{A_1,...,A_n\}$ implies that
\begin{equation}
(Es)(x)=(Et)(x),\forall s,t\in A_j, x\in \RR_{+},
\end{equation}
for all $j$; otherwise we could further refine the partition by breaking $A_j$. From the existence of $\{x_{jk}\}$, it is straightforward to construct smooth functions of compact support $\{h_{jk}\}_{j>k}$ such that $\mbox{Supp}(h_{jk})\subset \RR_{+}$, $w^j_{jk}\not=w^k_{jk}$, for all $j>k$, where
\begin{equation}
w^l_{jk}:=E(h_{jk},s_l)=\langle h_{jk}, E s_l\rangle,
\end{equation}
for some representative $s_l\in A_l$ of the subset $A_l$. For real $\alpha=(\alpha_{jk})_{j>k}$, define the family of test functions:
\begin{equation}
h(\alpha):=\sum_{j>k}\alpha_{jk}h_{jk}.
\end{equation}

Now, by Lemma \ref{lemma:causal}, there exists a region $\RR_{-}\subset \OO^-$, with $\overline{\RR_{+}}\subset D^+(\RR_{-})$ and where $\Phi(d)$ is not localizable. The causality of $\Omega$ implies, however, that $\Omega(e^{i\Phi(h(\alpha))})$ must be localizable in $\RR_{-}$. It thus follows that
\begin{equation}
\Omega^{(d)}(e^{i\Phi(h(\alpha))})=0;
\label{omega_d}
\end{equation}
otherwise, there would be a term of the form $e^{-i\Phi(h(\alpha)+d)}$, not localizable in $\RR_{-}$, in the Weyl decomposition of $\Omega(e^{i\Phi(h(\alpha))})$. Condition (\ref{omega_d}) can be rewritten as:
\begin{equation}
\sum_{l=1}^n e^{-i\vec{\alpha}\cdot \vec{w}^l}\sum_{s\in A_l}c^d_s=0,
\end{equation}
where we wrote $\vec{\alpha}\cdot \vec{w}^l:=\sum_{j>k}\alpha_{jk}w_{jk}^l$. Note that $\vec{w}^l\not=\vec{w}^m$, for $l>m$, because $w^l_{lm}\not=w^m_{lm}$; the functions (of $\alpha$) $\{e^{-i\vec{\alpha}\cdot \vec{w}^l}\}_l$ are thus linearly independent. Since the relation above holds for arbitrary $\alpha$, it follows that
\begin{equation}
\sum_{s\in A_l}c^d_s=0,\forall l.
\end{equation}
We have just proven the result when the partition $\{A_1,...,A_n\}$ is an element of $\mbox{Ref}(\P)$. To prove it for general $P\in \P$, it is enough to find $P'\in \mbox{Ref}(\P)$ with $P'>P$. Then that eq. (\ref{eq:null_sum}) holds for $P$ is a consequence of eq. (\ref{eq:null_sum}) holding for $P'$.

Let us now prove the opposite implication of the Theorem. Let eq. (\ref{eq:null_sum}) hold, let $\RR_{\pm}\subset \OO^\pm$ be regions such that $\overline{\RR_{+}}\subset D^+(\RR_{-})$ and let $h\in C_0^\infty(\RR_{+})$. For $d\in {\cal D}$, suppose that $J^-(\overline{\RR_{+}})\supset G_d$. In that case, by Lemma \ref{lemma:causal} $d$ can be localized in $\RR_{-}$, and thus so can $\Omega^{(d)}(e^{i\Phi(h)})$. Suppose, on the contrary, that $J(\overline{\RR_{+}})\not\supset G_d$. Then there exists a partition $\{A_1,...,A_n\}\subset \mbox{Ref}({\cal P}(d))$ such that, for $l=1,...,n$, it holds that
\begin{equation}
Es(x)=E s_l(x),\forall x\in \RR_{+}, \forall s\in A_l,
\end{equation}
for some representative $s_l\in A_l$ of the subset $A_l$. From equation (\ref{expr_omega_d}), we therefore infer that
\begin{align}
\Omega^{(d)}\left(e^{i\Phi(h)}\right)=\sum_{l=1}e^{-i\langle h, E s_l\rangle}\sum_{s\in A_l}c_s^d.
\end{align}
The last sum equals zero by (\ref{eq:null_sum}).

In sum, for all $d\in {\cal D}$ it holds that $\Omega^{(d)}\left(e^{i\Phi(h)}\right)$ is localizable in $\RR_{-}$. So is $\Omega\left(e^{i\Phi(h)}\right)=\sum_{d\in{\cal D}} \Omega^{(d)}\left(e^{i\Phi(h)}\right)$. Since $\RR_{\pm}, h$ were arbitrary, it follows that  $\Omega$ is causal.

\end{proof}

\section{Composition of causal Weyl channels: proof of Theorem \ref{theorem:causal_network}}
\label{app:causal_networks}
To prove Theorem \ref{theorem:causal_network}, we first demonstrate a useful geometric lemma.
\begin{lemma}
\label{lemma:causal_ordering}
Given compact causally orderable subsets $\OO_1 \prec \ldots \prec \OO_n$ of a globally hyperbolic spacetime $\M$, denote $\OO = \cup_j \OO_j$. Then, given regions $\RR_{\pm} \subset \OO^\pm$ satisfying $\overline{\RR_{+}} \subset D^+(\RR_{-})$, there exist regions $\{\RR^j_\pm\}_{j=1,\ldots,n}$ such that:
\begin{enumerate}
    \item $\overline{\RR^j_+} \subset D^+(\RR^j_-),\forall j$,
    \item $\RR_\pm^j \subset \OO^\pm_j,\forall j$,
    \item $\RR^j_{-} = \RR^{j-1}_+,\forall j$ with $\RR^n_+ := \RR_+$ and $\RR^1_- := \RR_-$.
\end{enumerate}    
\end{lemma}

\begin{proof} There is nothing to prove for $n=1$. Take $n=2$, then we have to choose a region $\RR:= \RR^1_+ = \RR^2_-$ such that $\overline{\RR} \subset \T:= D^+(\RR_-) \cap \OO_1^+ \cap \OO_2^-$ and $\overline{\RR_+} \subset D^+(\RR)$. For the proof we note that since $\OO_1 \prec \OO_2$, there exists a Cauchy surface $\Sigma \subset \OO_1^+ \cap \OO_2^-$ and $\Sigma_R := J^-(\overline{\RR_+}) \cap \Sigma$ is contained in $\OO_1^+ \cap \OO_2^-$ since $\Sigma$ is. Moreover, $\Sigma_R\subset D^+(\RR_-)$ by the following argument: as $\overline{\RR_+} \subset D^+(\RR_-)$, any past-inextendible causal curve $\gamma$ emanating from $\overline{\RR_+}$ must hit $\RR_-$. Since $\Sigma$ is a Cauchy surface, $\gamma$ must hit $\Sigma$ exactly once and this within $\Sigma_R$ since $\gamma$ is causal. Then, we note that $\T$ is a globally hyperbolic spacetime in its own, and $\Sigma_R$ is a compact subset contained in it. By Lemma $6$ of \cite{impossible_meas} (it is easily extended to finitely many connected components), there is a region $\RR$ containing $\Sigma_R$ with $\overline{\RR} \subset \T$. Finally, we note $\overline{\RR_+} \subset D^+(\Sigma_R) \subset D^+(\RR)$. It remains to prove that the statement holds for $n+1$ if it holds for $n$ for all suitable regions $\RR_\pm$. For this we simply define $\OO' = \cup_{j=2}^{n+1} \OO_j$ and note that we get the regions $\RR_\pm^1$ by applying the statement for $n=2$ to $\OO_1 \prec \OO'$ (with respect to $\RR_\pm$) and the regions $\RR_\pm^j, j=2,\ldots,n+1$ by applying the induction hypothesis to $\OO_2 \prec \ldots \prec \OO_{n+1}$ (with respect to $\RR_+$ and $\RR^1_-$).  
\end{proof}

\begin{proof}[Proof of Theorem \ref{theorem:causal_network}]
Let $\RR_\pm\subset \OO^\pm$, such that $\overline{\RR_{+}}\subset D^+(\RR_{-})$ and let 
\begin{equation}
\overline{\Omega}(\bullet):=\sum_{c} \Omega^{1,c(\wedge_1)}_{c_1} \circ \ldots \circ \Omega^{n,c(\wedge_n)}_{c_n}(\bullet),
\end{equation}
denote the measurement channel associated with $\Omega$, whereas $\overline{\Omega}^{j,c(\wedge_j)} := \sum_{c_j} \Omega^{j,c(\wedge_j)}_{c_j}$ denotes the measurement channel associated with $\Omega^{j,c(\wedge_j)}$. 

We next prove that, for $O\in \A(\RR_{+})$, it holds that $\overline{\Omega}(O)\in \A(\RR_-)$. First, given $\OO_1,...,\OO_n$, we invoke Lemma \ref{lemma:causal_ordering} to produce regions $\{\RR_{\pm}^j\}_{j=1}^n$ satisfying the constraints
\begin{enumerate}
    \item $\overline{\RR^j_+} \subset D^+(\RR^j_-),\forall j$,
    \item $\RR_\pm^j \subset \OO^\pm_j,\forall j$,
    \item $\RR^j_{-} = \RR^{j-1}_+,\forall j$ with $\RR^n_+ := \RR_+$ and $\RR^1_- := \RR_-$.
\end{enumerate}  

Second, define the operators $\{O_j\}_j$ also by induction, as:
\begin{align}
&O_n=O,\nonumber\\
&O_{j-1}=\sum_{c_j} \Omega_{c_j}^{j,c(\wedge_j)} (O_j), \label{eq:Odef}
\end{align}
where it is understood that each operator $O_j$ depends on the indices $c_1,...,c_{\mathrm{j}}$ and, clearly, one obtains $O_0 = \overline{\Omega}(O)$.

Finally, define the sets $(\AA_j)_{j=1}^n \subset \{ 1,\ldots,n\}$ by induction, as:
\begin{align}
\AA_n=&\emptyset,\nonumber\\
\AA_{j-1}=& \left\lbrace \begin{matrix} \AA_j\setminus\{j\}, &\quad \mbox{ if }\OO_j\perp \RR_{+}^{j},\\
 (\AA_j\cup \wedge_j)\setminus\{j\}, &\quad \mbox{ otherwise}. \end{matrix} \right.  \label{cindexrecursion}
\end{align}

We claim that $O_j$ depends only on $c(\AA_j)$ and that $O_j \in \A(\RR^{j+1}_-)$, for all $j$, which implies that $O_0\in \A(\RR_{-})$ and thus the causality of $\Omega$. We prove the claim by induction on $j$. The initial case $j=n$ is obviously satisfied. Now, let us assume that the claim holds beyond $j-1$ and show that it also holds for $j-1$. 

By the induction hypothesis, we have that
\begin{equation}
O_{j-1}=\sum_{c_j}\Omega^{j,c(\wedge_j)}_{c_j}(O_j),
\end{equation}
with $O_j \in\A(\RR^{j}_+) = \A(\RR^{j+1}_-)$ depending only on $c(I_j)$. There are two possibilities:
\begin{enumerate}
    \item $j\not\in\AA_j$. In that case, $O_j$ does not depend on $c_j$, which means that $O_{j-1}$ is the result of acting on $O_j$ with the causal measurement channel $\overline{\Omega}^{j,c(\wedge_j)}$. Thus, $O_{j-1}$ must be localizable in $\RR^{j}_-$, since $\overline{\RR^{j}_{+}} \subset D^+(\RR^{j}_-)$. 
    \item $j\in\AA_j$. In that case, there must have existed $k>j$, with $\OO_k\not\perp \RR_+^{k}$, thus $\RR_+^{k}\cap J(\OO_k)\not=\emptyset$, and $j\in \wedge_k$, hence $\OO_k \subset \OO_j^\vee$. This implies that $\RR_+^k\cap \OO_j^\vee \supset \RR_+^k \cap J^+(\OO_k) \not=\emptyset$. Since $\overline{\RR_+^k} \subset D^+(\RR_-^j)$ and $\OO_j\cap J^-(\RR^j_-)=\emptyset$, it follows that $\OO_j \subset D^+(\RR_-^j) \subset D(\RR_-^j)$. By Proposition \ref{prop:D_plus}, $D^+(\RR^j_-)$ is a region: this allows us to invoke Lemma \ref{lemma:weyllocality} and conclude that the Kraus operators of $\Omega^{j,c(\wedge_j)}$ are localizable in $D^+(\RR_-^j)$ and, by the timeslice property, also in $\RR_-^j$. Consequently, $\Omega^{j,c(\wedge_j)}_{c_j}(O_j)\subset \A(\RR_-^j)$, for all $c_j$, which implies that $O_{j-1}$ is localizable in $\RR_-^j=\RR_+^{j-1}$.
\end{enumerate}
In both cases, if $\OO_j\perp \RR^j_+$, instrument $\Omega^{j,c(\wedge_j)}$ will act trivially on $O_j$, and so $O_{j-1}$ will not inherit the indices $c(\wedge_j)$ that the instrument carries. Otherwise, $O_{j-1}$ will depend, in general, on $c(\wedge_j)$ and on those indices $O_j$ depends on. Either way, $O_{j-1}$ only depends on the indices $c(\AA_{j-1})$.
\end{proof}

We remark that the proof may be generalized to networks of normal instruments acting on an AQFT: The only required modification is to argue localizability of the Kraus operators within $D(\RR_-^j)$ in Item 2. This follows by the Haag property if $D(\RR_-^j)$ and thus $D(\RR_-)$ are connected. Following the proof in \cite{Jubb_2022} that the PSNI condition implies causality, one sees that it is sufficient to require the PSNI condition to hold for connected regions $\RR^\pm$.

\section{Generic test functions}\label{app:generic}

In this section we prove that a large class of QFTs is generic or null-generic. We also prove that (null) generic test functions in (null) generic QFTs are ubiquitous or, in other words, `generic'. Our main results are:

\begin{prop}\label{prop:generic}
Let $\Psi$ be a Green-hyperbolic QFT with commutator function $E:C_0^\infty(\M)\to C^\infty(\M)$ of the form:
\begin{equation}
E(x,y) = \tilde{E}(x,y)I(x,y) + H(x,y),
\label{cond_analytic}
\end{equation}
 where, for all $y \in \M$, $\tilde{E}(\bullet,y)$ is an entire function, $I$ is an indicator function satisfying 
 \begin{equation}
I(x,y) = \left\lbrace \begin{matrix} \pm 1, & \text{for }x \in J^\pm(y), x\neq y \\
0, & \text{otherwise}. \end{matrix} \right.,
\end{equation}
and $H(x,y)$ is a distribution that vanishes if $x \not\in N(y)$. 

Let $E$ furthermore satisfy the following property: for any open $Y \subset \M$, it holds that
\begin{equation}
\mathrm{dim} \left[ \mathrm{span} \left( \{ \tilde{E}(\bullet, y): y\in Y\} \right)\right] = \infty.
\label{cond_li}
\end{equation}

Then, $\Psi$ is generic.
\end{prop}

\begin{prop}\label{prop:null-generic}
Let $\Psi$ be a scalar Green-hyperbolic QFT with a commutator function $E:C_0^\infty(\M)\to C^\infty(\M)$ that is singular along null directions, namely, its wavefront set $WF(E)$ satisfies:
\begin{equation}\label{eq:wavefront}
WF(E) \supset ({\cal N}^+ \times {\cal N}^-)_r \cup ({\cal N}^- \times {\cal N}^+)_r,
\end{equation}
for the r.h.s. being nonempty; here ${\cal N}^\pm$ are the nonzero null future/past-directed covectors associated with $\M$ and the subscript $r$ indicates that we consider only such vectors $(x,k_x,y,-k_y)$ where $x,y \in \M$ are connected by a null geodesic segment with cotangent vector $k_x$ at $x$ and parallel-transported cotangent vector $k_y$ at $y$.

Then, $\Psi$ is null-generic.
\end{prop}

In Appendix \ref{app:generic_minkowski} it is proven that the massive Klein-Gordon field on 1+1- and 1+3-dimensional Minkowski spacetime satisfies conditions (\ref{cond_analytic}) and (\ref{cond_li}). Condition (\ref{eq:wavefront}) is satisfied by the Klein-Gordon theory in $d \geq 2$ spacetime dimensions and more generally all commutator functions associated with normally hyperbolic equations of motion. For these \eqref{eq:wavefront} is an equality (see e.g. \cite[Thm.~16]{QFTCS} for the Klein-Gordon case and \cite[Thm.~5.1]{Greenhyperbolic} for generality). As a corollary, we have that:
\begin{corollary}
The massive Klein-Gordon fields in 1+1 and 1+3-dimensional Minkowski spacetimes are generic QFTs. Klein-Gordon fields in $d\geq 2$ spacetime dimensions and arbitrary spacetimes are null generic.
\end{corollary}

In the following lines, we prove Propositions \ref{prop:generic} and \ref{prop:null-generic}. To this end, given an implicit commutator function $E$ and a set $X\subset \M$, it is convenient to introduce the map $\xi_X:C^\infty_0(\mathcal M)\rightarrow C^\infty(X)$, defined as:
\begin{equation}\label{eq:xi}
f \mapsto \xi_X(f) := \left. Ef\right|_{X}.
\end{equation}
The next notion will also be very handy.
\begin{defin}
A free scalar QFT with commutator function $E$ is (null) two-site generic if, for any $n\in\N$, any regular compact $\K\subset \M$ and any open set $V\subset \M$, with $V\cap J(\K)\not=\emptyset$ ($V\cap N(\K)\not=\emptyset$), there exist $f_1,...,f_n\in C_0^\infty(\K)$ such that
\begin{equation}
\{\xi_{V}(f_j)\}_j
\end{equation}
are linearly independent functions.
\end{defin}

In Appendix \ref{app:generic_two_site}, it is proven that any free scalar QFT satisfying conditions (\ref{cond_analytic}), (\ref{cond_li}) is two-site generic. Correspondingly, any free scalar QFT whose commutator function is singular along null directions [eq. (\ref{eq:wavefront})] is null two-generic. With this result, Propositions \ref{prop:generic} and \ref{prop:null-generic} are a consequence of the next two lemmas.

\begin{lemma}
\label{lemma:two_site_generic_density}
Let $\Phi$ be a (null) two-site generic QFT. Then, for any vector of regular compacts $\K:=(\K_j)_{j=1}^m$ and any open set $\V\subset \M$ with $V\cap J(\K_j)\not=\emptyset$ ($V\cap N(\K_j)\not=\emptyset$), for $j=1,...,m$, it holds that the set $\mathbb{T}^V_{\K}$ of $m$-tuples of test functions
\begin{equation}
\mathbb{T}^V_{\K}:=\{(f_j\in C^\infty_0(\K_j))_j:\xi_V(f_1),...,\xi_V(f_m),\mbox{ l.i.}\}    
\end{equation}
is dense in $C_0^\infty(\K_1)\times...\times C_0^\infty(\K_m)$.
\end{lemma}
\begin{proof}
By two-site (null) genericity, for each $j=1,...,m$, there exist $j$ functions $f_j^1,...,f_j^j\in C_0^\infty(\K_j)$ such that 
\begin{equation}
\{\xi_V(f^k_j):k=1,...,j\}
\end{equation}
are linearly independent. By taking linear combinations thereof, we redefine the functions $f^1_j,...,f^j_j$ such that the equation above represents an orthonormal set of functions in ${\cal L}^2(V)$.

Next, for arbitrary $\epsilon>0$, and given functions $(g_j\in C_0^\infty(\K_j))_j$, we next perturb the $g$'s with the $f$'s by an amount $\epsilon$ to arrive at some $\tilde{g}$'s such that $\{\xi_V(\tilde{g}^j)\}_j$ are linearly independent. That will prove the lemma's claim.

We start with $j=1$. If $\xi_V(g_1)\not=0$, then our perturbation is $\delta g^1=0$; otherwise, we set $\delta g^1=\epsilon f_1^1$. Given the perturbations $\delta g_1,...\delta g_{j-1}$, we decide the next perturbation $\delta g_j$ as follows: if $\xi_V(g_j)\not\in G_j:=\mbox{span}\{\xi_V(g_k+\delta g_k)\}_{k=1}^{j-1}$, then we set $\delta g_j=0$. Otherwise, there is a linear combination $f_j$ of $\{f^l_j\}_l$ (with bounded coefficients), with $\|f_j\|=1$ (we take the natural norm in $L^2(V)$), that is orthogonal to $G_j$, and hence to $\xi_V(g_j)\in G_j$. Choosing $\delta g_j=\epsilon f_j$, we therefore make sure that $\xi_V(g^j+\delta g^j)\not\in G_j$. 

Iterating till $j=m$, we therefore construct functions $(\tilde{g}_j:=g_j+\delta g_j\in C_0^\infty(K_j))_j$ with $\{\xi_V(g_j)\}_j$ linearly independent. Since both the magnitude $\epsilon$ of the perturbation and the test functions $g_1,...,g_m$ were arbitrary, it follows that $\mathbb{T}^V_K$ is indeed dense in $C_0^\infty(\K_1)\times...\times C_0^\infty(\K_m)$.
    
\end{proof}

\begin{lemma}
\label{lemma:density_generic}
Let $\Phi$ be a two-site (null) generic QFT. Then, $\Phi$ is (null) generic. Moreover, for any vector of regular compacts $\K:=(\K_j)_{j=1}^m$ the set of generic functions with respect to $\K$ is dense in $C_0^\infty(\K_1)\times...\times C_0^\infty(\K_m)$.
\end{lemma}
\begin{proof}
Define $\mathbb{V}:=C_0^\infty(\K_1)\times...\times C_0^\infty(\K_m)$, and, for $j=1,...,m$, let $\{V^k_j\}_k$ be any countable basis of open sets in $\M$ satisfying $V^k_j\cap J(\K_j)\not=\emptyset$ ($V^k_j\cap N(K_j)\not=\emptyset$). Note that any open set $V\subset \M$ causally connected to $\{\K_j\}_j$ (by null curves) must contain the open set $V(\vec{k}):=\cup_{j=1}^mV^{k_j}_j$, for some $\vec{k}\in \N^m$.

By Lemma \ref{lemma:two_site_generic_density}, we have that, for any $\vec{k}\in \N^m$, the set of functions $\mathbb{T}^{V(\vec{k})}_{\K}$ is dense in $\V$. By the Baire category theorem (see e.g. \cite[Section 2]{baire}), it then follows that
\begin{equation}
\{ g \in \mathbb{V} : g\text{ generic} \} = \bigcap_{\vec{k}\in \N^m} \mathbb{T}^{V(\vec{k})}_{\K}
\end{equation}
is similarly dense in $\mathbb{V}$, whereas its complement is nowhere dense.
\end{proof}

\begin{remark}
\label{remark:density}
(Null) generic QFTs are also two-site (null) generic. Hence, by the last lemma we have that, in any (null) generic QFT, the set of  (null) generic test functions is dense among the set of test functions.
\end{remark}

\subsection{Two scalar QFTs satisfying conditions (\ref{cond_analytic}), (\ref{cond_li})}
\label{app:generic_minkowski}
Here we will prove that the Klein-Gordon field in Minkowski spacetime in dimensions $1+1$ and $1+3$ satisfies the conditions of Proposition \ref{prop:generic}. 

First, it is immediate to see that the commutator functions of scalar QFTs in $1+1$- or $1+3$-dimensional Minkowski spacetime satisfy condition (\ref{cond_analytic}), for some $\tilde{E}(x,y)=\hat{E}(x-y)$. Indeed, the corresponding functions $\hat{E}$ are, respectively:
\begin{equation}
\hat{E}_{1+1}(t,x)=J_0(m\sqrt{t^2-x^2}), \hat{E}_{1+3}(t,x)=\frac{J_1(m\sqrt{t^2-x^2})}{\sqrt{t^2-x^2}}.
\label{commutator_func}
\end{equation}
Since $J_0$ ($J_1$) is an even (odd) analytic function, it follows that the above two functions admit an analytic continuation to all of $\M$.

For such translation-invariant kernels $\tilde{E}(x,y)=\hat{E}(x-y)$, the next lemma provides a simple, sufficient criterion to certify condition (\ref{cond_li}).

\begin{lemma}
Let $\hat{E}$ be a real-valued entire function on Minkowski spacetime $\mathbb{M}$ and let there be a nonzero vector $v \in \mathbb{M}$ such that, for all $z\in\mathbb{M}$, $f: \R \to \R, s \mapsto \hat{E}(z+s v)$ is not a finite sum of products of exponentials and polynomials in $s$. Then, condition (\ref{cond_li}) holds for $\tilde{E}(x,y)=\hat{E}(x-y)$.
\end{lemma}
\begin{proof}
Given open sets $X,Y\subset \mathbb{M}$, choose $x\in X$, $y\in Y$, and let $\lambda\in \R,\Lambda > 0$ be such that, for $\lambda\in (-\Lambda,\Lambda)$, $y+\lambda s\in Y$. Assume that the set of functions
\begin{equation}
\{f_\lambda(s):=\hat{E}(x-y+(s-\lambda)v):\lambda\in (-\Lambda,\Lambda)\},
\label{space_func}
\end{equation}
is finite dimensional, i.e., there exist functions $g_1,...,g_n:\R \to\R$, spanning the function space $\mathbb{F}$, such that
\begin{equation}
f_\lambda \in \mathbb{F},\forall \lambda\in (-\Lambda,\Lambda).
\end{equation}
In that case, the differentials of $f_\lambda(s)$ with respect to $s$ are also in $\mathbb{F}$. Indeed, for $\lambda\in (-\Lambda,\Lambda)$ we have that
\begin{equation}
\frac{\partial f_\lambda(s)}{\partial s}=\lim_{\epsilon\to 0}\frac{f_\lambda(s+\epsilon)-f_\lambda(s)}{\epsilon}=\lim_{\epsilon\to 0}\frac{f_{\lambda-\epsilon}(s)-f_\lambda(s)}{\epsilon}\in \mathbb{F},
\end{equation}
where the last relation follows from the fact that, for sufficiently small $\epsilon$, we have $\lambda-\epsilon \in (-\Lambda,\Lambda)$ and that $\mathbb{F}$ is closed with respect to pointwise limits. Higher order derivatives (well-defined, due to $\hat{E}$ being entire) can be similarly expressed as a limit of linear combinations of $f_\lambda$'s and thus satisfy $\frac{\partial^k f_\lambda}{\partial s^k}\in \mathbb{F}$. Since $\mathbb{F}$ is $n$-dimensional, though, there must exist coefficients $c_0,...,c_n$ such that 
\begin{equation}
\sum_{k=0}^n c_k\frac{\partial^k f_\lambda(s)}{\partial s^k}=0.
\end{equation}
This would imply that $f_0(s)$ is a finite sum of products of exponentials and polynomials in $s$, which contradicts the assumption of the lemma. It follows, therefore, that (\ref{space_func}) is infinite dimensional. So is then the set of functions $\{f^y:X\to\R\}_{y\in Y}$, with
\begin{equation}
f^{y}(x):=\hat{E}(x-y+sv)),
\end{equation}
and thus so is the set of functions in (\ref{cond_li}).
\end{proof}

To see that the functions in (\ref{commutator_func}) cannot be expressed as finite linear combinations of exponentials multiplied by polynomials, it suffices to evaluate them with complex variables and invoke the asymptotic properties of Bessel functions for large arguments:
\begin{equation}
J_\alpha(z)\sim \sqrt{\frac{2}{\pi z}}\left(\cos\left(z-\frac{\alpha \pi}{2}-\frac{\pi}{4}\right)+e^{|\mbox{Im}(z)|}O\left(|z|^{-1}\right)\right).
\end{equation}
Finite linear combinations of exponentials multiplied by polynomials, if they decay, they do so exponentially, not as $O(z^{-1/2})$. It follows that both $\hat{E}_{1+1}$ and $\hat{E}_{1+3}$ satisfy the conditions of the lemma.

\subsection{Two-site (null) genericity}
\label{app:generic_two_site}
In this section, we will prove that all QFTs satisfying the conditions of Propositions \ref{prop:generic} or \ref{prop:null-generic} are, respectively, two-site generic and two-site null-generic QFTs.

\begin{lemma}
\label{lemma:n_indep}
Let $\Phi$ be a scalar QFT with commutator function $E$ satisfying conditions (\ref{cond_analytic}), (\ref{cond_li}). Then, $\Phi$ is two-site generic. 
\end{lemma}
\begin{proof}
Let $\K$ be a regular, compact set, and let $V\subset \M$ be an open set with $V\cap J(\K)\not=\emptyset$. Then, there exist open sets $X\subset V, Y\subset \K$ such that either $X\subset Y^\vee\setminus N(Y)$ or $X\subset Y^\wedge\setminus N(Y)$. Indeed, from $V\cap J(\K)\not=\emptyset$ there exists $y\in \mbox{int}(\K)$ such that either $J^+(y)\cap (V\setminus N(y))\not=\emptyset$ or $J^-(y)\cap (V\setminus N(y))\not=\emptyset$. Let us assume the former possibility, which can be rewritten as $\{y_j\}^\vee\cap (V\setminus N(y))\not=\emptyset$. Since the first set is closed; and the second, open, there exists a neighborhood $Y$ of $y$ such that $Y\subset \K$, $Y^\vee\cap (V\setminus N(y))\not=\emptyset$. Thus, there exists an open set $X\subset V$ with $X\subset Y^\vee\setminus N(Y)$. Had we assumed $J^-(y)\cap V\setminus N(y)\not=\emptyset$, we would have similarly concluded that $X\subset Y^\wedge\setminus N(Y)$.

Now, by condition \eqref{cond_li}, there exist $y^1,\ldots,y^n \in Y$ such that 
    \begin{equation}
        \{ \tilde{E}(\bullet,y^j) \}_{j=1}^n
    \end{equation}
    are linearly independent. Since these functions are entire \eqref{cond_analytic}, their restrictions to $X$,
    \begin{equation}
        \{ \left. \tilde{E}(\bullet,y^j)\right\rvert_{X} \}_{j=1}^n,
    \end{equation}
    are also linearly independent. Thus, there exist functions $g^1,...,g^n\in L^2(X)$ such that the $n\times n$ matrix $M$ with entries
    \begin{equation}
    M_{jk}:=\int_X \mathrm{d}x g^j(x) \tilde{E}(x,y^k)
    \end{equation}
    satisfies $M=\id_n$. In particular, $\mbox{det}(M)=1$. 

For $j=1,...,n$, let $(f^j_k)_k$ denote a sequence of smooth functions with support within $Y$ such that (in the distributional sense) $\lim_{k\to\infty}f^j_k(y)=\delta(y-y^j)$. Then we have that
\begin{equation}
    \lim_{k\to\infty} \mbox{det}(M(\vec{f}_k)) =1,
\end{equation}
where, for any set of smooth functions $s^1,...,s^n$, the matrix $M(\vec{s})$ is defined as $M_{jk}(\vec{s}):=\langle g^j, \tilde{E}s^k\rangle$. Thus, for sufficiently large $k$, it follows that $\vec{f}:=\vec{f}_k$ satisfies $\mbox{det}(M(\vec{f}))\not=0$ and hence 
\begin{equation}
\{\left.\tilde{E}f^j\right|_{X}\}_{j=1}^n
\end{equation}
are linearly independent. Finally, since $\mbox{supp}(f^j)\subset Y$, it follows, by eq. (\ref{cond_analytic}), that 
\begin{equation}
\left.\tilde{E}f^j\right|_{X}= \left\lbrace \begin{matrix} +\left.Ef^j\right|_{X},&\mbox{ if } X\subset Y^\vee\setminus N(Y)\\
-\left.Ef^j\right|_{X},&\mbox{ if } X\subset Y^\wedge\setminus N(Y). \end{matrix}\right.
\end{equation}
Therefore, $\{ \left.Ef^j\right|_{X} \}_{j=1}^n$, and hence $\{ \left.Ef^j\right|_{V} \}_{j=1}^n$, forms a linearly independent set. 
\end{proof}

\begin{lemma}\label{lem:nullgeneric}
Let $\Phi$ be a scalar QFT whose commutator function $E$ is singular along null directions. Then, $\Phi$ is a two-site null-generic QFT.
\end{lemma}

\begin{proof}
     Let $E$ be singular along null directions. Assuming that $\operatorname{Ran} \xi_V$ is finite-dimensional implies that $\xi_V$ is finite-rank, i.e., for any open set $W \subset K$ there exist $u_j \in C^\infty(V)$ and $v_j \in {\cal D}'(W)$ (smooth distributions) such that
    \begin{equation}
        E(x,y) = \sum_{j=1}^l u_j(x) v_j(y), \quad \forall x \in V, y \in W.
    \end{equation}
    But then, by \cite[Theorem 8.2.9]{Hrmander1990TheAO}
    \begin{align} \label{eq:wfe}
        WF(\left. E\right|_{V \times W}) \subset &\bigcup_{j=1}^l \bigg(\left(WF(u_j) \times WF(v_j) \right)\cup (\supp{u_j}\times \{0\}\times WF(v_j)) \nonumber\\ &\cup (WF(u_j)\times \supp (v_j)\times \{0\})\bigg)\subset \{ (x,0; y, \eta) | \, x \in V, (y,\eta) \in T^\ast W \}
    \end{align}
    using that $WF(u_j) = \emptyset$. However, since $V \cap N(K) \neq \emptyset$ and $V$ is open, also $V \cap N(W) \neq \emptyset$, choosing $W$ to be the interior of $K$. Thus, in this case, $V$ and $W$ are connected by a null geodesic segment and, by the assumption that $E$ is singular along null directions, the wavefront set of the restriction of $E$ to $V\times W$ must contain nonzero null covectors contradicting \eqref{eq:wfe}. Thus, $\operatorname{Ran} \xi_V$ is infinite-dimensional and for any $n \in \N$, $j=1,..,n$, we find functions $f_j \in C_K^\infty(\M)$ such that the functions $\{ \left. Ef_j \right|_{V} \}_{j=1}^n$ are linearly independent.
\end{proof}

\section{A lengthy integral calculation in QFT: proof of Proposition \ref{prop:lengthy_calculation}}
\label{app:lengthy_calculation}

To prove the Proposition, we will rely on the following result:
\begin{prop}
\label{prop:integral_tocha}
Let $\varphi:\R\times \R^{2n}\to\H_\Phi$ be a continuous, bounded function, let $T_\lambda:=\Psi(g'_\lambda)$, for some $g'_\lambda\in C_0^\infty(\M)$ such that $\lim_{\lambda\to 0}g'_\lambda$ exists in $C_0^\infty(\M)$ and $[T_\lambda,R_k]=0,\forall k$. Let $\vec{\alpha}_1,\vec{\alpha}_2,\vec{\beta}_1, \vec{\beta}_2\in \R^{2n}$. The expression
\begin{align}
\lim_{\lambda\to 0}&\frac{1}{(2\pi)^n}\int d\vec{\xi}e^{-\frac{\vec{\xi}^2}{4}}\omega'\left(e^{-i\frac{\cos(\lambda)}{\lambda}\vec{\alpha}_1\cdot \vec{R}}e^{i\frac{\cos(\lambda)}{\lambda}\vec{\beta}_1\cdot \vec{R}}e^{-i\cos(\lambda)\vec{\xi}\cdot\vec{R}}e^{-i\sin(\lambda)\vec{g}\cdot\vec{R}}\right.\nonumber\\
&\left.e^{-i\frac{\cos(\lambda)}{\lambda}\vec{\beta}_2\cdot \vec{R}}e^{i\frac{\cos(\lambda)}{\lambda}\vec{\alpha}_2\cdot \vec{R}}e^{-i\lambda T_\lambda}\right)\varphi(\lambda, \sin(\lambda)\vec{\xi})
\label{nasty_expr}
\end{align}
equals $\varphi(0,0)$, if
\begin{equation}
\vec{\alpha}_1=\vec{\beta}_1, \vec{\alpha}_2=\vec{\beta}_2
\end{equation}
or $0$, otherwise.
\end{prop}
\begin{proof}
By eq. (\ref{cova_prime}), we have that:
\begin{align}
&\frac{1}{(2\pi)^n}\int d\vec{\xi}e^{-\frac{\vec{\xi}^2}{4}}\omega'\left(e^{-i\frac{\cos(\lambda)}{\lambda}\vec{\alpha}_1\cdot \vec{R}}e^{i\frac{\cos(\lambda)}{\lambda}\vec{\beta}_1\cdot \vec{R}}e^{-i\cos(\lambda)\vec{\xi}\cdot\vec{R}}e^{-i\sin(\lambda)\vec{g}\cdot\vec{R}}\right.\nonumber\\
&\left.e^{-i\frac{\cos(\lambda)}{\lambda}\vec{\beta}_2\cdot \vec{R}}e^{i\frac{\cos(\lambda)}{\lambda}\vec{\alpha}_2\cdot \vec{R}}e^{-i\lambda T_\lambda}\right)\varphi(\lambda, \sin(\lambda)\vec{\xi})\nonumber\\
&=\frac{1}{(2\pi)^n}\int d\vec{\xi}e^{-\frac{\vec{\xi}^2}{4}}e^{\frac{i}{2}\frac{\cos(\lambda)^2}{\lambda^2}(\vec{\alpha}_1\cdot \sigma\vec{\beta}_1+\vec{\beta}_2\cdot \sigma\vec{\alpha}_2+(\vec{\alpha}_1-\vec{\beta}_1)\cdot \sigma (\vec{\alpha}_2-\vec{\beta}_2)}e^{-\frac{i}{2}\frac{\cos(\lambda)\sin(\lambda)}{\lambda}(\vec{\alpha}_1-\vec{\beta}_1+\vec{\alpha}_2-\vec{\beta}_2)\cdot \sigma \vec{g}}\nonumber\\
&e^{\frac{i}{2}\frac{\cos(\lambda)^2}{\lambda}(\vec{\alpha}_1-\vec{\beta}_1+\vec{\alpha}_2-\vec{\beta}_2)\cdot \sigma \vec{\xi}}e^{\frac{i}{2}\sin(\lambda)\cos(\lambda)\vec{\xi}\cdot \sigma \vec{g}}\omega'\left(e^{-i\left[\frac{\cos(\lambda}{\lambda}(\vec{\alpha}_1-\vec{\beta}_1-\vec{\alpha}_2+\vec{\beta}_2)-\cos(\lambda)\vec{\xi}+\sin(\lambda)\vec{g}\right]\cdot \vec{R}}e^{-i\lambda T_\lambda}\right)\nonumber\\
&\varphi(\lambda, \sin(\lambda)\vec{\xi})\nonumber\\
&=\frac{e^{-\lambda^2 c(\lambda)}}{(2\pi)^n}\int d\vec{\xi}e^{-\frac{\vec{\xi}^2}{4}}e^{\frac{i}{2}\frac{\cos(\lambda)^2}{\lambda^2}(\vec{\alpha}_1\cdot \sigma\vec{\beta}_1+\vec{\beta}_2\cdot \sigma\vec{\alpha}_2+(\vec{\alpha}_1-\vec{\beta}_1)\cdot \sigma (\vec{\alpha}_2-\vec{\beta}_2)}e^{-\frac{i}{2}\frac{\cos(\lambda)\sin(\lambda)}{\lambda}(\vec{\alpha}_1-\vec{\beta}_1+\vec{\alpha}_2-\vec{\beta}_2)\cdot \sigma \vec{g}}\nonumber\\
&e^{\frac{i}{2}\frac{\cos(\lambda)^2}{\lambda}(\vec{\alpha}_1-\vec{\beta}_1+\vec{\alpha}_2-\vec{\beta}_2)\cdot \sigma \vec{\xi}}e^{\frac{i}{2}\sin(\lambda)\cos(\lambda)\vec{\xi}\cdot \sigma \vec{g}}e^{-\frac{1}{4}\left[\frac{\cos(\lambda}{\lambda}(\vec{\alpha}_1-\vec{\beta}_1-\vec{\alpha}_2+\vec{\beta}_2)-\cos(\lambda)\vec{\xi}+\sin(\lambda)\vec{g}\right]^2}\nonumber\\
&\varphi(\lambda, \sin(\lambda)\vec{\xi}),
\label{comp_expr}
\end{align}
for some bounded $c(\lambda)>0$.

For fixed $\vec{\xi}$, the last exponential tends to zero in the limit $\lambda\to 0$, unless
\begin{equation}
\vec{\alpha}_1-\vec{\beta}_1-(\vec{\alpha}_2-\vec{\beta}_2)=0.
\label{first_vector_cond}
\end{equation}
This, together with the integration measure $d\xi e^{-\frac{\vec{\xi}^2}{4}}$, implies that the left-hand side of eq. (\ref{comp_expr}) tends to zero as $\lambda\to 0$ if the vector condition (\ref{first_vector_cond}) is violated.

Hence, we next assume that eq. (\ref{first_vector_cond}) holds, in which case eq. (\ref{comp_expr}) simplifies to:
\begin{align}
&\frac{e^{-\lambda^2 c(\lambda)}}{(2\pi)^n}\int d\vec{\xi}e^{-\frac{\vec{\xi}^2}{4}}e^{i\frac{\cos(\lambda)^2}{2\lambda^2}(\vec{\alpha}_1\cdot \sigma\vec{\beta}_1+\vec{\beta}_2\cdot \sigma\vec{\alpha}_2)}e^{-\frac{i}{2}\frac{\cos(\lambda)\sin(\lambda)}{\lambda}(\vec{\alpha}_1-\vec{\beta}_1+\vec{\alpha}_2-\vec{\beta}_2)\cdot \sigma \vec{g}}\nonumber\\
&e^{i\frac{\cos(\lambda)^2}{2\lambda}(\vec{\alpha}_1-\vec{\beta}_1+\vec{\alpha}_2-\vec{\beta}_2)\cdot \sigma \vec{\xi}}e^{\frac{i}{2}\sin(\lambda)\cos(\lambda)\vec{\xi}\cdot \sigma \vec{g}}e^{-\frac{(\sin(\lambda)\vec{g}-\vec{\xi})^2}{4}}\varphi(\lambda, \sin(\lambda)\vec{\xi}).
\end{align}

Taking its vector norm, eliminating the term $\frac{e^{-\lambda^2 c(\lambda)}}{(2\pi)^n}$ and rearranging, we arrive at:
\begin{align}
&\gamma_\lambda:=\left\|\int d\vec{\xi}e^{-\frac{\vec{\xi}^2}{4}}e^{i\frac{\cos(\lambda)^2}{2\lambda}(\vec{\alpha}_1-\vec{\beta}_1+\vec{\alpha}_2-\vec{\beta}_2)\sigma\vec{\xi}}e^{-\frac{(\sin(\lambda)\vec{g}-\cos(\lambda)\vec{\xi})^2}{4}}\tilde{\varphi}(\lambda,\sin(\lambda)\vec{\xi})\right\|,
\end{align}
with
\begin{equation}
\tilde{\varphi}(\lambda,\sin(\lambda)\vec{\xi}):=e^{\frac{i}{2}\cos(\lambda)\sin(\lambda)\vec{\xi}\cdot \sigma \vec{g}}\varphi(\lambda,\sin(\lambda)\vec{\xi}).
\end{equation}
Note that $\tilde{\varphi}$ is also continuous and bounded.

Making the change of variables 
\begin{equation}
\vec{\xi}\to \vec{\xi}':=\sqrt{1+\cos(\lambda)^2}\vec{\xi}-\frac{\sin(\lambda)\cos(\lambda)}{\sqrt{1+\cos(\lambda)^2}}\vec{g},    
\end{equation}
we obtain, after eliminating irrelevant complex phases, that
\begin{align}
&\gamma_\lambda=s(\lambda)\left\|\int d\vec{\xi}e^{-\frac{\vec{\xi}^2}{4}}e^{i\frac{\cos(\lambda)^2}{2\lambda \sqrt{1+\cos(\lambda)^2}}(\vec{\alpha}_1-\vec{\beta}_1+\vec{\alpha}_2-\vec{\beta}_2)\sigma\vec{\xi}}\tilde{\varphi}\left(\lambda,\sin(\lambda)\vec{\xi}_\lambda\right)\right\|,
\end{align}
with 
\begin{align}
s(\lambda):=\frac{e^{-\frac{\sin(\lambda)^2\vec{g}^2}{4(1+\cos(\lambda)^2)}}}{(1+\cos(\lambda)^2)^n},
\vec{\xi}_\lambda:=\frac{1}{\sqrt{1+\cos(\lambda)^2}}\vec{\xi}+\frac{\sin(\lambda)\cos(\lambda)}{1+\cos(\lambda)^2}\vec{g}.
\end{align}
Note that 
\begin{align}
&\lim_{\lambda\to 0}s(\lambda)=\frac{1}{2^n},\lim_{\lambda\to 0}\vec{\xi}_\lambda=\frac{\vec{\xi}}{\sqrt{2}}.
\end{align}

Define $\tilde{\gamma}_\lambda:=\frac{\gamma_\lambda}{s(\lambda)}$. Thanks to the continuity and boundedness of $\tilde{\varphi}$, the value of $\tilde{\gamma}_\lambda$ can be upper bounded by dividing the range of integration over $\vec{\xi}$. Let $\|\tilde{\varphi}(\bullet)\|\leq C$ and, for any $\delta>0$ sufficiently small, let $r>0$ be such that
\begin{equation}
\int_{\|\xi\|>r} d\vec{\xi}e^{-\frac{\vec{\xi}^2}{4}}=\delta.
\end{equation}
Then we have that
\begin{align}
\tilde{\gamma}_\lambda&\leq \left\|\int_{\|\xi\|\leq r} d\vec{\xi}e^{-\frac{\vec{\xi}^2}{4}}e^{i\frac{\cos(\lambda)^2}{2\lambda \sqrt{1+\cos(\lambda)^2}}(\vec{\alpha}_1-\vec{\beta}_1+\vec{\alpha}_2-\vec{\beta}_2)\sigma\vec{\xi}}\tilde{\varphi}\left(\lambda,\sin(\lambda)\vec{\xi}_\lambda\right)\right\|+ C\delta\nonumber\\
&\leq \left\|\int_{\|\xi\|\leq r} d\vec{\xi}e^{-\frac{\vec{\xi}^2}{4}}e^{i\frac{\cos(\lambda)^2}{2\lambda \sqrt{1+\cos(\lambda)^2}}(\vec{\alpha}_1-\vec{\beta}_1+\vec{\alpha}_2-\vec{\beta}_2)\sigma\vec{\xi}}\tilde{\varphi}(0,0)\right\|\nonumber\\
&+\int_{\|\xi\|\leq r} d\vec{\xi}e^{-\frac{\vec{\xi}^2}{4}}\left\|\tilde{\varphi}(0,0)-\tilde{\varphi}\left(\lambda,\sin(\lambda)\vec{\xi}_\lambda\right)\right\|+C\delta.
\end{align}

In the limit $\lambda\to 0$, the second term of the last upper bound tends to zero. Moreover, if the vector 
\begin{equation}
\vec{\alpha}_1-\vec{\beta}_1+\vec{\alpha}_2-\vec{\beta}_2    
\label{second_vector}
\end{equation}
is not zero, so does the first term too. In that case, $\lim_{\lambda\to 0}\tilde{\gamma}_\lambda\leq \delta C$. Since $\delta$ was arbitrary, it follows that $\lim_{\lambda\to 0}\tilde{\gamma}_\lambda=0$.

In other words: in the limit $\lambda\to 0$, eq. (\ref{comp_expr}) vanishes, unless eq. (\ref{first_vector_cond}) holds and the vector (\ref{second_vector}) vanishes. This would imply that $\vec{\alpha}_1=\vec{\beta}_1$, $\vec{\alpha}_2=\vec{\beta}_2$, in which case eq. (\ref{nasty_expr}) reduces to:
\begin{align}
&s(\lambda)\frac{e^{-c(\lambda)\lambda^2}}{(2\pi)^n}\int d\vec{\xi}e^{-\frac{\vec{\xi}^2}{4}}\tilde{\varphi}(\lambda,\sin(\lambda)\vec{\xi}_\lambda).
\end{align}
Dividing the integral into integration regions like above, one finds that, in the limit $\lambda\to 0$, this expression tends to:
\begin{align}
&\frac{1}{(4\pi)^n}\int d\vec{\xi}e^{-\frac{\vec{\xi}^2}{4}}\tilde{\varphi}(0,0)=\varphi(0,0).
\end{align}
    
\end{proof}

\begin{proof}[Proof of Proposition \ref{prop:lengthy_calculation}]
Let $g=\vec{g}\cdot\vec{h}+\tilde{g}$, with $E(\tilde{g},h_j)=0$, for all $j$, and let $\ket{\phi}\in\H_\Phi$. Then,
\begin{align}
&(\omega_\lambda'\otimes \id_\Phi)\circ\Theta_\lambda\left(e^{i\frac{\Psi(f_j)}{\lambda}}\Pi e^{-i\frac{\Psi(f_k)}{\lambda}} e^{i\Phi(g)}\right)\ket{\phi}\nonumber\\
&=\frac{1}{(2\pi)^n}\int d\vec{\xi} e^{-\frac{\vec{\xi}^2}{4}}(\omega_\lambda'\otimes \id_\Phi)\circ\Theta_\lambda\left(e^{i\frac{\Psi(f_j)}{\lambda}} e^{-i\vec{\xi}\cdot \Psi(\vec{h})}e^{i\Phi(g)}e^{-i\frac{\Psi(f_k)}{\lambda}} \right)\ket{\phi}\nonumber\\
&=\frac{1}{(2\pi)^n}\int d\vec{\xi} e^{-\frac{\vec{\xi}^2}{4}}\omega'_\lambda\left(e^{i\frac{\cos(\lambda)}{\lambda}\Psi(f_j)}e^{-i\cos(\lambda)\vec{\xi}\cdot \vec{R}}e^{-i\sin(\lambda)\vec{g}\cdot\vec{R}}e^{-i\frac{\cos(\lambda)}{\lambda}\Psi(f_k)}e^{-i\lambda T_\lambda}\right)\cdot\nonumber\\
&\left(e^{i\frac{\sin(\lambda)}{\lambda}\Phi(f_j)} e^{-i\sin(\lambda)\vec{\xi}\cdot \Phi(\vec{h})}e^{i\Phi(g_\lambda)}e^{-i\frac{\sin(\lambda)}{\lambda}\Phi(f_k)}\right)\ket{\phi},
\label{long_eq}
\end{align}
where $T_\lambda:=\Psi(g'_\lambda)$ is defined through the relation:
\begin{align}
\Theta_\lambda(\Phi(\tilde{g}))=\Phi(\tilde{g}_\lambda)+\lambda\Psi(g'_\lambda),
\end{align}
and 
\begin{equation}
g_\lambda:=\tilde{g}_\lambda+\cos(\lambda)\vec{g} \cdot \vec{h}.
\end{equation}
Note that, for $0<|\lambda|< \frac{\pi}{4}$, $E(g'_\lambda,h_j)=0$ for all $j$. This follows from the fact that 
\begin{equation}
\Theta_\lambda\left(-\sin(\lambda)\Phi(h_j)+\cos(\lambda)\Psi(h_j)\right)=\cos(2\lambda)\Psi(h_j).
\end{equation}
Since $E(\tilde{g},h_j)=0$ and $\Theta_\lambda$ is an isomorphism, we have that
\begin{align}
0&=[-\sin(\lambda)\Phi(h_j)+\cos(\lambda)\Psi(h_j),\Phi(\tilde{g})]\nonumber\\
&=[\cos(2\lambda)\Psi(h_j),\Phi(\tilde{g}_\lambda)+\lambda\Psi(g'_\lambda)]\nonumber\\
&=\lambda\cos(2\lambda)[\Psi(h_j),\Psi(g'_\lambda)].
\end{align}

By eq. (\ref{eq:weylcontinuity}), we have that the function
\begin{equation}
\varphi(\lambda,\vec{\eta}):=\left(e^{i\frac{\sin(\lambda)}{\lambda}\Phi(f_j)} e^{-i\vec{\eta}\cdot \Phi(\vec{h})}e^{i\Phi(g_\lambda)}e^{-i\frac{\sin(\lambda)}{\lambda}\Phi(f_k)}\right)\ket{\phi}
\end{equation}
is continuous. Moreover, $\|\varphi(\lambda,\vec{\eta})\|\leq 1$. 

Now, let $\Psi(f_k)=Q_k=\vec{f}_k\cdot\vec{R}$. With this notation, and the definition of $\omega'_{\lambda}$, eq. (\ref{long_eq}) can be rewritten as:
\begin{align}
&(\omega_\lambda'\otimes \id_\Phi)\circ\Theta_\lambda\left(e^{i\frac{\Psi(f_j)}{\lambda}}\Pi e^{-i\frac{\Psi(f_k)}{\lambda}}e^{i\Phi(g)}\right)\ket{\phi}\nonumber\\
&=\frac{1}{(2\pi)^n}\int d\vec{\xi} e^{-\frac{\vec{\xi}^2}{4}}(\omega_\lambda'\otimes \id_\Phi)\circ\Theta_\lambda\left(e^{i\frac{\Psi(f_j)}{\lambda}} e^{-i\vec{\xi}\cdot \Psi(\vec{h})}e^{-i\frac{\Psi(f_k)}{\lambda}} e^{i\Phi(g)}\right)\ket{\phi}\nonumber\\
&=\sum_{l,m}\sqrt{p_l}\sqrt{p_m}\frac{1}{(2\pi)^n}\int d\vec{\xi}e^{-\frac{\vec{\xi}^2}{4}}\omega'\left(e^{-i\frac{\cos(\lambda)}{\lambda}\vec{f}_l\cdot \vec{R}}e^{i\frac{\cos(\lambda)}{\lambda}\vec{f}_j\cdot \vec{R}}e^{-i\vec{\xi}\cdot\vec{R}}e^{-i\sin(\lambda)\vec{g}\cdot\vec{R}}\right.\nonumber\\
&\left.e^{-i\frac{\cos(\lambda)}{\lambda}\vec{f}_k\cdot \vec{R}}e^{i\frac{\cos(\lambda)}{\lambda}\vec{f}_m\cdot \vec{R}}e^{-i\lambda T_\lambda}\right)\varphi(\lambda,\sin(\lambda)\vec{\xi}).
\end{align}
Now we invoke Proposition \ref{prop:integral_tocha}: in the limit $\lambda\to 0$, each of the summands above respectively tends to either $\varphi(0,0)$ or to $0$, depending on whether $\vec{f}_l=\vec{f}_j,\vec{f}_m=\vec{f}_k$ or not. The result is thus
\begin{align}
&\sqrt{p_j}\sqrt{p_k}\varphi(0,0)=\sqrt{p_j}\sqrt{p_k}\left(e^{i\Phi(f_j)} e^{i\Phi(g)}e^{-i\Phi(f_k)}\right)\ket{\phi}.
\end{align}

Since $\ket{\phi}$ was an arbitrary vector, it follows that the limit holds strongly.
    
\end{proof}

\section{The massless scalar field: proof of Proposition \ref{prop:Bell_massless}}
\label{app:massless}
Before we proceed, we introduce some basic notation and a useful fact regarding the massless scalar field.

In the following, we will be dealing with Minkowski spacetime in $1+1$ dimensions. Traditionally a point $y\in \M$ of this manifold is denoted as $y=(t,x)$,  with the first (second) coordinate indicating time (space). The standard way to express the point $y\in\M$ is thus: $y=t\hat{t}+x\hat{x}$, with $\hat{t}:=(1,0)$, $\hat{x}:=(0,1)$. Sometimes it will be convenient, though, to write $y$ in the coordinates $y=u\hat{u}+v\hat{v}$, with 
\begin{equation}
\hat{u}=\frac{1}{\sqrt{2}}(-1,1), \hat{v}=\frac{1}{\sqrt{2}}(1,1), u=\frac{x-t}{\sqrt{2}},v=\frac{x+t}{\sqrt{2}}.
\end{equation}
In Minkowski spacetime, the quintessential causally convex set is the causal diamond. A causal diamond of size $\Delta>0$ centered in $y\in\M$ is the subset of $\M$ given by
\begin{equation}
\{y+(t,x):\frac{|x\pm t|}{\sqrt{2}}\leq \Delta\}.
\end{equation}
In $(u,v)$ coordinates:
\begin{equation}
\{y+u\hat{u}+v\hat{v}:|u|,|v|\leq \Delta\}.
\end{equation}

We conclude this short review by reminding the reader that, for a massless free scalar field in Minkowski spacetime in $1+1$ dimensions, the commutator operator $E$ has the very simple form \cite{Jubb_2022}:
\begin{align}
Ef(y)&=\frac{1}{2}\int_{J^+(y)}f(z)\mathrm{d}z-\frac{1}{2}\int_{J^-(y)}f(z)\mathrm{d}z \nonumber \\
&= \frac{1}{2} \int_{u\leq 0,v\geq 0} (f((u^y+u)\hat{u}+(v^y+v)\hat{v})-f((u^y+u)\hat{u}+(v^y+v)\hat{v})) \mathrm{d} u \mathrm{d} v.
\label{propagator_massless}
\end{align}

\begin{proof}[Proof of Proposition \ref{prop:Bell_massless}]
For $X=A,B$, consider compact causal diamonds $\S^X,\T^X\subset\M$ of size $\Delta$, respectively centered in $y^X_S=(t_S^X, x_S^X),y^X_T=(t_T^X, x_T^X)\in\M$, with
\begin{align}
&x^A_{S}=x^A_{T} \not=x^B_{S}=x^B_{T},t^A_{S}=t^B_{S} \not=t^A_{T}=t^B_{T},
\end{align}
and satisfying conditions (\ref{conds_box_channel}), (\ref{space-like_conds}), see \Cref{fig:massless_zones}.
\begin{figure}
\centering
\includegraphics[width=0.9\textwidth]{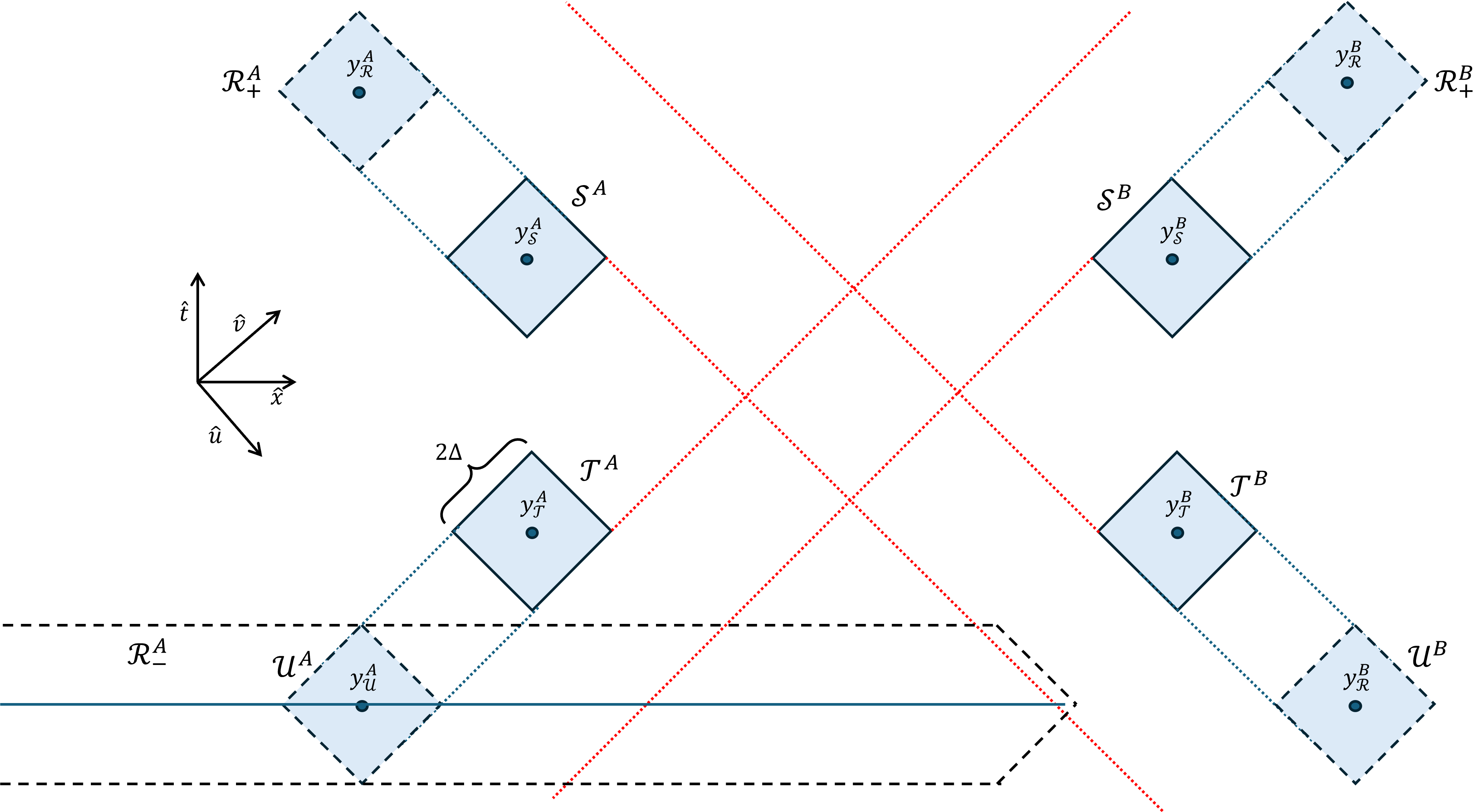}
\caption{Interaction zones and Bell regions for a box channel for massless scalar fields. Note that $\RR_-^A$ contains $\Sigma$, a Cauchy surface for $\overline{\RR_+^A}$.}
\label{fig:massless_zones}
\end{figure}
Let $\chi^\epsilon$ be a non-negative bump function with $|\chi^\epsilon(\bullet)|\leq 1$ and
\begin{align}
\chi^\epsilon(z)=&1,\mbox{ for }  |z|\leq \Delta-\epsilon,\nonumber\\
&0,\mbox{ for }  |z|\geq \Delta-\frac{\epsilon}{2}.
\end{align}
We define the channel's test functions (in $(u,v)$ coordinates) as
\begin{align}
&f_A(u,v):=\chi^\epsilon(u-u^A_S)\cos\left(\frac{2\pi(v-v^A_S)}{\Delta}\right)\chi^\epsilon(v-v^A_S),\nonumber\\
&g_A(u,v):=-\chi^\epsilon(u-u^A_T)\cos\left(\frac{2\pi(u-u^A_T)}{\Delta}\right)\chi^\epsilon(v-v^A_T),\nonumber\\
&f_B(u,v):=\chi^\epsilon(u-u^B_S)\cos\left(\frac{2\pi(u-u^B_S)}{\Delta}\right)\chi^\epsilon(v-v^B_S),\nonumber\\
&g_B(u,v):=-\chi^\epsilon(u-u^B_T)\cos\left(\frac{2\pi(v-v^B_T)}{\Delta}\right)\chi^\epsilon(v-v^B_T),
\end{align}
for $\epsilon>0$ small enough (soon it will be clear how small).

Now, given $\OO$, we define the regions
\begin{align}
&\RR_+^A=\Int{\S^A}-\lambda\hat{u},\U^A=\Int{\T^A}-\lambda\hat{v},\nonumber\\
&\RR_+^B=\Int{\S^B}+\lambda\hat{v},\U^B=\Int{\T^B}+\lambda\hat{u},
\end{align}
for $\lambda>0$ large enough so that $\RR_+^A, \RR_+^B\subset \OO^+, \U^A,\U^B\subset \OO^-$. The choice of vectors in the equation above, together with relations (\ref{conds_box_channel}), (\ref{space-like_conds}), implies that
\begin{equation}
\overline{{\cal F}^A}\perp \overline{{\cal G}^B},\mbox{ for }{\cal F}, {\cal G}\in\{\RR_{+},\S,\T,\U\}.
\label{all_commute}
\end{equation}
Moreover, the union of displacements of $\U^X$ along a part of the $\hat{x}$ axis allows one to define regions $\RR_-^X$ such that 
\begin{equation}
\overline{\RR_-^A}\perp \overline{\U^B},\overline{\RR_-^B}\perp \overline{\U^A},\overline{\RR^X_+}\subset D^+(\RR_-^X),
\end{equation}
see \Cref{fig:massless_zones}. With such a choice of Bell regions, conditions (\ref{conds_Bell_QFT}) are thus satisfied.

Let us now show that condition (\ref{sympl_cond}) also holds. For $y =u\hat{u}+v\hat{v} \in \U^A$, equation (\ref{propagator_massless}) implies that\begin{equation}
E g_A(y)=\frac{1}{2}\int_{-\Delta}^{u-u_U^A} \cos\left(\frac{2\pi\tilde{u}}{\Delta}\right)\chi^\epsilon(\tilde{u}) \mathrm{d}\tilde{u} \int_{-\Delta}^\Delta \chi^\epsilon(\tilde{v}) \mathrm{d}\tilde{v},
\end{equation}
where $u_U^A$ denotes the $u$-variable for the center of $\U^A$; see \Cref{fig:integral} for a graphical proof.
\begin{figure}
\centering
\includegraphics[width=0.5\textwidth]{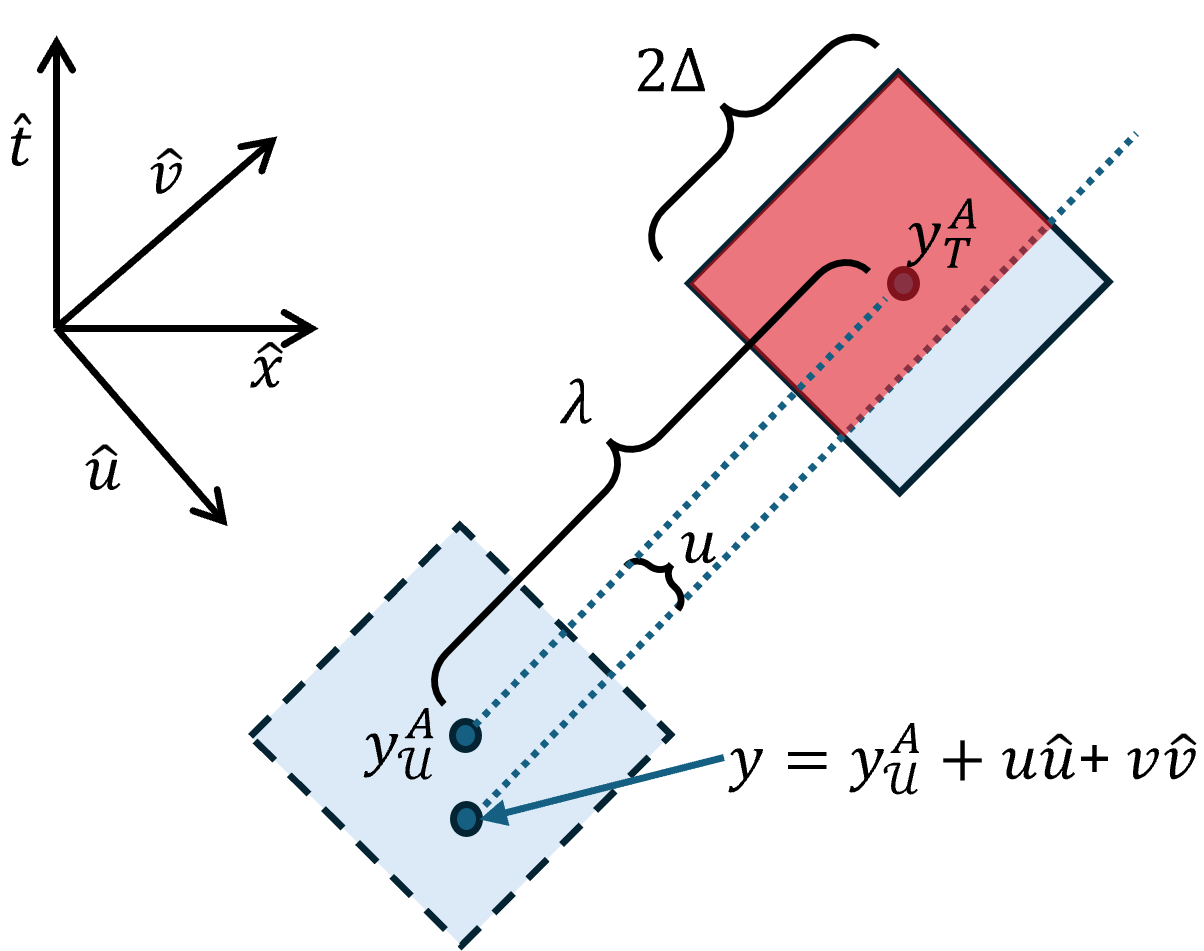}
\caption{Constructing test functions in the massless scalar field. The integrated area is shown in in red.}
\label{fig:integral}
\end{figure}
Thus, for $\epsilon\ll 1$, it holds that
\begin{equation}
\left. Eg_A\right|_{\U^A}(u,v)\approx \frac{\Delta^2}{2\pi}\sin\left(\frac{2\pi (u-u^A_U)}{\Delta}\right).
\label{proj_1}
\end{equation}
Similarly, 
\begin{equation}
\left. Ef_A\right|_{\RR_+^A}(u,v)\approx \frac{\Delta^2}{2\pi}\sin\left(\frac{2\pi (v-v^A_S)}{\Delta}\right),
\label{proj_2}
\end{equation}
Hence, we can choose 
\begin{equation}
\hat{g}_A(u,v)\propto \sin\left(\frac{2\pi (u-u^A_U)}{\Delta}\right)\chi^\epsilon(u-u^A_U)\chi^\epsilon(v-v^A_U).
\end{equation}
to ensure that
\begin{equation}
\langle g_A, E\hat{g}_A\rangle=1.
\end{equation}
By $\S^A\subset (\T^A)^\vee$, any point in $\RR_{+}^A$ is in the causal future of the whole of $\T^X,\U^X$. Eq. (\ref{propagator_massless}) therefore implies that 
\begin{equation}
\left.Eg_A\right|_{\RR^A_+} = \text{const.}, \left.E\hat{g}_A\right|_{\RR^A_+} = \text{const}.
\end{equation}
For $\epsilon$ sufficiently small, it holds that
\begin{equation}
\left.Ef_A\right|_{\RR_+^A}
\end{equation}
is linearly independent from the constant function on $\RR_+^A$. Thus, one can find $\hat{f}_A\in C_0^\infty(\RR_+^A)$ such that
\begin{equation}
 \langle \hat{f}_A,Ef_A\rangle =1,\langle g_A,E\hat{f}_A\rangle =\langle \hat{g}_A,E\hat{f}_A\rangle=0.
\end{equation}

One can do similar operations to construct $f_B,\hat{f}_B,g_B,\hat{g}_B$ with the analogous symplectic relations. 
\begin{equation}
\langle \hat{f}_B,Ef_B\rangle =1, \langle g_B,E\hat{g}_B\rangle =1,\langle g_B,E\hat{f}_B\rangle=\langle \hat{g}_B,E\hat{f}_B\rangle =0.
\end{equation}
Ultimately, $\epsilon$ must be chosen small enough to guarantee that the left-hand sides of eqs. (\ref{proj_1}), (\ref{proj_2}) and their $B$-analogs are linearly independent from the constant function. This property is independent of the parameter $\lambda$ used to quantify how far from the origin the Bell regions $\RR^X_\pm,\U^X$ are.

Finally, due to eq. (\ref{all_commute}) we furthermore have that symplectic terms mixing $A$ and $B$ functions vanish. Condition (\ref{sympl_cond}) is hence proven.
\end{proof}

\section{Generic QFTs: proof of Proposition \ref{prop:Bell_massive}}
\label{app:massive}
The Proposition follows from the next two lemmas.

\begin{lemma}
Let conditions (\ref{conds_unbounded_O}) hold. Then, for any $\OO\supset \bigcup_{X=A,B}\S^X\cup\T^X$, there exist regions $\RR_+^X\subset \J^X\cap \OO^+, \U^X\subset \K^X\cap \OO^-$, $\RR_-^X\subset\OO^-$ that form a Bell setup around $\OO$ [eq. (\ref{conds_Bell_QFT})]. In addition:
\begin{equation}
\overline{\RR_{+}^A}\perp \overline{\U^B}, \overline{\RR_{+}^B}\perp \overline{\U^A}.
\label{perpRUs}
\end{equation}

\end{lemma}
\begin{proof}
Choose arbitrary regions $\RR_+^X\subset \J^X\cap \OO^+$. Then, by $\overline{\J^A}\perp \overline{\J^B}$ it holds that $\overline{\RR^A_{+}}\perp \overline{\RR^B_{+}}$. Next, consider any Cauchy surface $\Sigma\subset \OO^-$. We start by defining the surfaces $\Sigma^X_{-}:=\Sigma\cap J^-(\overline{\RR_{+}^X})$, $\Sigma^X_{U}:=\Sigma\cap \overline{\K^X}$. From $\J^X\subset J^+(\S^X)$, $\K^X\subset J^-(\T^X)$ and $\S^X\subset (\T^X)^\vee$ it follows that $\J^X\subset (\K^X)^\vee$. Consequently, it holds that $\Sigma^X_{U}\subset J^-(\RR^X_{+})$, and so $\Sigma^X_U\subset \Sigma_{-}^X$.

Due to the conditions $\overline{\J^A}\perp\overline{\K^B}$, $\RR^A_{+}\subset\J^A$, we have that $\Sigma^A_{-}\perp\Sigma^B_U$. Similarly, $\Sigma^B_{-}\perp\Sigma^A_U$. Taking the union of sufficiently small neighborhoods of points in $\Sigma^X_{-}, \Sigma^X_{U}$ and then their causal hull, we construct regions $\RR^X_{-}\supset\Sigma^X_{-}\cap \OO^-$, $\U^X\subset\RR^X_{-}$ such that $\overline{\RR^A_{-}}\perp \overline{\U^B}$, $\overline{\RR^B_{-}}\perp \overline{\U^A}$. Moreover, it holds that $\overline{\RR_{+}^X}\subset D^+(\RR_{-})$, since $\Sigma^X_{-}$ is a Cauchy surface for $\overline{\RR_{+}^X}$. 

Eq. (\ref{perpRUs}) follows from $\overline{\J^A}\perp \overline{\K^B}$, $\overline{\J^B}\perp \overline{\K^A}$ and the inclusion relations $\RR_+^X\subset \J^X$, $\U^X\subset \K^X$.
    
\end{proof}

\begin{lemma}
Consider a generic scalar QFT, and let $\S^X,\T^X$ satisfy (\ref{conds_box_channel}), (\ref{space-like_conds}). Let $f_X\in C^\infty_0(\S^X)),g_X\in C^\infty_0(\T^X))$ be generic test functions. Then, for any $\OO\subset \M$ and $\RR_{\pm}^X,\U^X\subset \M$ forming a Bell setup around $\OO$ and satisfying condition (\ref{perpRUs}), there exist test functions $\hat{f}_X,\hat{g}_X$ such that condition (\ref{sympl_cond}) holds.
\end{lemma}
\begin{proof}
For $X=A,B$, find three test functions $\hat{g}^1_X, \hat{g}^2_X,\hat{g}^3_X\in C_0^\infty(\T^X)$ such that
\begin{align}
&\langle \hat{g}^1_X, E g_X\rangle\not=0,\nonumber\\
&\left.E\hat{g}_X^1\right|_{\RR_+^X},\left.E\hat{g}_X^2\right|_{\RR_+^X},\left.E\hat{g}_X^3\right|_{\RR_+^X},\mbox{ linearly independent}.
\label{generic_g}
\end{align}
This is easy to do: one can, e.g., take $\hat{g}_X^1$ to be
\begin{equation}
\hat{g}^1_X(x)=\chi(x)Eg_X,
\end{equation}
where $\chi$ is a non-negative bump function that vanishes outside $\U^X$ and approximates the constant function $1$ in most of $\U^X$. That suffices to ensure that the first line of eq. (\ref{generic_g}) is satisfied. Taking three generic functions in $\U^X$, one can define two linear combinations thereof $\hat{g}_X^2, \hat{g}_X^3$ such that the second line of (\ref{generic_g}) is satisfied as well.

Next, find a linear combination $\hat{g}_X$ of $\hat{g}^1_X, \hat{g}^2_X,\hat{g}^3_X$ such that
\begin{align}
&\langle g_X, E \hat{g}_X\rangle=1,\nonumber\\
&\left.Ef_X\right|_{\U^X},\left.Eg_X\right|_{\U^X}, \left.E\hat{g}_X\right|_{\U^X},\mbox{ linearly independent}.
\label{generic_g2}
\end{align}
This is possible thanks to eq. (\ref{generic_g}) and the linear independence of $\left.Ef_X\right|_{\U^X},\left.Eg_X\right|_{\U^X}$ due to genericity.

Finally, choose $\hat{f}_X\in C_0^\infty(\U^X)$ such that
\begin{equation}
\langle \hat{f}_X, E f_X\rangle=1, \langle \hat{f}_X, E g_X\rangle=\langle\hat{f}_X, E \hat{g}_X\rangle=0.
\end{equation}
The existence of $\hat{f}_X$ follows from the second line of eq. (\ref{generic_g2}).

To finish the proof, it remains to realize that the test functions $\{\hat{f}_X,\hat{g}_X:X=A,B\}$ so constructed satisfy
\begin{equation}
\langle s_A, Et_B\rangle=0,
\end{equation}
for $s,t\in\{f, \hat{f}, g, \hat{g}\}$. Indeed, such relations follow from the support of these functions and conditions 
\begin{equation}
\overline{\S^A}\perp \overline{\T^B}, \overline{\S^B}\perp \overline{\T^A}, \overline{\RR_{+}^A}\perp\overline{\RR^B_+}, \overline{\U^A}\perp \overline{\U^B},\overline{\RR_{+}^A}\perp \overline{\U^B}, \overline{\RR_{+}^B}\perp \overline{\U^A}.
\end{equation}
\end{proof}

\section{Bell experiments around implementable channels: proof of Proposition \ref{prop:FV_wirings_box}}
\label{app:FV_wirings}
To simplify the proof, we consider the effect of the communication protocol generating $\Omega$, not just on the target theory, but also on the probe and ancillary systems. In this regard, our proposal to feedforward the outcomes of FV schemes implies that we can model each FV scheme as a channel with interaction zone $\K_1$ followed by an instrument with interaction zone $\K_2$ on the probe. To model FV schemes conditioned on feedforwarded information, we assume that, at every network node $j$, the corresponding family of FV-realizable instruments $\{\Omega^{j,a}\}_a$ is such that the interaction (processing) zone $\K_1^j$ ($\K_2^j$) is the same for all $a$'s. We will use, however, a different probe to model each of the instruments $\{\Omega^{j,a}\}_a$. This means that the actual interaction at node $j$ will be a morphism $\Theta^a$, connecting probe(s) $a$ with the target and the ancillary QFTs. The measurement at $\K_2^j$ will be carried out on all the probes, and the outcomes of all the probes will be directed to each node in the total future of $\K_2^j$. We model the measurement of the probes through a (non-causal) instrument $(\Gamma_c)_c$, of the form
\begin{equation}
\Gamma_c(\bullet)=\sqrt{M_c}\bullet\sqrt{M_c},
\end{equation}
where $(M_c)_c$ is the POVM to be conducted on the probe QFT. To the effect of characterizing the correlations induced through wirings of FV schemes, we will not consider FV schemes where $\K_1\perp\K_2$. In that case, one can avoid the measurement of the probes altogether by choosing a joint outcome $c$ thereof, replacing the initial state $\varphi$ of the probes by the state $\frac{1}{p_c}\varphi\circ \Gamma_c$, with $p_c=\omega\circ \Gamma_c(1)$, and feeding the outcome $c$ to every later instrument that depended on the outcome of the measurement. This brings about a distribution $P(a,b|\alpha,\beta,c)$, such that
\begin{equation}
\sum_cp_cP(a,b|\alpha,\beta,c)=P(a,b|\alpha,\beta,\Omega).
\end{equation}
Since $C_{\mathrm{qa}}$ is convex, it thus suffices to prove that each $P(a,b|\alpha,\beta,c)\in C_{\mathrm{qa}}$ to show that $P(a,b|\alpha,\beta,\Omega)\in C_{\mathrm{qa}}$.

Let $\OO_1 \prec \ldots \prec \OO_n$ be the interaction zones of all the probe measurements and scattering morphisms involved in the implementation of channel $\Omega$. Respectively applying Lemma \ref{lemma:causal_ordering} (\Cref{app:causal_networks}) to regions $\RR^A_\pm, \RR^B_\pm$, we have that there exist regions $(\RR^{A,j}_{\pm})_j$, $(\RR^{B, j}_{\pm})_j$ satisfying:
\begin{align}
&\RR^{X,j}_{\pm}\subset \OO_j^\pm,\nonumber\\
&\overline{\RR^{X,j}_{+}}\subset D^+(\RR^{X,j}_{-}),\nonumber\\
&\RR^{X,j}_{-}= \RR^{X,j-1}_{+},\nonumber\\
&\RR^{X,n}_{+}=\RR^X_{+},\RR^{X,1}_{-}=\RR^X_{-},
\label{regions_ordered}
\end{align}
for $j=1,...,n$, $X=A,B$. Note that the conditions include $\RR_+^X\subset D^+(\RR^X_-)$, for $X=A,B$. 

Following the proof of Theorem \ref{theorem:causal_network}, we introduce the sets $\wedge_j:=\{l\in \{1,...,n\}:\OO_j\subset \OO^\vee_l\}$. For $\mathbb{X}=\mathbb{A},\mathbb{B}$, we further define the sets of indices:
\begin{align}
\mathbb{X}_m=&\{\emptyset\},\nonumber\\
\mathbb{X}_{j-1}=&(\mathbb{X}_{j}\cup 
\wedge_j)\setminus \{j\},\mbox{ if }\OO_j\not\perp \RR^{X,j}_+,\noindent\\
&\mathbb{X}_{j}\setminus\{j\},\mbox{ otherwise}.
\end{align}
Note that $\mathbb{X}_1\in \{\emptyset, \{1\}\}$. Since the instrument associated to $\OO_1$ is a scattering morphism (and thus has only one outcome), it follows that $c(\mathbb{X}_1)$ does not take more than one value.

Then the overall action of the protocol on target, probes and ancillas is given by
\begin{equation}
\hat{\Omega}(\bullet)=\sum_{c} \hat{\Omega}^{1,c(\wedge_1)}_{c_1} \circ \ldots \circ \hat{\Omega}^{n,c(\wedge_n)}_{c_n}(\bullet),
\end{equation}
where each instrument $(\hat{\Omega}^{j,c(\wedge_j)}_{c_j})_{c_j}$ is either a scattering morphism or a measurement of one or several probes. Consequently, if we denote by $\zeta$ the overall initial state of target, probes and ancillas, the observed distribution will be:
\begin{equation}
P(a,b|\alpha,\beta)=\zeta\circ \Lambda^{A,\alpha}\circ \Lambda^{B,\beta}\circ\hat{\Omega}(M^A_a M^B_b),
\label{expr_box}
\end{equation}
where we trivially extend the action of the instruments $\Lambda^{A,\alpha}, \Lambda^{B,\beta}$ from the target QFT to the composite system formed by target, ancillary and probe QFTs.

To progress further, we need the following two lemmas:
\begin{lemma}
\label{lemma:Aes}
Let $j\in\mathbb{A}_j$. Then $J^-(\RR^{A,j}_+)\supset \OO_j$,  $j\not\in\mathbb{B}_j$ and $\OO_j\perp \RR^{B,j}_+$. The same statement holds with the replacements $A\leftrightarrow B$, $\mathbb{A}\leftrightarrow\mathbb{B}$.    
\end{lemma}
\begin{proof}
Suppose that $j\in \mathbb{A}_j$. Then, there exists $k>j$, $j\in \wedge_k$, with $\OO_k\not\perp \RR^{A,k}_+$, hence $J^-(\RR^{A,k}_+)\cap \OO_k\not=\emptyset$. This last relation implies, by $\OO_k\subset \OO_j^\vee$, that $J^-(\RR^{A,k}_+)\supset \OO_j$. Thus, $J^-(\RR^{A,j}_+)\supset \OO_j$, since $\RR^{A,k}_{+}\subset D^+(\RR^{A,j}_+)$. Also, $\RR^{B,j}_+\perp \OO_j$; otherwise there would exist $x\in\OO_j\cap J^-(\RR^{A,j}_{+})\cap J^-(\RR^{B,j}_{+})$, which would imply, by $\RR^{X,j}_{-}\subset D^+(\RR^X_{-})$, that $\RR^A_{-}\not\perp \RR^B_{-}$, hence contradicting eq. (\ref{cond_split}). Similarly, there cannot exist $k$ with $j\in \wedge_k$ such that $\RR^{B,k}_{+}\not\perp \OO_k$: otherwise, $\RR^{B,j}_+\not\perp \OO_j$ and again we would get a contradiction with $\RR^A_{-}\perp \RR^B_{-}$. Since, for all $k$ with $j\in \wedge_k$, it is the case that $\RR^{B,k}_{+}\perp \OO_k$, we have that $j\not\in \mathbb{B}_j$.     
\end{proof}

\begin{lemma}
\label{lemma:inductio}
Let $O^{A}\in \A(\RR^A_+)$, $O^{B}\in \A(\RR^B_+)$, where $\A$ denotes the algebra of the target QFT, and define
\begin{equation}
O_n=O^{A}O^{B},O_{j-1}:=\sum_{c_j}\hat{\Omega}^{j,c(\wedge_j)}_{c_j}(O_j),
\end{equation}
where $O_j$ is understood to depend (at most) on the indices $c_1,...,c_j$.
Then, it holds that
\begin{equation}
O_{j}=O^{A, c(\mathbb{A}_j)}_jO^{B, c(\mathbb{B}_j)}_j,
\end{equation}
where $O^{A, c(\mathbb{A}_j))}_j, O^{B, c(\mathbb{B}_j)}_j$ are respectively localized in $\RR^{A,j}_+$, $\RR^{B,j}_+$, act trivially on any probe system not associated to the instruments localized in $\{\OO_k:k>j\}$ and
\begin{align}
&O^{A,c(\mathbb{A}_{j})}_{j}=\theta_j^{A,c(\mathbb{A}_{j})}(O^A),\nonumber\\
&O^{B,c(\mathbb{B}_{j})}_{j}=\theta_j^{B,c(\mathbb{B}_{j})}(O^B),
\end{align}
for some completely positive and trace-preserving maps $\theta_j^{A,c(\mathbb{A}_{j})}$, $\theta_j^{B,c(\mathbb{B}_{j})}$.    
\end{lemma}
\begin{proof}
We prove the claim by induction. Suppose that the induction hypothesis holds from level $j$ upwards. To see how $\hat{\Omega}^{j, c(\wedge_j)}$ acts on $O_j$, we must consider the two possible types of map separately. Let $\Omega^{c(\wedge_j)}_j$ be a probe measurement. We are assuming that, in each FV scheme, the interaction zone for the probe measurement intersects the future of the interaction zone of the corresponding scattering morphism. This implies that the interaction zone $\OO_k$ for the corresponding scattering morphism satisfies $\OO_j\succ \OO_k$, and so $k<j$. Thus, by the assumptions of the lemma, $O^{A, c(\mathbb{A}_j)}_j,O^{B, c(\mathbb{B}_j)}_j$ act trivially on the probe(s) where the measurement takes place. Call $(M^j_{c_j})_{c_j}$ the POVM, localized in $\OO_j$, on the corresponding probes. Then it holds that 
\begin{equation}
\left[\sqrt{M^j_{c_j}},O^{A, c(\mathbb{A}_j))}_j\right]=\left[\sqrt{M^j_{c_j}},O^{B, c(\mathbb{B}_j)}_j\right]=0.
\end{equation}
Consequently,
\begin{align}
&\sum_{c_j}\hat{\Omega}^{j,c(\wedge_j)}_{j,c_j}(O^{A,c(\mathbb{A}_j)}_jO^{B,c(\mathbb{B}_j)}_j)\nonumber\\
&=\sum_{c_j}M^j_{c_j}O^{A,c(\mathbb{A}_j)}_jO^{B,c(\mathbb{B}_j)}_j.
\end{align}

By Lemma \ref{lemma:Aes}, either $j\in \mathbb{A}_j, j\not\in \mathbb{B}_j$ or $j\in \mathbb{B}_j, j\not\in \mathbb{A}_j$ or $j\not\in \mathbb{A}_j,\mathbb{B}_j$. In the first case, we can write
\begin{equation}
O^{A,c(\mathbb{A}_{j-1})}_{j-1}=\sum_{c_j}M^j_{c_j}O^{A,c(\mathbb{A}_j)}_j, O^{B,c(\mathbb{B}_{j-1})}_{j-1}=O^{B,c(\mathbb{B}_{j})}_{j}.
\end{equation}
In particular, one can verify that
\begin{align}
&O^{A,c(\mathbb{A}_{j-1})}_{j-1}=\theta_{j-1}^{A,c(\mathbb{A}_{j-1})}(O^A),\nonumber\\
&O^{B,c(\mathbb{B}_{j-1})}_{j-1}=\theta_{j-1}^{B,c(\mathbb{B}_{j-1})}(O^B),
\end{align}
for
\begin{align}
&\theta_{j-1}^{A,c(\mathbb{A}_{j-1})}(\bullet):=\sum_{c_j}\sqrt{M^j_{c_j}}\theta_{j}^{A,c(\mathbb{A}_{j})}(\bullet)\sqrt{M^j_{c_j}},\nonumber\\
&\theta_{j-1}^{B,c(\mathbb{B}_{j-1})}(\bullet):=\theta_{j}^{B,c(\mathbb{B}_{j})}(\bullet).
\end{align}

By $\overline{\RR_{+}^{X,j}}\subset D^+(\RR_{-}^{X,j})$, we can respectively localize $O^{A,c(\mathbb{A}_j)}_j$, $O^{B,c(\mathbb{B}_{j})}_{j}$ in $\RR_{-}^{A,j}=\RR_{+}^{A,j-1},\RR_{-}^{B,j}=\RR_{+}^{B,j-1}$. Furthermore, by Lemma \ref{lemma:Aes} it holds that $J^-(\RR^{A,j}_+)\supset \OO_j$, so we can also localize $M^j_{c_j}$ in $\RR^{A,j-1}_{+}$. We hence have that $O^{A,c(\mathbb{A}_{j-1})}_{j-1}, O^{B,c(\mathbb{B}_{j-1})}_{j-1}$ are respectively localized in $\RR_+^{A,j-1}$, $\RR_+^{B,j-1}$.

The case $j\not\in \mathbb{A}_j,j\in\mathbb{B}_j$ is treated similarly. If $j\not\in \mathbb{A}_j,\mathbb{B}_j$, then $O_j$ does not depend on $c_j$, so one can sum the outcomes of the POVM, obtaining $O_{j-1}=O_j$.

In either of the three cases, since $M^j_{c_j}$ only acts on the probes related to the instrument family at $\OO_j$, the condition that $O^{A,c(\mathbb{A}_{j-1})}_{j-1}$, $O^{B,c(\mathbb{B}_{j-1})}_{j-1}$ act trivially on the probes not disturbed by the instruments in $\{\OO_k\}_{k>j-1}$ is also satisfied.

Suppose now that $\hat{\Omega}^j$ is a morphism. Then, the index $c_j$ has only one value, so we can safely skip it and write:
\begin{equation}
O^{A,c(\mathbb{A}_{j-1})}_{j-1}=\Theta^{j,c(\wedge_j)}\left(O^{A,c(\mathbb{A}_j)}_j\right), O^{B,c(\mathbb{B}_{j-1})}_{j-1}=\Theta^{j,c(\wedge_j)}\left(O^{B,c(\mathbb{B}_{j})}_{j}\right).
\end{equation}
The index assignment is consistent: if $\OO_j\perp \RR^{A,j}_+$ ($\OO_j\perp \RR^{B,j}_+$), then the morphism will act like the identity on $O^{A,c(\mathbb{A}_j)}_j$ ($O^{A,c(\mathbb{A}_j)}_j$). In addition, due to the properties of the scattering morphism, $O^{A,c(\mathbb{A}_{j-1})}_{j-1}, O^{B,c(\mathbb{B}_{j-1})}_{j-1}$ are respectively localized in $\RR^{A,j}_-, \RR^{B,j}_-$.

Since the scattering morphism is a channel that only perturbs the target and the probe systems associated to $\OO_j$, both the probe and channel conditions in the statement of the lemma hold.
    
\end{proof}

By Lemma \ref{lemma:inductio}, the distribution (\ref{expr_box}) observed in the Bell experiment satisfies:
\begin{align}
&P(a,b|\alpha,\beta)=\zeta\circ \Lambda^{A,\alpha}\circ\Lambda^{B,\beta}\circ\hat{\Omega}(M^A_aM^B_b)\nonumber\\
&=\zeta\circ \Lambda^{A,\alpha}\circ\Lambda^{B,\beta}(\tilde{M}^A_a\tilde{M}^B_b),
\end{align}
for some POVMs $(\tilde{M}^A_a)_a, (\tilde{M}^B_b)_b$ respectively localized in $\RR^A_-$, $\RR^B_-$. By the spacetime conditions (\ref{conds_Bell_QFT}), we have that $\U^X\subset \RR_-^X$, which, together with condition (\ref{cond_split}), implies that the Kraus operators of $\Lambda^{A,\alpha}$ commute with $\tilde{M}^B_b$; and those of $\Lambda^{B,\beta}$, with $\tilde{M}^A_{a}$. It follows that
\begin{align}
&P(a,b|\alpha,\beta)=\zeta(\tilde{M}^A_{a|\alpha}\tilde{M}^B_{b|\beta}),
\end{align}
for
\begin{equation}
\tilde{M}^A_{a|\alpha}:=\Lambda^{A,\alpha}(\tilde{M}^A_a), \tilde{M}^B_{b|\beta}:=\Lambda^{B,\beta}(\tilde{M}^B_b).
\end{equation}
Since the Kraus operators of $\Lambda^{X,\alpha}$ and the operator $\tilde{M}^X_c$ are localized in $\RR_-^X$, we have that the operators $\tilde{M}^A_{a|\alpha}$, $\tilde{M}^B_{b|\beta}$ are respectively localized in the strictly space-like separated regions $\RR^A_-$, $\RR^B_-$. Hence, by the split property there exist Hilbert spaces $\H_A,\H_B$ and a unitary $U:\H\to \H_A\otimes\H_B$ such that
\begin{equation}
U\tilde{M}^A_{a|\alpha}U^*=M^A_{a|\alpha}\otimes \id_B,U\tilde{M}^B_{b|\beta}U^*=\id_A\otimes M^B_{b|\beta}.
\end{equation}
Defining $\zeta'(\bullet)=\zeta(U^*\bullet U)$, we thus have that
\begin{equation}
P(a,b|\alpha,\beta)=\zeta'\left(M^A_{a|\alpha}\otimes M^B_{b|\beta}\right),  
\end{equation}
which implies that $P(a,b|\alpha,\beta)\in C_{\mathrm{qs}}$. 

We have just shown that, for $\Omega$ admitting a realization through compositions of FV schemes with QFTs satisfying the split property, it is the case that any distribution obtained in a Bell test satisfies $P(a,b|\alpha,\beta)\in C_{\mathrm{qs}}$. For $\Omega$ implementable and $M^A, M^B$ having POVM elements in $\pi(\hat{\W})$, there exists a sequence $(\Omega^j)_j$ of such channels, with $\lim_{j\to\infty}\Omega^j=\Omega$, and so the obtained distribution $P$ satisfies $P=\lim_{j\to\infty}P^j$. Since $\overline{C_{\mathrm{qs}}}=C_{\mathrm{qa}}$, it must be the case that $P(a,b|\alpha,\beta)\in C_{\mathrm{qa}}$.

\section{$\epsilon$-implementability and nonlocality: proof of Proposition~\ref{prop:epsilon_imp2quan}}
\label{app:epsilon_imp}
Let $\tilde{P}\in C_{\mathrm{qa}}$ be the minimizer of eq. (\ref{min_dist_as}), and consider the implementable, FV-friendly box channel $\tilde{\Omega}$ with distribution $\tilde{P}$. Let $\omega, O$ respectively denote a QFT state and an operator of norm smaller than or equal to $1$. Then we have that
\begin{align}
&\omega\circ(\Omega-\tilde{\Omega})(O)\nonumber\\
&=\omega\circ\sum_{a,b,\alpha,\beta}(P(a,b|\alpha,\beta)-\tilde{P}(a,b|\alpha,\beta))\Omega^A_\alpha\circ\Omega^B_\beta\circ D^{a g_A+bg_B}(O)\nonumber\\
&=\sum_{\alpha,\beta}Q(\alpha,\beta)\Delta P(a,b|\alpha,\beta)\omega_{\alpha,\beta}(O_{a,b}),
\end{align}
with
\begin{align}
&\Delta P(a,b|\alpha,\beta):=P(a,b|\alpha,\beta)-\tilde{P}(a,b|\alpha,\beta),\nonumber\\
&Q(\alpha,\beta):=\omega\circ \Omega^A_\alpha\circ\Omega^B_\beta(1),\nonumber\\
&\omega_{\alpha,\beta}(\bullet):=\frac{1}{Q(\alpha,\beta)}\omega\circ \Omega^A_\alpha\circ\Omega_\beta^B(\bullet),\nonumber\\
&O_{a,b}:=D^{a g_A+bg_B}(O)
\end{align}
Note that $Q(\alpha,\beta)$ is a distribution, that $\{\omega_{\alpha,\beta}\}_{\alpha,\beta}$ are normalized states and that $\|O_{ab}\|\leq 1,\forall a,b$. It follows that
\begin{align}
&|\omega\circ(\Omega-\tilde{\Omega})(O)|\leq \sum_{\alpha,\beta,a,b}|\Delta P(a,b|\alpha,\beta)|Q(\alpha,\beta)\nonumber\\
&\leq \max_{\alpha,\beta}\sum_{a,b}|\Delta P(a,b|\alpha,\beta)|=\|P-\tilde{P}\|=\mu.
\end{align}
Hence, $\Omega$ is $\mu$-implementable.

Next, suppose that $\Omega$ is $\epsilon$-implementable. Then, there exists an implementable channel $\tilde{\Omega}$ satisfying eq. (\ref{approx_O}) for all $O\in\pi(\hat{\W})$, $\|O\|\leq 1$. Now we apply Proposition \ref{prop:recovery_P} to $\Omega$: there exists a family of states $\omega_\lambda$, Weyl POVMs $M^A, M^B$, respectively localized in $\RR_+^A,\RR^B_+$, and displacement channels $\{D^{A,\alpha}\}_{\alpha}$, $\{D^{B,\beta}\}_{\beta}$, respectively localized in $\U^A, \U^B$, such that
\begin{equation}
P(a,b|\alpha,\beta)= \lim_{\lambda\to 0}\omega_\lambda\circ D^{A,\alpha}\circ D^{B,\beta}\circ \Omega(M^A_aM^B_b).
\end{equation}    
Define then
\begin{equation}
P_\lambda(a,b|\alpha,\beta):=\omega_\lambda\circ D^{A,\alpha}\circ D^{B,\beta}\circ \Omega(M^A_aM^B_b)
\end{equation}
and 
\begin{equation}
\tilde{P}_\lambda(a,b|\alpha,\beta):=\omega_\lambda\circ D^{A,\alpha}\circ D^{B,\beta}\circ \tilde{\Omega}(M^A_aM^B_b).
\end{equation}
By Proposition \ref{prop:FV_wirings_box}, we have that $\tilde{P}_\lambda(a,b|\alpha,\beta)\in C_{\mathrm{qa}}$. Moreover, from the definition of $\epsilon$-implementability it follows that
\begin{equation}
|P_\lambda(a,b|\alpha,\beta)-\tilde{P}_\lambda(a,b|\alpha,\beta)|=\left|\omega_\lambda\circ(\Omega-\tilde{\Omega})(M^A_{a|\alpha}M^B_{b|\beta})\right|\leq \epsilon.
\end{equation}
This implies that
\begin{equation}
\|P_\lambda-\tilde{P}_\lambda\|= \max_{\alpha,\beta}\sum_{a,b}|P_\lambda(a,b|\alpha,\beta)-\tilde{P}_\lambda(a,b|\alpha,\beta)|\leq \epsilon M^2,
\end{equation}
for all $\lambda$. Taking the limit $\lambda\to\infty$, we have that $\mu$ is bounded by the right-hand side of the equation above.

\section{Perturbations of quantum boxes: proof of Lemma \ref{lemma:non_loc}}
\label{app:non_loc}
To prove the lemma, we define the set of local correlations $C_l$ as the set of boxes $P$ admitting a local hidden variable model:
\begin{equation}
P(a,b|x,y)=\sum_\lambda p_\lambda P_A(a|x,\lambda)P_B(b|y,\lambda),
\end{equation}
for some probability distributions $p_\lambda,P_A, P_B$. This set is known to be a closed polytope, with $C_l\subset C_{\mathrm{qs}}$ \cite{rev_nl2}. Moreover, it is the case that $\mbox{span}(C_l)=\mbox{span}(C_{\mathrm{ns}})=:\mathbb{S}(M,N)$ \cite{neg_probs}. 

In any Bell scenario $(M,N)$, the local polytope is defined as:
\begin{equation}
\mathbb{S}(M,N)\cap \{P:\sum_{a,b}P(a,b|x,y)=1\}\cap\bigcap_{j=1}^n\{P:B_j(P)\leq K_j\},
\end{equation}
for some computable linear functionals $\{B_j\}_j$ with rational coefficients and rational numbers $\{K_j\}$, satisfying
\begin{equation}
\mbox{dim}\left(C_l\cap \{P:B_j(P)\leq K_j\}\right)=\mbox{dim}(C_l)-1.
\end{equation}
They are called facet inequalities \cite{rev_nl2}, and the condition above precludes the possibility that one of them is saturated by all extreme points of $C_l$. This implies that $P^I$ does not saturate any facet inequality, since $P^I$ can be expressed as a strict convex combination of all extreme points of $C_l$.

Now, consider the effect of perturbing $P^I$ by $\Delta P$, with
\begin{equation}
\Delta P\in \mathbb{S}(M,N),\sum_{a,b}\Delta P(a,b|\alpha,\beta)=0.
\label{delta_constr}
\end{equation}
Note that $P^I+\Delta P$ is normalized and satisfies $P^I+\Delta P\in \mathbb{S}(M,N)$. We are interested in finding the maximum real number $R(M,N)>0$ such that, for any $\Delta P(a,b|\alpha,\beta)$ with $\|\Delta P(a,b|\alpha,\beta)\|\leq R(M,N)$, it is the case that
\begin{equation}
P^I+\Delta P\in C_l.    
\end{equation}
From the above, it suffices to compute
\begin{align}
R_j(M,N):=&\min \|\Delta P\|,\nonumber\\
\mbox{such that }&B_j(P^I+\Delta P)=K_j,\nonumber\\
&\Delta P\in \mathbb{S}(M,N),\sum_{a,b}\Delta P(a,b|0,0)=1,
\end{align}
for $j\in\{1,...,n\}$: indeed, $R(M,N)$ corresponds to the minimum of $\{R_j(M,N)\}_j$. For each $j$, the calculation of $R_j(M,N)$ can be formulated as a linear program with rational entries \cite{linear_programming}, which implies that $R(M,N)\in \mathbb{Q}$ and that $R(M,N)$ is computable for any Bell scenario $(M,N)$.

We next argue that, for any box $P_\mu=(1-\mu) P+\mu P^I$, with $P\in C_{\mathrm{qa}}$, eq. (\ref{perturb_local}) holds. Let $\tilde{P}\in C_{\mathrm{ns}}$ satisfy $\|P_\mu-\tilde{P}\|\leq \mu R(M,N)$. Then we have that
\begin{align}
&\tilde{P}=P_\mu +\tilde{P}-P_\mu=(1-\mu)P+\mu\left(P^I+\frac{\tilde{P}-P_\mu}{\mu}\right).
\end{align}
Note that $\Delta P:=\frac{\tilde{P}-P_\mu}{\mu}$ satisfies the constraints (\ref{delta_constr}) and that $\|\Delta P\|\leq R(M,N)$: it follows that $P^I+\Delta P\in C_l$. The box $\tilde{P}$ is therefore a convex combination of $P\in C_{\mathrm{qa}}$ and $P^I+\Delta P\in C_l\subset C_{\mathrm{qa}}$. Since $C_{\mathrm{qa}}$ is convex, we arrive at $\tilde{P}\in C_{\mathrm{qa}}$.
\end{appendix}

\end{document}